%% file: main.tex
\documentclass{article}
\input{macros.tex}

\usepackage{amsmath}

\title{Work-Efficient Parallel Derandomization II: \\
Optimal Concentrations via Bootstrapping}

\begin{document}
\date{}
 \author{Mohsen Ghaffari \\ \small MIT \\ \small ghaffari@mit.edu \and Christoph Grunau \\ \small ETH Zurich \\ \small cgrunau@ethz.ch }
\maketitle

\begin{abstract}
In this paper, we present an efficient parallel derandomization method for randomized algorithms that rely on concentrations such as the Chernoff bound. This settles a classic problem in parallel derandomization, which dates back to the 1980s. 

Concretely, consider the \textit{set balancing} problem where $m$ sets of size at most $s$ are given in a ground set of size $n$, and we should partition the ground set into two parts such that each set is split evenly up to a small additive (discrepancy) bound. A random partition achieves a discrepancy of $O(\sqrt{s \log m})$ in each set, by Chernoff bound. We give a deterministic parallel algorithm that matches this bound, using near-linear work $\tilde{O}(m+n+\sum_{i=1}^{m} |S_i|)$ and polylogarithmic depth $\poly(\log(mn))$. The previous results were weaker in discrepancy and/or work bounds: Motwani, Naor, and Naor [FOCS'89] and Berger and Rompel [FOCS'89] achieve discrepancy $s^{\eps} \cdot O(\sqrt{s \log m})$ with work $\tilde{O}(m+n+\sum_{i=1}^{m} |S_i|) \cdot m^{\Theta(1/\eps)}$ and polylogarithmic depth; the discrepancy was optimized to $O(\sqrt{s \log m})$ in later work, e.g. by Harris [Algorithmica'19], but the work bound remained prohibitively high at $\tilde{O}(m^4n^3)$. Notice that these would require a large polynomial number of processors to even match the near-linear runtime of the sequential algorithm. Ghaffari, Grunau, and Rozhon [FOCS'23] achieve discrepancy $s/\poly(\log(nm)) + O(\sqrt{s \log m})$ with near-linear work and polylogarithmic-depth. Notice that this discrepancy is nearly quadratically larger than the desired bound and barely sublinear with respect to the trivial bound of $s$. 

Our method is different from prior work. It can be viewed as a novel bootstrapping mechanism that uses crude partitioning algorithms as a subroutine and sharpens their discrepancy to the optimal bound. In particular, we solve the problem recursively, by using the crude partition in each iteration to split the variables into many smaller parts, and then we find a constraint for the variables in each part such that we reduce the overall number of variables in the problem. The scheme relies crucially on an interesting application of the multiplicative weights update method to control the variance losses in each iteration.

Our result applies to the much more general \textit{lattice approximation} problem, thus providing an efficient parallel derandomization of the randomized rounding scheme for linear programs. 
\end{abstract}

 {   
     \thispagestyle{empty}
 }

{   \newpage
    \smallskip
    \hypersetup{linkcolor=blue}
    \tableofcontents
    \setcounter{page}{0}
    \thispagestyle{empty}
}

\newpage
\setcounter{page}{1}

\section{Introduction}
This paper presents an efficient parallel method for derandomizing randomized algorithms that rely on concentrations of measure such as Chernoff and Hoeffding bounds. This settles one of the classic and central questions in parallel derandomization~\cite{motwani1989probabilistic, berger1989simulating}. 

Let us start with a concrete and simple-to-state problem, known as the \textit{set balancing} or \textit{set discrepancy} problem, which has been used as the primary target in this line of work\footnote{This problem and similar others arise naturally and frequently in algorithm design. Here is a simple example: suppose we want to partition the edges of a graph $G=(V, E)$ into $k$ parts, such that each node has almost equal number of edges in different parts. Such a partition is useful, e.g., when computing an edge-coloring as different parts can be colored independently.}. By known reductions~\cite{motwani1989probabilistic}, our result applies to much more general problems such as \textit{lattice approximation}~\cite{raghavan1986probabilistic}. 

\medskip
\paragraph{Set balancing} Consider $m$ subsets $S_{1}, S_2, \dots, S_m \subseteq [n]$ in a ground set of $n$ elements, where $|S_i|\leq s$ for each $i\in [m]$. We want to split the ground set into two parts such that each of the subsets is split as evenly as possible. Concretely, we want a vector $\chi \in \{\-1,1\}^n$ that minimizes the \textit{set system's discrepancy} defined as $\max_{i\in [m]} disc(S_i)$ where the discrepancy of set $S_i$ is defined as $disc(S_i):=|\sum_{j\in S_i} \chi_j|$.
For a random $\chi \in \{\-1,1\}^n$, the Chernoff bound implies a discrepancy of $O(\sqrt{s\log m})$, via a union bound over all $m$ sets.\footnote{A beautiful result of Spencer~\cite{spencer1985six} gives a discrepancy bound of $O(\sqrt{n(1+\log(m/n))})$, and algorithmic variants were provided in recent breakthroughs ~\cite{bansal2010constructive, lovett2012constructive}. However, our focus here is on bounds that have $\sqrt{s}$ as the leading factor, instead of $\sqrt{n}$, because this is frequently needed in algorithmics, e.g., in generalization to lattice approximation. Moreover, the difference between the second factors $\sqrt{\log(2m/n)}$ and $\sqrt{\log m}$ is considerable only in the special case where $m\ll n^{1.001}$.  } 

Results of Spencer \cite{spencer1977balancing} and Raghavan \cite{raghavan1986probabilistic} present deterministic algorithms that achieve the same bound, via the method of conditional expectation. But these algorithms are inherently sequential. Our objective is to achieve a similar result via a deterministic parallel algorithm. We next review the state of the art, after a quick recap on the parallel terminology.

\medskip
\paragraph{Parallel model and terminology: depth, work, and work-efficiency} We follow the standard \textit{work-depth} model~\cite{jaja1992introduction, blelloch1996programming}, where the algorithm runs on $p$ processors with access to a shared memory. In any algorithm $\mathcal{A}$, its depth $D(\mathcal{A})$ is the longest chain of computational steps in $\mathcal{A}$ each of which depends on the previous ones. In other words, this is the time that it would take the algorithm to run even if we were given an infinite number of processors. The work $W(\mathcal{A})$ is the total number of the computational steps in $\mathcal{A}$. Given $p$ processors, the algorithm clearly needs $\max\{D(\mathcal{A}), W(\mathcal{A})/p\}$ time.
By Brent's principle~\cite{brent1974parallel}, we can run the algorithm in $D(\mathcal{A})+ W(\mathcal{A})/p$ time using $p$ processors. The objective in parallel computations is to devise algorithms that run faster than their sequential counterpart, ideally with a speed-up proportional to the number of processors $p$. In particular, this requires the parallel algorithm to have a work bound (asymptotically) equal to the best known sequential algorithm, in which case we call the algorithm \textit{work-efficient}. There have been a number of exciting recent work on achieving work-efficient parallel algorithms (or \textit{nearly work-efficient} algorithms where the work bound is relaxed by a polylogarithmic factor) for various problems; see, e.g.,  \cite{fineman2018nearly, jambulapati2019parallel, blelloch2020parallelism, li2020faster, andoni2020parallel, cao2020efficient,dhulipala2021theoretically,anderson2021parallel,rozhovn2022undirected,rozhovn2022deterministic, GGR2023Chernoff}.

\subsection{State of the art}
For reference, let us note that the sequential deterministic algorithms of \cite{spencer1977balancing, raghavan1986probabilistic} achieve the desired discrepancy bound of $O(\sqrt{s\log m})$ in $\tilde{O}(n+m+\sum_{i=1}^{m} |S_i|)$ time, which is near-linear in the input size. Thus, the ideal parallel algorithm is one with the same near-linear work bound and polylogarithmic depth, achieving the same discrepancy. The state-of-the-art deterministic parallel algorithms achieve weaker results: 

\begin{itemize}
\item[(1)] Motwani, Naor, and Naor~\cite{motwani1989probabilistic}, and Berger and Rompel~\cite{berger1989simulating}, presented parallel deterministic algorithms with $\poly(\log(mn))$ depth that achieve discrepancy $s^{\eps} \cdot O(\sqrt{s\log m})$, for $\eps\in (0, 0.5]$, using $\tilde{O}(n+m+\sum_{i=1}^{m} |S_i|) \cdot m^{\Theta(1/\eps)}$ work. Motwani, Naor, and Naor~\cite{motwani1989probabilistic} also present an improvement, which limits the work to a bound of roughly $\tilde{O}(n+m+\sum_{i=1}^{m} |S_i|) \cdot m^{4}$, though at the expense of increasing the depth to $\log^{\Theta(1/\eps+1)}(mn)$. Via a parallel randomized automata fooling method of Karger and Koller~\cite{karger1994NC}, Mahajan, Ramos, and Subrahmanyam~\cite{mahajan2001solving} optimized the discrepancy bound to $O(\sqrt{s\log m})$, using a high work bound of $\tilde{O}(m^{10} n^7)$, which was later improved by Harris~\cite{harris2019deterministic} to $\tilde{O}(m^4 n^3)$. All these algorithms are quite far from work efficiency, and require a large polynomial number of processors to even match the speed of sequential algorithms---a prohibitively high requirement. 

\item[(2)] Recently, Ghaffari, Grunau, and Rozhon~\cite{GGR2023Chernoff} gave a parallel deterministic algorithm  using $\tilde{O}(n+m+\sum_{i=1}^{m} |S_i|)$ work and $\poly(\log(mn))$ depth that achieves a discrepancy of $s/\poly(\log(mn))+O(\sqrt{s\log m})$. That is they split the ground set into two parts such that for each of the $m$ subsets of size at most $s$, the intersection with each part has size at most $s(1/2+1/\poly(\log(mn))$, for $s=\Omega(\poly(\log(mn))$. However, notice that the achieved discrepancy is almost quadratically higher than the desired bound. 
\end{itemize}

\subsection{Our results} In this paper, we present the first work-efficient deterministic parallel algorithm that achieves the optimal discrepancy in polylogarithmic depth, thus essentially settling the above question.

\begin{restatable}{theorem}{thmMainUnweighted}
\label{thm:main-unweighted}
Consider $m \geq 2$ subsets $S_{1}, S_2, \dots, S_m \subseteq [n]$ and suppose $|S_i|\leq s$ for each $i\in [m].$ There is a deterministic parallel algorithm, with $\tilde{O}(n+m+\sum_{i=1}^{m} |S_i|)$ work and $\poly(\log(mn))$ depth, that computes a vector $\chi \in \{\-1,1\}^n$ such that, for each $i\in [m]$, we have $disc(S_i)=|\sum_{j\in S_i} \chi_j| =O(\sqrt{s\log m})$.
\end{restatable}

\medskip

\paragraph{Weighted set balancing} Our method, with some extra work, applies to the more general weighted discrepancy problem defined as follows:  Consider an $m\times n$ matrix $A$, and $m$ subsets $S_{1}, S_2, \dots, S_m \subseteq [n]$. We want a vector $\chi \in \{\-1,1\}^n$ that gives a small weighted discrepancy $disc(i):=|\sum_{j=1}^{n} a_{ij} \chi_j|$ for each set $i\in[m]$. For a random $\chi$, Hoeffding's bound implies $disc(i) = O(\sqrt{\sum_{j =1}^n a_{ij}^2 \cdot \log m})$, for each set $S_i$, with high probability. We give a work-efficient deterministic parallel algorithm that gets the same guarantee.

\begin{restatable}{theorem}{thmMainWeighted}
\label{thm:main-weighted}
Let $n,m\in \mathbb{N}$ with $m \geq 2$, and $A \in \mathbb{R}^{m \times n}$. There is a deterministic parallel algorithm algorithm with work $\tilde{O}(nnz(A) + n + m)$ and depth $\poly(\log(mn))$ that computes a vector $\chi \in \{-1,1\}^n$ such that, for every $i \in [m]$, it holds that $disc(i) = |\sum_{j=1}^{n} a_{ij} \chi_j| = O(\sqrt{\sum_{j =1}^n a_{ij}^2 \cdot \log m})$.
\end{restatable}

As a very special case, this implies in the unweighted case that each set $S_{i}$ has discrepancy $O(\sqrt{|S_i| \log m})$. With this weighted set balancing algorithm, via a reduction of \cite{motwani1989probabilistic}, our result generalizes much further, to a problem known as \textit{lattice approximation}~\cite{raghavan1986probabilistic,motwani1989probabilistic}. Informally, this is the problem of finding an integral point that approximates a fractional solution, for a number of given linear constraints.

\bigskip
\paragraph{Lattice approximation} Suppose we are given an $m\times n$ matrix $A$, with each entry $a_{ij}\in [0,1]$, as well as a vector $\mathbf{p}\in [0,1]^{n}$. The \textit{lattice approximation} problem asks for a vector $\mathbf{q}\in \{0,1\}^n$ with a small bound on $|\sum_{j=1}^{n} a_{ij} q_j - \sum_{j=1}^{n} a_{ij} p_j|$ for each $i\in [m]$.

Notice that if we set $\mathbf{q}$ randomly where $Pr[q_j=1]=p_j$ for each $j\in [n]$, by Chernoff bound we get a solution such that, with high probability, for each $i\in [m]$, we have $|\sum_{j=1}^{n} a_{ij} q_j - \sum_{j=1}^{n} a_{ij} p_j| \leq O(\sqrt{\mu_i\log m}+\log m)$, where $\mu_i = \sum_{j=1}^{n} a_{ij} p_j$. This is what Raghavan called \textit{randomized rounding for linear programs}~\cite{raghavan1986probabilistic}. Motwani, Naor, and Naor~\cite{motwani1989probabilistic} provided a parallel algorithm that achieves discrepancy $O(\mu^{\eps}\cdot \sqrt{\mu_i\log m}+\log^{1/(1-2\eps)} m)$ for each $i\in [m]$, using a large polynomial work, with $\Theta(1/\eps)$ in the exponent of the polynomial. The discrepancy was improved later to the ideal bound ~\cite{mahajan2001solving, harris2019deterministic}, but using work bounds that are large polynomials and are thus far from work efficiency. The method of Ghaffari, Grunau, and Rozhon~\cite{GGR2023Chernoff} gives a work-efficient algorithm, but achieving only a slightly sublinear discrepancy of $\mu_i/\poly(\log(mn))$ when the expectation $\mu_i$ is lower bounded by a sufficiently high $\poly(\log(mn))$. Our result stated below provides an efficient deterministic parallel algorithm that achieves the same result as the Chernoff bound. It thus can be viewed as an efficient parallel derandomization of randomized rounding for linear programs. 
\begin{restatable}{theorem}{thmLattice}
\label{thm:main-lattice}
    Suppose we are given an $m\times n$ matrix $A$, with each entry $a_{ij}\in [0,1]$, as well as a vector $\mathbf{p}\in [0,1]^{n}$. There is a deterministic parallel algorithm that computes a vector $\mathbf{q}\in \{0,1\}^n$ such that, for each $i\in [m]$, we have $|\sum_{j=1}^{n} a_{ij} q_j - \sum_{j=1}^{n} a_{ij} p_j| \leq O(\sqrt{\mu_i\log m} + \log m)$, where $\mu_i = \sum_{j=1}^{n} a_{ij} p_j$. The algorithm has $\poly(\log(mn))$ depth and $\tilde{O}(n+m+nnz(A))$ work, where $nnz(A)$ denotes the number of nonzero entries in the matrix $A$.
\end{restatable}
This result follows from \Cref{thm:main-weighted} along with a reduction of Motwani, Naor, Naor detailed in their journal version \cite[Sections 8]{motwani1994probabilisticJournal}, for which we provide a brief sketch in \Cref{app:lattice}.

\paragraph{Example application---edge coloring}
By Vizing's theorem, any graph with maximum degree $\Delta$ admits an edge coloring with at most $\Delta+1$ colors. It is not known how to achieve this result using a $\poly(\log n)$ depth parallel algorithm (in general graphs; the bipartite case is easy and known even with $\Delta$ colors~\cite{lev1981fast}.) However there are algorithms that come close to this bound, and our focus is on deterministic parallel algorithms. Motwani, Naor, and Naor~\cite{motwani1989probabilistic} presented a $\poly(\log n)$-depth deterministic parallel algorithm with $\Delta + \Delta^\eps \cdot O(\sqrt{\Delta \log n})$ colors, and using work $m^{\Theta(1/\eps)}$ work. This was essentially a direct application of the set balancing problem, by following the edge coloring framework of Karloff and Shmoys~\cite{karloff1987edge-coloring}. The improvements of Mahajan et al~\cite{mahajan2001solving} and Harris~\cite{harris2019deterministic} improved the number of colors to $\Delta + O(\sqrt{\Delta \log n})$, but using prohibitively large polynomial work bounds. 
Plugging our set balancing algorithm instead, we get a deterministic parallel edge-coloring algorithm with the number of colors improved to $\Delta + O(\sqrt{\Delta \log n})$, with near-linear work. A proof sketch is presented in \Cref{app:edge-coloring}.

\begin{restatable}{corollary}{CrlEdgeColoring}
\label{crl:edge-coloring}
    There is a deterministic parallel algorithm that, given an undirected graph $G=(V, E)$ with maximum degree $\Delta$, computes an edge-coloring of it with $\Delta + O(\sqrt{\Delta \log n})$ colors, using $\poly(\log n)$ depth and $\tilde{O}(m)$ work, where $n=|V|$ and $m=|E|$.
\end{restatable}
\subsection{Method Overview}
Our method is different than those of the prior work~\cite{motwani1989probabilistic, berger1989simulating, mahajan2001solving, harris2019deterministic,GGR2023Chernoff}. In short, the method of ~\cite{motwani1989probabilistic, berger1989simulating} works via an efficient binary search in a $k$-wise independent space for $k=\log m/(\eps \log s)$, trying to find a point in this space that satisfies all the discrepancy constraints. Methods of \cite{mahajan2001solving, harris2019deterministic}, which in part rely on those of \cite{karger1994NC}, work via a scheme of building pseudorandom spaces that fool certain automata, and they seem to inherently require a large polynomial work. Finally, the method of \cite{GGR2023Chernoff} determines the $n$ random variables as a result of a $\poly(\log (mn))$-iteration random walk, with merely pairwise independence in each iteration, and derandomizes each iteration work-efficiently by maintaining certain pairwise objectives. 

The method we present in this paper can be viewed as an orthogonal approach, which makes use of algorithms with weaker discrepancy bounds as a basic partitioning tool, and sharpens their discrepancy via bootstrapping. Indeed, our method can be applied on top of the previous methods of \cite{motwani1989probabilistic, berger1989simulating,GGR2023Chernoff} to improve their discrepancy to the optimal bound (though with different work bounds). We use in particular the latter result since this yields a work-efficient algorithm overall. Below, we provide a high-level and intuitive (though, admittedly imprecise) overview of our approach. The actual technical proofs will deviate from this outline in some parts, due to details not discussed here.

\medskip
\paragraph{High-level approach} On a high level, our general approach is to partition the $n$ random variables $\chi_{j}\in\{-1, 1\}$ for $j\in[n]$ into $L\leq n/2$ parts $P_{1}\sqcup \dots \sqcup P_{L}$ and constrain the variables in each part to a one-dimensional space. That is, for each part $P_{t}$, we compute a \textit{tentative} $\chi_{j}$ for $j\in P_{t}$, but we allow that potentially all of these variables in part $P_t$ can be negated. More concretely, for each $j, j'\in P_{t}$, their product $\chi_j\chi_{j'}$ in the final $\chi$ is fixed, so once we know the final $\chi_j$, we know $\chi_{j'}$ for all $j'\in P_{t}$. Then the remaining problem will be to determine whether each part's tentative assignment should be negated or not. This is like computing a vector in $\{-1, 1\}^{L}$, indicating the negations, such that we satisfy some linear constraints. We solve that recursively as another (\textit{weighted}) instance of the set balancing problem, with $L\leq n/2$ variables. In the base case, once the number of variables is as small as $\poly(\log(nm))$, we can fix them one by one via conditional expectations~\cite{raghavan1986probabilistic}. 

The crucial point is how to perform the recursive scheme. Note that the desired output is that each set's discrepancy is within an $O(\sqrt{\log m})$ of its standard deviation upper bound $\sqrt{s}$. Also, this is roughly the best one can hope for, given that we want to achieve this for $m$ sets simultaneously. Thus, we need to perform the recursion such that, in each recursion level, we do not increase the \textit{standard deviation} of each set (by more than a $1+1/\log n$ factor, as we will have roughly $\log n$ recursion levels). Furthermore, some care and much extra work will be needed due to the \textit{weightedness} introduced in the recursion, even if we start with an unweighted instance; we will ignore that for now for this brief overview.

\medskip
\paragraph{Each recursion level, and how to fix the variables in each part}
Let us discuss how we fix the variables inside each part. We later discuss how the partition is computed, once we understand the properties needed from the partition. Consider a fixed part $P_{t}$, and a set $S_{i}$ for a fixed $i\in [m]$. Consider a random $\chi\in\{-1,1\}^n$, let us define its signed discrepancy $sdisc({S_i)}=\sum_{j\in S_i} \chi_j$, in contrast to the absolute discrepancy $disc(S_i)=|\sum_{j\in S_i} \chi_j|$. We have $\E[sdisc(S_i)]=0$. Instead of working with the \textit{standard deviation}, we zoom in on the variance of $sdisc(S_i)$, due to its nice additiveness properties. We have $Var(sdisc(S_i)) = \E[sdisc^2(S_i)] - (\E[sdisc^2(S_i)])^2 = \E[disc^2(S_i)].$
In particular, notice that $\E[disc^2(S_i)]=\E[(\sum_{j\in S_i} \chi_j)^2] = \E[\sum_{j\in S_i} (\chi_j)^2]=|S_i|\leq s$. Initially, the variables in $S_i\cap P_{t}$ contribute exactly $|S_i \cap P_{t}|$ to $Var(sdisc(S_i))= \E[disc^2(S_i)]$. After constraining these variables, their contribution to $\E[disc^2(S_i)]$ will be exactly equal to $(\sum_{j\in S_i\cap P_t} \chi_j)^2$, as it will be independent of the contributions of other parts. Thus, we want to determine $\chi_j$ for $j\in P_{t}$ such that $(\sum_{j\in S_i\cap P_t} \chi_j)^2$ remains almost the same bound as $|S_i \cap P_{t}|$. 

Unfortunately, that is impossible! This is exactly a set balancing problem for sets $S_1\cap P_t$, $S_2\cap P_t$, \ldots, $S_m\cap P_t \subseteq P_{t}$, and we know that, regardless of how we choose $\chi_j$ for $j\in P_{t}$, some set $S_i$ will have $\frac{(\sum_{j\in S_i\cap P_t} \chi_j)^2}{|S_i\cap P_t|} = \Omega(\log m)$; this is a $\Theta(\log m)$ factor loss in the variance. Thus, it seems we would lose a $\Theta(\log m)$ factor in variance in each recursion level. See the warm-up in \Cref{thm:rootdepth-suboptimal} where we follow this scheme for a two-level recursion, getting an $\tilde{O}(\sqrt{n})$ depth parallel algorithm with a $\sqrt{\log m}$ loss in discrepancy. If we follow this for $\log n$ recursion levels, the variance losses would multiply to $({\log m})^{\log n}$---a useless bound.

\medskip
\paragraph{Enforcing an average-guarantee, and leveraging it via multiplicative weights update}
There is something to be optimistic about: even though the worst variance loss factor among the $m$ sets $S_1\cap P_t$, $S_2\cap P_t$, \ldots, $S_m\cap P_t$ will be $\Theta(\log m)$, the average loss across these sets should be smaller. Indeed, by Chernoff bound, for a random $\chi$, the probability of a $z$ factor loss is $exp(-\Theta(z^2))$. More useful for us, we have $\E[\sum_{i=1}^{m} (\sum_{j\in S_i\cap P_t} \chi_j)^2] = \sum_{i=1}^{m} |S_i\cap P_t|$. We will be able to enforce this (and even a weighted variant of it) as an actual constraint in the derandomization process, in a work-efficient manner by utilizing that it uses only pairwise independence. However, this average across different sets $S_1\cap P_t$, \ldots, $S_m \cap P_t$ in one part $P_t$ is not useful on its own. We want \textit{each} set $S_i$ to have a small discrepancy. So we need a different averaging, for each set $S_i$ across different parts $P_t$. That is, we need each set $S_i$ not to experience too much loss, once we add up the new variances in $S_i\cap P_1$, $S_i\cap P_2$, \dots, $S_i\cap P_L$. That is, we want $\sum_{t=1}^{L}(\sum_{j\in S_i\cap P_t} \chi_j)^2 \approx s$ where the approximation can tolerate a $1+1/\log n$ factor.

For that, we appeal to the Multiplicative Weights Update (MWU) method~\cite{arora2012multiplicative}. Suppose we work through the parts $P_1$, $P_2$, \dots, $P_L$ sequentially in $L$ rounds (the algorithm will be different, as we will discuss, since depth $L$ would be expensive). In each round, each set $S_i$ will have an importance value $imp(i)$, such that sets with larger hitherto variance losses have higher importance. We then enforce an importance-weighted variant of the previous average guarantee across the sets $S_1\cap P_t$, \ldots, $S_m \cap P_t$. This in effect tries to have smaller variance losses in this round for sets $S_i$ that have higher importance. Roughly speaking, this will force an averaging for each set $S_i$ across different parts $P_{t}$. As a result, in the end, the overall variance loss for each set $S_i$ summed up over all its parts  $S_i\cap P_1$, $S_i\cap P_2$, \dots, $S_i\cap P_L$ will be small and we get $\sum_{t=1}^{L}(\sum_{j\in S_i\cap P_t} \chi_j)^2 \leq s (1+1/\log n)$. See the warm-up in \Cref{thm:rootdepth-optimal} where we use this idea to sharpen the discrepancy of the previously discussed $\tilde{O}(\sqrt{n})$-depth algorithm.

Processing $L$ parts sequentially would not yield a polylogarithmic depth (unless $L$ itself is just polylogarithmic, which brings other issues, as each part would have many variables, and they are solved one after the other). Instead of processing one part in each round, we process $L/T$ many of them in one round in parallel and independently. The number of rounds will be set to $T=\poly(\log(mn))$, and this will be sufficient for MWU to bring down the average loss in each set to $1+1/\log n$. Roughly speaking, this is because the worst-case loss in each round is at most an $O(\log m)$ higher than the average loss factor of $1$. 

We will set the parts such that each part induces an easy problem: in one application of the scheme each part will consist of only $\poly(\log(mn))$ variables, and thus we will be able to fix these variables via sequential derandomization~\cite{raghavan1986probabilistic}. In another application, the number of variables in the part will be large but each set will have an intersection of size $\poly(\log(mn))$ with this part, and therefore we will be able to use pairwise parallel derandomization, a la Luby~\cite{luby1988removing}.

\paragraph{Partitioning} For the scheme described above to work, we need that, in each set $S_i$, the different $T=\poly(\log(mn))$ rounds of MWU process portions of $S_i$ that expect roughly the same variance, up to $1+1/\poly(\log n)$ factors. To produce this $\poly(\log(mn))$-way partitioning, we use the work-efficient result of Ghaffari, Grunau, Rozhon~\cite{GGR2023Chernoff}, as the base partitioning tool, in an essentially black-box way. It is crucial that we do not need partitions with optimal or near-optimal additive loss here, and the multiplicative error $1+1/\poly(\log (mn))$ of \cite{GGR2023Chernoff} will be tolerable in our overall scheme\footnote{Let us clarify the distinction between the additive loss and the multiplicative error: Consider for instance their scheme for partitioning into two parts, and $m$ sets of size at most $s$. The ideal additive discrepancy would be $O(\sqrt{s\log m})$. The additive discrepancy of \cite{GGR2023Chernoff} will be up to $s/\poly(\log(mn))$. This means, in each set, each of the two halves will have size in $s/2(1\pm 1/\poly(\log(mn)))$. Thus the multiplicative error is $1+1/\poly(\log(mn))$. In our scheme, this $1/\poly(\log(mn))$ partitioning loss will add up with the $1/\poly(\log(mn))$ loss of the MWU method, and still keep the overall loss in variance per level below a $1+1/\log n$ factor.}. 

Finally, the above outline discusses only the unweighted instance and an intuitive explanation of how we perform one level of recursion. However, even if we start with the unweighted set balancing problem, the recursion creates a weighted problem for the next iteration. Fortunately, in the case of the unweighted set balancing problem, the ``weights'' will remain within a relatively close range, and thus an algorithm along the lines discussed above will still work. We present this in \Cref{sec:unweighted} as the proof of \Cref{thm:main-unweighted}.

For our weighted set balancing result, \Cref{thm:main-weighted}, the partitioning is more complex. We do not provide an overview here but just mention this: in each constraint, a few variables of very large weight cannot be handled as before, and we will need the partitioning to put essentially all of these into separate parts (modulo a loss, that will have to be controlled very tightly). We present the proof of \Cref{thm:main-weighted} in \Cref{sec:weighted}. 

\section{Preliminaries: notations and tools from prior work}
\paragraph{Notations} Throughout, we work with the ground set $[n]$ and usually $m$ sets in it $S_1, S_2, \dots, S_m \subseteq [n]$. For a given \textit{value assignment} vector $\chi\in\{-1,1\}^n$, the discrepancy of each set $S_i$ is defined as $disc(S_i) = |\sum_{j=1}^{n} \chi_j|$. We generally assume that $m\geq 2$, as for $m=1$ the problem is trivial. In the weighted generalization, we are given a matrix $A\in \mathbb{R}^{m\times n}$, where the entry $i,j$ is denoted by $a_{ij}$, and the discrepancy of the $i^{th}$ constraint is defined as $disc(i) = |\sum_{j=1}^{n} a_{ij} \chi_j|$. Sometimes, instead of focusing on the absolute value of the discrepancy, we talk about the signed discrepancy $sdisc(i) = \sum_{j=1}^{n} a_{ij} \chi_j$.

Often, we work with a partition of the ground set into parts $P_1 \sqcup P_2 \sqcup \dots \sqcup P_L = [n]$, and we need to determine a value assignment for the variables in each part. We use the notation $\chi^t\in \{-1, 1\}^{[P_t]}$ to indicate the part of the value assignment in part $P_t$, which is a binary vector of length $|P_t|$, but with the convenient indexing inherited from $[n]$ such that their concatenation forms the output vector $\chi\in\{-1,1\}^n$, that is, $\chi_j=\chi^{t}_j$ for $j\in P_t$. 

\subsection{Crude partitioning}
We use in our work the algorithm of Ghaffari, Grunau, and Rozhon~\cite{GGR2023Chernoff} as a crude partitioning tool. We next state their key result here, as well as two simple generalizations that we derive here with some extra work. The proofs of these extensions are deferred to \Cref{app:prelimAppendix-partition}.

\begin{theorem}[Ghaffari, Grunau, Rozhon~\cite{GGR2023Chernoff}]
\label{thm:FOCS23} 
Let $n,m,k \in \mathbb{N}$ with $m \geq 2$ and let $\{S_1,S_2,\ldots, S_m\}$ be a family of subsets of $[n]$. Then, there exists a deterministic parallel algorithm with work $\tilde{O}(n + m + \sum_{i=1}^m|S_i|)\poly(k)$ and depth $\poly(\log(nm)k)$ that computes a partition $P_1 \sqcup P_2 = [n]$ satisfying that $\max(|S_i \cap P_1|,|S_i \cap P_2|) = |S_i|/2 + O(\sqrt{|S_i| \log m}) + \frac{|S_i|}{k}$ for every $i \in [m]$.
\end{theorem}

In particular, the following states the variant when we partition into $L$ parts, instead of just two parts:
\begin{restatable}{lemma}{UnweightedPartition}[Unweighted multi-way partition]
\label{lemma:partition_unweighted}
Let $n,m,L \in \mathbb{N}$ with $m \geq 2$, $\eps \in (0,0.5]$, and let $\{S_1,S_2,\ldots,S_m\}$ be a family of subsets of $[n]$. Also, let $L$ be a power of two. Then, there exists a deterministic parallel algorithm with work $\tilde{O}(n + m + \sum_{i=1}^m|S_i|) \cdot \poly(1/\eps)$ and depth $\poly(\log(nm)/\eps)$ that computes a partition $P_1 \sqcup P_2 \sqcup \ldots \sqcup P_L = [n]$ satisfying that $|S_i \cap P_\ell| \leq (1+\eps)|S_i|/L + O(\log (m)/\eps^2)$ for every $i \in [m]$ and $\ell \in [L]$. 
\end{restatable}

We also use the following weighted variant, which follows roughly speaking by managing together in each constraint the variables of almost the same weight. 

\begin{restatable}{lemma}{WeightedPartition}[Weighted multi-way partition]
\label{lem:weighted_partition}
Let $n,m,L \in \mathbb{N}$ with $m \geq 2$, $A \in \mathbb{R}^{n \times m}$ and let $\eps \in (0,1]$. Furthermore, assume that $L$ is a power of two.
There exists a deterministic parallel algorithm with work $\tilde{O}(n + m + nnz(A)) \cdot \poly(1/\eps)$ and depth $\poly(\log(nm)/\eps)$ that computes a partition $P_1 \sqcup P_2 \sqcup \ldots \sqcup P_L = [n]$ satisfying the following for every $\ell \in [L]$:
\begin{itemize}
    \item $|P_\ell| \leq (1+\eps) n/L + O(\log(nm)/\eps^2)$,
    \item $|P_\ell \cap \{j \in [n] \colon a_{ij} \neq 0\}| \leq \frac{(1+\eps)}{L}|\{j \in [n] \colon a_{ij} \neq 0\}| + O(\log(nm)/\eps^2)$ for every $i \in [m]$ and
    \item $\sum_{j \in P_\ell} a^2_{ij} \leq \frac{1+\eps}{L} \sum_{j = 1}^n a^2_{ij} + O(\log^2(nm)/\eps^3) (a_{max})^2$ for every $i \in [m]$,
\end{itemize}
where we define $a_{max} = \max_{i \in [m], j \in [n]} |a_{ij}|$.
\end{restatable}

\subsection{Sequential derandomization}
We also make use of the sequential derandomization method of Raghavan~\cite{raghavan1986probabilistic}, which solves the weighted set balancing problem by fixing the random variables one by one, using the method of conditional expectations. The following theorem abstracts this result. We do not provide a proof for this version, but we will later state and use a more general result as \Cref{thm:sequential_derandomization}, and we present a proof for that one in \Cref{app:seqDerandwithAverage}. 

\begin{theorem}[Sequential derandomization]
\label{thm:sequential_derandomization_raw}
There exists an absolute constant $C>0$ for which the following holds. Let $n,m\in \mathbb{N}$, and $A \in \mathbb{R}^{m \times n}$. There exists a deterministic parallel algorithm algorithm with work $\tilde{O}(nnz(A) + n + m)$ and depth $n \poly(\log m)$ that computes a vector $\chi \in \{-1,1\}^n$ such that, for every $i \in [m]$, it holds that $disc_i^2 = C\log m \cdot \sum_{j =1}^n a_{ij}^2$. 
\end{theorem}

\section{Warm-Up}
In this section, we present two warm-up results, and in their context, we discuss two of the ideas that we use in our main results. Suppose we are given $m$ sets $S_1, S_2, \dots, S_m \subseteq [n]$, for $m\geq 2$, such that $|S_i|\leq s$ for each $i\in [m]$. By Chernoff bound, we know that a random value assignment $\chi\in \{-1,1\}^n$ creates a discrepancy of at most $disc(S_i)=|\sum_{j\in S_i} \chi_j| = O(\sqrt{s \log m})$ in each of the $m$ set. In the first result, \Cref{thm:rootdepth-suboptimal}, we show a deterministic parallel algorithm that achieves a slightly suboptimal discrepancy of $O(\sqrt{s} \log m)$ in near-linear work and $\tilde{O}(\sqrt{n})$ depth. In the second result, \Cref{thm:rootdepth-optimal}, we show how to improve the discrepancy to the optimal bound of $O(\sqrt{s \log m})$ while keeping near-linear work and $\tilde{O}(\sqrt{n})$ depth. As stated before, these are warm-up results, presently chiefly as contexts for introducing two of the ideas that we use frequently in our main results. Next, we discuss these two ideas from a high-level and informal viewpoint.

The key idea that we will present in the first result is how to create some parallelism in the task of computing a low-discrepancy value assignment, by partitioning the variables. Roughly speaking, we partition the ground set $[n]$ into about $\sqrt{n}$ parts and we find a value assignment in each part independently and all in parallel. This is such that the discrepancy of each set in each part is ``small". Then, in the end, we need to find a good mixture of the assignments of the different parts. Intuitively, this mixture selection will involve negating the solution coming from some of the parts; we will choose the negated parts carefully to achieve a good discrepancy in the output for all the $m$ sets. 

As we will see, compared to the standard deviation upper bound of $\sqrt{s}$, the above approach loses a $\sqrt{\log m}$ in the discrepancy in the parts and another $\sqrt{\log m}$ in finding a good mixture, thus creating the suboptimal discrepancy of $O(\sqrt{s} \log m)$. To remedy this, in the second result, we create a certain multi-round game for the process of determining the solutions in different parts, and we use an instantiation of the Multiplicative Weights Update method to ``average out" the losses of each set throughout the rounds of this game. As a result, we will be able to remove one of the $\sqrt{\log m}$ factor losses essentially completely.

\subsection{Near-optimal discrepancy with $\tilde{O}(\sqrt{n})$ depth}
\begin{theorem}\label{thm:rootdepth-suboptimal}
Let $n, m \in \mathbb{N}$ with $m \geq 2$, and let $\{S_1, S_2, \ldots, S_m\}$ be a family of subsets of $[n]$. Then, there exists a deterministic parallel algorithm that can compute a vector $\chi \in \{-1, 1\}^n$ with $\tilde{O}(n + m + \sum_{i = 1}^m |S_i|)$ work and $\tilde{O}(\sqrt{n})$ depth such that, for every $i \in [m]$, it holds that $disc^2(S_i)=(\sum_{j \in S_i} \chi_j)^2 = O(|S_i|\log^2 m)$.
\end{theorem}

\begin{proof}
Recall that the sequential derandomization method (\Cref{thm:sequential_derandomization_raw}) provides a method to fix the $n$ variables $\chi_j$ for $j\in [n]$ one by one, in depth $\tilde{O}(n)$, such that we have $(\sum_{j \in S_i} \chi_j)^2 = O(|S_i|\log m)$ for every $i\in [m]$. To reduce this  $\tilde{O}(n)$ depth to $\tilde{O}(\sqrt{n})$, we use the following approach: (1) first we partition the variables into roughly $\sqrt{n}$ parts each with $O(\sqrt{n})$ variables, (2) we perform a sequential derandomization inside each part and all in parallel, and then (3) we determine how to merge the $\sqrt{n}$ parts together via another sequential derandomization. We next make this outline concrete.

(1) We compute a partitioning of the $[n]$ into $L=2^{\lceil\log_2(\sqrt{n})\rceil}$ parts $P_1 \sqcup P_2 \sqcup \ldots \sqcup P_L = [n]$ such that for each $t\in [L]$, we have $|P_t|\leq 2\sqrt{n}$. This can be achieved easily via the unweighted partitioning recalled in \Cref{lemma:partition_unweighted}, in $\tilde{O}(n + m + \sum_{i = 1}^m |S_i|)$ work and $\poly(\log (nm))$ depth. 

(2) As a result of the above partitioning, we have $L$ independent discrepancy problems on disjoint variables: problem $t$ consists of sets $S_1\cap P_t, S_2, \cap P_t, \dots, S_m \cap P_t$. We solve these problems in parallel and independently of each other, each using sequential derandomization. Since each part has at most $2\sqrt{n}$ variables, we can invoke the sequential derandomization method (cf. \Cref{thm:sequential_derandomization_raw}) to solve each part $t$ in $\tilde{O}(\sqrt{n})$ depth, getting a vector $\bar{\chi}^{t}\in \{-1, 1\}^{[P_t]}$ such that for each $i\in [m]$, we have $(\sum_{j\in S_i\cap P_t} \bar{\chi}^t_j)^2= |S_i\cap P_t| \cdot O(\log m)$. This works for all the $L$ parts in parallel, in $\tilde{O}(\sqrt{n})$ depth, and using a total of $\tilde{O}(n + m + \sum_{i = 1}^m |S_i|)$ work. To ensure this work bound, we need a simple clean-up before invoking the sequential randomization in each part $t$: we discard from each part $P_t$ sets $i$ for which $S_i\cap P_t=\emptyset$. 

(3) Finally, we need to determine how to merge these $L\approx \sqrt{n}$ solutions $\bar{\chi}^{t}$ for $t\in [L]$. Let us first discuss the situation from an intuitive viewpoint. Notice that naively taking the output vector $\chi\in\{-1, 1\}^n$ as the ``concatenation" of vectors $\bar{\chi}^{t} \in \{-1, 1\}^{[P_t]}$ computed in different parts can result in a large discrepancy. For instance, in the absence of any further guarantee on how the parts are mixed, for set $S_i$, we can have a discrepancy as large as \[|\sum_{j\in S_i} \chi_j| = |\sum_{t=1}^{L} \sum_{j\in S_i\cap P_t} \chi_j| = \sqrt{|S_i\cap P_1| \cdot O(\log m)} + \sqrt{|S_i\cap P_2| \cdot O(\log m)}+\dots+\sqrt{|S_i\cap P_L| \cdot O(\log m)}.\] This can reach $\sqrt{|S_i| L \cdot O(\log m)}$, which is an $\sqrt{L}$ factor loss compared to the ideal bound of $\sqrt{|S_i| \cdot O(\log m)}$. Intuitively, this loss stems from mixing $L$ variables, each expected to be roughly $\sqrt{|S_i|/L \cdot O(\log m)}$, in an arbitrary way. This arbitrary way allows all of the $L$ terms to contribute positively and thus the absolute values add up. The challenge is that this can happen in any one of the $m$ sets. 

To remedy the above issue, we need to find a \textit{good mixture} of the solutions. More concretely, to obtain the output vector $\chi\in\{-1, 1\}^n$, for each part $t$, we can determine whether to \textit{take $\bar{\chi}^{t}$ as is or to negate it}. The negation of course does not change the absolute value of discrepancy in each part. However, if we choose it wisely, it can help us avoid the unfortunate case of discrepancies of the $L$ different parts adding up in the same direction. 

Formally, we set up a new discrepancy problem with a variable $\chi'\in \{-1, 1\}^{L}$---which determines whether each part is negated or not---with the interpretation that we will set the overall output vector as $\chi_{j} = \bar{\chi}^{t}_j \cdot \chi'_{t}$ where $j\in P_t$. For each set $i\in [m]$, we have 
\[
(\sum_{j\in S_i} \chi_j)^2 =  (\sum_{t=1}^{L} \sum_{j\in S_i\cap P_t} \chi_j)^2 = (\sum_{t=1}^{L} \sum_{j\in S_i\cap P_t} \bar{\chi}^{t}_j \cdot \chi'_{t})^2 = (\sum_{t=1}^{L} \chi'_{t} \cdot (\sum_{j\in S_i\cap P_t} \bar{\chi}^{t}_j))^2
\]
To summarize, we have a new ``weighted" discrepancy problem with an output vector $\chi'\in \{-1, 1\}^{L}$, which consists of $m$ sets where each set $S_i$ has $L$ elements and element $t$ has weight $a_{i,t} = (\sum_{j\in S_i\cap P_t} \bar{\chi}^{t}_j)$ in the discrepancy of set $S_i$. That is, we have $disc(S_i) =  (\sum_{t=1}^{L} a_{i,t} \cdot  \chi'_{t})$. The guarantee provided is that for each $t\in[L]$, we have $a^2_{i,t} = (\sum_{j\in S_i\cap P_t} \bar{\chi}^{t}_j)^2 = |S_i\cap P_t| \cdot O(\log m).$

We can solve this discrepancy problem via sequential derandomization of \Cref{thm:sequential_derandomization_raw} in depth $\tilde{O}(L) = \tilde{O}(\sqrt{n})$. Also, the work is no more than $\tilde{O}(n + m + \sum_{i = 1}^m |S_i|)$. We get a vector $\chi'\in \{-1, 1\}^{L}$ such that  $disc^2(S_i) = (\sum_{t=1}^{L} a_{i,t} \cdot  \chi'_{t})^2 = (\sum_{t=1}^{L} a^2_{i,t} \cdot O(\log m)).$ Hence, overall, we have an output vector $\chi\in\{-1,1\}^n$---with the definition $\chi_{j} = \bar{\chi}^{t}_j \cdot \chi'_{t}$ where $j\in P_t$---such that for each $i\in [m]$, we have
\[disc^2(S_i) = (\sum_{j\in S_i} \chi_j)^2 = \sum_{t=1}^{L} a^2_{i,t} \cdot O(\log m)= \sum_{t=1}^{L} |S_i\cap P_t| \cdot O(\log m) \cdot O(\log m) = |S_i| \cdot O(\log^2 m).\qedhere\]
\end{proof}

\subsection{Optimal discrepancy with $\tilde{O}(\sqrt{n})$ depth}

\paragraph{An intuitive/informal discussion of the suboptimality in \Cref{thm:rootdepth-suboptimal} and how we remedy it} \Cref{thm:rootdepth-suboptimal} has a suboptimal discrepancy of $disc(S_i)=\sqrt{|S_i|} \cdot O(\log m)$, instead of the ideal bound of $disc(S_i)=\sqrt{|S_i|} \cdot O(\sqrt{\log m})$. Intuitively, there are two $O(\sqrt{\log m})$ factors in the achieved discrepancy: one $O(\sqrt{\log m})$ factor from the sequential derandomization in solving each of the $\sqrt{n}$ parts as described in step (2), and another $O(\sqrt{\log m})$ factor from the process of finding a good mixture of the parts as described in part (3). In this subsection, we discuss how to remedy this loss.

To alleviate the problem, we zoom in on the first loss: the $O(\sqrt{\log m})$ loss in solving each of the $\sqrt{n}$ parts in step (2). This loss is optimal in the worst-case sense, meaning that in each part, there will be at least one set that experiences an $O(\sqrt{\log m})$ factor loss. However, there is something to be optimistic about: in each part, most of the sets should not experience such a loss. As a matter of fact, the $\sqrt{\log m}$ factor in the sequential derandomization is to allow a union bound over all the $m$ sets, by reducing the probability of each set breaking the bound to $1/m$. But most sets should have a much smaller loss. 

Indeed, we observe that, with some extra work, we can create a variant of the sequential derandomization result stated in \cref{thm:sequential_derandomization_raw} that ensures the average of the losses to be only a $1+o(1)$ factor. We next state this variant (for the first reading, the reader can think of all $imp(i)$ as equal to $1$; later we will need the importance-weighted generality).

\begin{restatable}{theorem}{SeqDerandWithAverage}[Sequential derandomization, augmented with importance-weighted averaging]
\label{thm:sequential_derandomization}
There exists an absolute constant $C>0$ for which the following holds. Let $n,m\in \mathbb{N}$, $M \in \mathbb{R}$ with $ M\geq \max(m,2)$, $A \in \mathbb{R}^{m \times n}$, and $imp(i) \in \mathbb{R}_{\geq 0}$ for each $i\in [m]$. There exists a deterministic parallel algorithm algorithm with work $\tilde{O}(nnz(A) + n + m))$ and depth $n \poly(\log M)$  that computes a vector $\chi \in \{-1,1\}^n$ such that $\sum_{i=1}^m imp(i) \cdot  disc^2_i \leq \left(1 + \frac{1}{M}\right) \sum_{i=1}^m imp(i) \cdot \big(\sum_{j = 1}^n a^2_{ij}\big)$. Moreover, for every $i \in [m]$, it also holds that $disc_i^2 = C\log M \cdot \sum_{j =1}^n a_{ij}^2$. 
\end{restatable}

The proof of this variant is deferred to \Cref{app:seqDerandwithAverage}. Intuitively, this variant says that the average loss among the sets is ``negligible". However, the average loss among the sets is not directly helpful. We somehow need the overall loss of each of the $m$ sets to be small. That is, we need a mechanism that ensures that the average loss of each set, averaged over all the parts, is small (ideally $1+o(1)$). 

For this mechanism, we appeal to the Multiplicative Weights Update (MWU) method. This method uses varying degrees of importance for different sets based on the losses they have experienced so far, ensuring that sets with large hitherto losses have large importance. \Cref{thm:sequential_derandomization} allows us to make use of this, by enforcing that sets with larger importance experience smaller losses in the next part. To make this MWU averaging work, we need to set up some sequential dependency between the parts (previously, the parts were solved in parallel, independently). Next, we first provide a reminder on MWU, phrased concretely for our usage, and then present the algorithm that formalizes this intuition. A proof of this MWU lemma is presented in \Cref{app:MWU}.

\begin{restatable}{lemma}{MWU}[Multiplicative Weights Update]
\label{lem:MWU} 
Consider a multi-round game with $m$ constraints and an oracle. Each constraint $i\in [m]$ has an importance value $imp(i) \in \mathbb{R}^{+}$, which is initially set $imp^1(i)=1$ and changes over the rounds. In each round $t$, we give the oracle the importance $imp^t(.)$ of the constraints, and the oracle gives back for each constraint a value $gap^t(i) \in [0, W]$, with the guarantee that $\sum_{i=1}^{m} imp^t(i) \cdot gap^t(i) \leq \sum_{i=1}^{m} imp^t(i)$. For any given value $\eps\in [0,1/2]$, there is a way to set the importance values during the game such that, at the end of a game with $T=\Omega(W\log m/\eps^2)$ rounds, for each constraint, we have $(\sum_{t=1}^{T} gap^t(i))/T \leq 1+\eps.$ The rule for updating the importance values is simple: in each round $t$, set $imp^{t+1}(i) = imp^{t}(i) \cdot (1+\eta \cdot gap^t(i))$ where $\eta = {\eps}/{(3W)}$.
\end{restatable}

We are now ready for our second warm-up, which achieves the optimal discrepancy in $\tilde{O}(\sqrt{n})$ depth.

\begin{theorem}\label{thm:rootdepth-optimal}
Let $n, m \in \mathbb{N}$ with $m \geq 2$, and let $\{S_1, S_2, \ldots, S_m\}$ be a family of subsets of $[n]$ such that $|S_i|\leq s$ for all $i\in [m]$. Then, there exists a deterministic parallel algorithm that can compute a vector $\chi \in \{-1, 1\}^n$ with $\tilde{O}(n + m + \sum_{i = 1}^m |S_i|)$ work and $\tilde{O}(\sqrt{n})$ depth such that, for every $i \in [m]$, it holds that $disc^2(S_i)=(\sum_{j \in S_i} \chi_j)^2 = O(s \log m)$.
\end{theorem}
\begin{proof}
Set $\eps=1/(10\log m)$. Also, we assume $n\geq \log^5 m$ (otherwise, we can solve the entire problem in $\poly(\log m)$ depth and near-linear work using sequential derandomization), and that $s\geq \Omega(\log^{10} m)$ (otherwise, the statement of the theorem follows from \Cref{thm:FOCS23}, by setting $k=\sqrt{s}$).

First, we partition the variables into $T=\Theta(\log^5 m)$ parts $P_1 \sqcup P_2 \sqcup \ldots \sqcup P_T = [n]$ with the following guarantee: For each set $i\in [m]$ and each part $t\in [T]$, we should have $|S_i\cap P_t| \leq (1+\eps) (|S_i|/T + \delta)$ for $\delta=O(\log m/\eps^2)$, and moreover, for each $t\in T$, we should have $|P_t| \leq (1+\eps)n/T$. This partition can be done directly via \Cref{lemma:partition_unweighted}, in $\tilde{O}(n+m+\sum_{i=1}^{m})$ work and $\poly(\log(nm))$ depth. We will process these $T=\Theta(\log^5 m)$ parts sequentially, using MWU, in the sense that processing part $t\in T$ will be regarded as round $t$ of the MWU game, as described in \Cref{lem:MWU}.

Let us first discuss what we do in each part. Consider part $t$, which has $|P_t|=\Theta(n/\log^5 m)$ variables, and the sets $S_1\cap P_t$, $S_2\cap P_t$, \dots, $S_m\cap P_t$. To solve this part in $\tilde{O}(\sqrt{n})$ depth, we perform something similar to steps (1) and (2) of the proof of \Cref{thm:rootdepth-suboptimal}. Concretely, first we partition $P_t$ further into $L=2^{\lceil{\log\sqrt{n/\log^5 m}\rceil}}\approx\sqrt{n/\log^5 m}$ \textit{pieces} $P_{t, 1} \sqcup P_{t, 2} \sqcup \ldots \sqcup P_{t, L}$ such that for each $t'\in [L]$, we have $|P_{t, t'}| \leq 2|P_{t}|/L = O(\sqrt{n/\log^5 m})$. This can be done via \Cref{lemma:partition_unweighted}, in $\poly(\log(nm))$ depth and using  $\tilde{O}(n+m+\sum_{i=1}^{m})$ total work (again, we need to perform the simple clean-up of removing from each part sets that have an empty intersection with the part). Then, we use sequential derandomization inside each of the pieces, all independently and in parallel. In particular, by invoking \Cref{thm:sequential_derandomization} in each piece $t'\in [L]$, we get an output $\bar{\chi}^{t, t'}\in \{-1, 1\}^{[P_{t, t'}]}$ with the following two guarantees:
\begin{itemize}
    \item For each set $i$, we have $(\sum_{j\in S_i \cap P_{t, t'}} \bar{\chi}^{t, t'}_j)^2 \leq |S_i \cap P_{t, t'}| \cdot O(\log m)$
    \item $\sum_{i=1}^{m} imp(i) \cdot (\sum_{j\in S_i \cap P_{t, t'}} \bar{\chi}^{t, t'}_j)^2 \leq (1+\eps) \sum_{i=1}^{m} imp(i) \cdot |S_i \cap P_{t, t'}|$
\end{itemize}
Over all the pieces $t'$, which are solved in parallel, the algorithm works in $\tilde{O}(n+m+\sum_{i=1}^{m} |S_i|)$ work and $\tilde{O}(\sqrt{n})$ depth. Also, from the second inequality, we can deduce that
\begin{align*} 
\sum_{i=1}^{m} imp(i) \cdot \bigg(\sum_{t'=1}^{L} (\sum_{j\in S_i \cap P_{t, t'}} \bar{\chi}^{t, t'}_j)^2\bigg) &=
\sum_{t'=1}^{L} \sum_{i=1}^{m} imp(i) \cdot (\sum_{j\in S_i \cap P_{t, t'}} \bar{\chi}^{t, t'}_j)^2\\
&\leq \sum_{t'=1}^{L} (1+\eps) \sum_{i=1}^{m} imp(i) \cdot |S_i \cap P_{t, t'}| \\
&= (1+\eps )\sum_{i=1}^{m} imp(i)  \cdot \bigg(\sum_{t'=1}^{L} |S_i \cap P_{t, t'}|\bigg)\\
&= (1+\eps) \sum_{i=1}^{m} imp(i)  \cdot |S_i \cap P_{t}|\\
&\leq (1+\eps) \sum_{i=1}^{m} imp(i) \cdot \big((1+\eps)(\frac{s}{T} + \delta)\big)\\
&= (1+\eps)^2(\frac{s}{T} + \delta) \cdot \sum_{i=1}^{m} imp(i)
\end{align*}
Now, let us define 
\[gap^{t}(i) =  \frac{ \sum_{t'=1}^{L} (\sum_{j\in S_i \cap P_{t, t'}} \bar{\chi}^{t, t'}_j)^2}{(1+\eps)^2(\frac{s}{T} + \delta)}.\] 
Notice that 
\begin{align*} gap^t(i) &= \frac{ \sum_{t'=1}^{L} (\sum_{j\in S_i \cap P_{t, t'}} \bar{\chi}^{t, t'}_j)^2}{(1+\eps)^2(\frac{s}{T} + \delta)} \\
&\leq \frac{ \sum_{t'=1}^{L} |S_i \cap P_{t, t'}| \cdot O(\log m)}{(1+\eps)^2(\frac{s}{T} + \delta)} \\
&= \frac{ |S_i\cap P_{t}| \cdot O(\log m)}{(1+\eps)^2(\frac{s}{T} + \delta)} \leq \frac{O(\log m)}{(1+\eps)^2},
\end{align*} 
and moreover,
\begin{align*} 
\sum_{i=1}^{m} imp(i)\cdot gap^{t}(i) &= \sum_{i=1}^{m} imp(i) \cdot \frac{\sum_{t'=1}^{L} (\sum_{j\in S_i \cap P_{t, t'}} \bar{\chi}^{t, t'}_j)^2}{(1+\eps)^2(\frac{s}{T} + \delta)} \\
&\leq \frac{(1+\eps)^2(\frac{s}{T} + \delta) \cdot \sum_{i=1}^{m} imp(i)}{(1+\eps)^2(\frac{s}{T} + \delta)} =\sum_{i=1}^{m} imp(i).
\end{align*}
Hence, the values $gap^{t}(i)$ satisfy the two properties of the MWU statement in \Cref{lem:MWU} with $W=O(\log m)$. Thus, after running the game for $T$ rounds by going through the parts $P_1$, $P_2$, \dots, $P_T$ and doing the above for each of them sequentially, since $T=\Omega(W\log m/\eps^2)$, we get that for each set $i\in [m]$, we have $\sum_{t=1}^{T}gap^{t}(i) \leq (1+\eps)T$. This means 
\[
\sum_{t=1}^{T} \sum_{t'=1}^{L} (\sum_{j\in S_i \cap P_{t, t'}} \bar{\chi}^{t, t'}_j)^2 \leq (1+\eps)^3 (s+T\delta) = (1+\eps)^3 s+O(\log^8 m).
\]
At this point, we have $TL = \Theta(\sqrt{n\log^5 m})$ solutions $\bar{\chi}^{t, t'}\in \{-1,1\}^{[P_{t,t'}]}$ for $t\in [T]$ and $t'\in [L]$ and we need to find a good mixture of them. That is, we want to find a mixture vector $\chi'\in \{-1, 1\}^{TL}$ which determines for each of these solutions whether to take itself or its negation, by setting the output vector $\chi\in\{-1, 1\}^n$ as $\chi_j=\bar{\chi}^{t, t'}_j \cdot \chi'_{t,t'}$ where $j\in P_{t,t'}$. This part is quite similar to step (3) in the proof of \Cref{thm:rootdepth-suboptimal}. In particular, define $a_{i,(t-1)T+t'} = (\sum_{j\in S_i\cap P_{t, t'}} \bar{\chi}^{t,t'}_j).$ By invoking the sequential derandomization of \Cref{thm:sequential_derandomization}, which runs in $\tilde{O}(\sqrt{n})$ depth and with $\tilde{O}(n+m+\sum_{i=1}^{m}|S_i|)$ work, we get a mixture vector $\chi'\in \{-1, 1\}^{TL}$ such that  $disc^2(i) = (\sum_{k=1}^{TL} a_{i,k} \cdot  \chi'_{k})^2 = (\sum_{k=1}^{TL} a^2_{i,k} \cdot O(\log m)).$ Hence, overall, for each $i\in [m]$, we have
\begin{align*}
    disc^2(S_i)&=(\sum_{j\in S_i} \chi_j)^2 \\
    &= (\sum_{k=1}^{TL} a_{i,k} \cdot \chi'_{k})^2 \\
    &\leq \sum_{k=1}^{TL} a^2_{i, k} \cdot O(\log m) \\
    &= \sum_{t=1}^{T} \sum_{t'=1}^{L} (\sum_{j\in S_i\cap P_{t, t'}} \bar{\chi}^{t,t'}_j)^2 \cdot O(\log m) \\
    &\leq ((1+\eps)^3 s+O(\log^8 m))\cdot O(\log m) \\
    &= O(s \log m).
\end{align*}
Here, the last inequality uses that $\eps=1/(10\log m)$ and $s\geq \Omega(\log^{10} m)$.
\end{proof}

\section{Optimal discrepancy in polylogarithmic depth --- unweighted}
\label{sec:unweighted}
In this section, we present a deterministic parallel algorithm that achieves an asymptotically optimal discrepancy for the (unweighted) set balancing problem---i.e., matching what follows from the Chernoff bound---using near-linear work and polylogarithmic depth, therefore proving \Cref{thm:main-unweighted}.

In \Cref{subsec:nlogm}, we present the core ingredients in this result, in the format of a lemma that achieves $O(\sqrt{n\log m})$ discrepancy. The lemma will actually be somewhat more general for two reasons: (1) to make way for its own proof via recursion, and (2) to facilitate its later usage. Later, in \Cref{subsec:slogm}, we use this result and some extra helper lemma to get discrepancy $O(\sqrt{s\log m})$.


\subsection{A polylogarithmic-depth recursive algorithm for $O(\sqrt{n\log m})$ discrepancy}
\label{subsec:nlogm}
In this subsection, we prove the following result: 

\begin{lemma}\label{lem:inductionUnweightedDisc}
There is an absolute constant $C'>0$ for which the following holds. Let $n,m \in \mathbb{N}$ with $m,n \geq 2$ and $\disc \in \mathbb{R}_{\geq 0}$. Let $A \in \mathbb{R}^{m \times n}$ satisfying that $\sum_{j=1}^n a^2_{ij} \leq \disc$ for every $i \in [m]$ and $\max_{i \in [n], j \in [m]} a^2_{ij} \leq \frac{(\log^5 (nm))}{n}\disc$. Then, there exists a deterministic parallel algorithm that can compute a vector $\chi \in \{-1,1\}^n$ with $\tilde{O}(n + m + nnz(A))$ work and $\poly(\log (nm))$ depth such that, for every $i \in [m]$, it holds that $(\sum_{j=1}^{n} a_{ij} \chi_j)^2 = (2-1/\log n) \cdot (C'\log m) \cdot \disc$.   
\end{lemma}

This result itself can be viewed as a generalization of achieving discrepancy $O(\sqrt{n\log m})$ in the unweighted setting. In particular, for that purpose, the reader can interpret $a_{ij}=1$ for $j\in S_i$ and $a_{ij}=0$ otherwise, and $\disc=n.$ The lemma generalizes the statement mainly by allowing a polylogarithmic range of weights for nonzero coefficients $a_{ij}$. As mentioned before, this generality is necessary for our recursive algorithm that proves the lemma, and moreover, it helps in later applications of the result.

The key ingredient in proving \Cref{lem:inductionUnweightedDisc} is a helper lemma, which we next state as \Cref{lem:RecursionCreator}. The former provides a certain partitioning scheme for the variables, along with a value assignment inside each part, which basically sets up the recursion. Then the task of finding a good mixture of these assignments will be solved by recursion. We first prove \Cref{lem:RecursionCreator}, and then go back to proving \Cref{lem:inductionUnweightedDisc}.

\begin{lemma}\label{lem:RecursionCreator}
There exists an absolute constant $c > 0$ such that the following holds. Let $n,m \in \mathbb{N}$ with $m \geq 2$ and $n \geq c\log^{30}(m)$ and $\disc \in \mathbb{R}_{\geq 0}$.  Let $A \in \mathbb{R}^{m \times n}$ satisfying that $\sum_{j=1}^n a^2_{ij} \leq \disc$ and $\max_{i \in [m], j \in [n]} a^2_{ij} \leq \frac{(\log^5 (nm))}{n} \disc$. Then, there exists a deterministic parallel algorithm that can compute a vector $\chi \in \{-1,1\}^n$ and a partition $[n] = P_1 \sqcup P_2 \sqcup \ldots \sqcup P_L$ for some $L \leq n/2$  with $\tilde{O}(n + m + nnz(A))$ work and $\tilde{O}(1)$ depth satisfying that \begin{itemize}
    \item[(A)] $\sum_{ \ell \in L} \left(\sum_{j \in P_\ell} a_{ij} \chi_j \right)^2 \leq\left(1 + \frac{1}{10\log^2 n} \right) \disc$, and
    \item[(B)] for every $i \in [m]$ and $\ell \in [L]$, $\left(\sum_{j \in P_t} a_{ij} \chi_j \right)^2 \leq \frac{O(\log (nm))}{L}\disc$.
\end{itemize}
\end{lemma}
\begin{proof}
Set $\eps = 1/(100\log^2(nm))$. First, we partition the $n$ variables into $L=2^{\lceil \log(\frac{n}{\log^{20} (nm)})\rceil} \approx \frac{n}{\log^{20} (nm)}$ parts $P_1 \sqcup P_2 \sqcup \ldots \sqcup P_L=[n]$ with the following two properties:
\begin{itemize}
    \item[(I)] For each part $\ell\in [L]$, we have $|P_{\ell}|\leq 2n/L = O(\log^{20}(nm))$
    \item[(II)] For each part $\ell \in [L]$ and each set $i\in [m]$, we have $\sum_{j \in P_\ell} a^2_{ij} \leq (1+2\eps)\Delta/L.$
\end{itemize}
This can be computed via \Cref{lem:weighted_partition} using $\tilde{O}(nnz(A) + n + m))$ work and $\poly(\log (mn))$ depth. In particular, the third property in \Cref{lem:weighted_partition} implies that for each set $i\in [m]$ and each part $\ell\in [L]$, we have \begin{align*}
    \sum_{j \in P_\ell} a^2_{ij} &\leq  (1+\eps) \frac{\disc}{L} + O(\frac{\log^2 (nm)}{\eps^3}) \cdot (\max_{i\in [m], j \in [n]} a^2_{ij})\\
    &\leq (1+\eps) \frac{\disc}{L} + O(\frac{\log^2 (nm)}{\eps^3}) \cdot \frac{(\log^5 (nm))}{n} \disc \\
    &= (1+\eps) \frac{\disc}{L} + O(\frac{\log^2 (nm)}{\eps^3}) \cdot \frac{(\log^5 (nm))}{L\log^{20} (nm)} {\disc} \\
    &\leq (1+\eps) \frac{\disc}{L} + \eps \frac{\disc}{L} \\ 
    &= (1+2\eps)\frac{\Delta}{L}.
\end{align*}

Notice that we can easily apply sequential derandomization (\Cref{thm:sequential_derandomization}) in each part $\ell \in [L]$ to obtain a vector $\chi^{\ell}\in \{-1,1\}^{[P_\ell]}$ such that for every $i \in [m]$, we have $\left(\sum_{j \in P_\ell} a_{ij} \chi^\ell_j \right)^2 \leq \frac{O(\log (m))}{L}\disc$. This would satisfy guarantee (B) in the lemma we are proving. Furthermore, since each part has only $2n/L = O(\log^{20}(nm))$ variables, and we solve different parts in parallel, this would take depth $\poly(\log (mn))$ and work $\tilde{O}(nnz(A) + n + m))$. But that alone would not provide the more important guarantee (I). To achieve (A), we work somewhat differently by appealing to the Multiplicative Weight Updates (MWU) method, as recalled in \Cref{lem:MWU}, in a manner similar to what we did in the proof of \Cref{thm:rootdepth-optimal}.

Concretely, we break the parts into $T=\Theta(\log^2 (nm)/\eps^2)$ groups, by viewing parts $(t-1)(L/T)+1$ to $tL/T$ as group $t$, for each $t\in [T]$. Here, we choose $T$ to be a power of $2$ so that $T|L$. We will process the groups sequentially, each as one round of MWU. Initially, we set the importance value of each set $i\in[m]$ as $imp(i)=1$. Then, we process the groups one by one and adjust the importance values as we will describe. Let us zoom in on one round.

Consider round $t\in [T]$ and the corresponding group of $L/T$ parts $P_{(t-1)(L/T)+1}, P_{(t-1)(L/T)+2}, \dots, P_{t(L/T)}$. We invoke the sequential derandomization of \Cref{thm:sequential_derandomization} in each of the parts independently, all with the current importance value $imp(i)$ for each set $i\in [m]$. Since each part has at most $2n/L = O(\log^{20}(nm))$ variables, and we solve different parts in parallel, this takes depth $\poly(\log (mn))$. From \Cref{thm:sequential_derandomization} (setting $M = mn$), we get two properties for each set $i\in [m]$ and each $\ell \in [(t-1)(L/T)+1, t(L/T)]$:
\begin{itemize}
    \item $\left(\sum_{j \in P_\ell} a_{ij} \chi^\ell_j \right)^2 \leq \frac{O(\log (nm))}{L}\disc$
    \item $\sum_{i=1}^{m} imp(i) \cdot \left(\sum_{j \in P_\ell} a_{ij} \chi^\ell_j \right)^2  \leq  (1+\eps) \sum_{i=1}^{m} imp(i) \cdot(1+2\eps)\frac{\disc}{L}$.
\end{itemize}
As mentioned before, the first already gives property (B) of the lemma. We next examine property (A). Let us define 
\[
gap^{t}(i) = \frac{\sum_{\ell=(t-1)(L/T)+1}^{t(L/T)} \left(\sum_{j\in P_{\ell }}a_{ij} \chi^\ell_j \right)^2}{(1+\eps) (1+2\eps)\frac{\disc}{T}}
\]
From the above two properties, we conclude the following two guarantees about $gap^{t}(i)$: 
\begin{itemize}
    \item for each $i\in [m]$, we have $gap^{t}(i) \in [0, W]$ for $W=O(\log (mn))$,
    \item $\sum_{i=1}^{m} imp(i) \cdot gap^{t}(i) \leq \sum_{i=1}^{m} imp(i)$
\end{itemize}
Hence, the guarantees fit exactly the definition of the oracle in the MWU framework, as recapped in \Cref{lem:MWU}. Thus, by running the game for $T=\Theta(W\log m/\eps^2)$ rounds and processing all the groups with importance values updated according to \Cref{lem:MWU}, we get the following guarantee: for each $i\in [m]$, we have $\sum_{t=1}^{T} gap^{t}(i)\leq (1+\eps)T$. That is,

\begin{align*}
    \sum_{t=1}^{T} \frac{\sum_{\ell=(t-1)(L/T)+1}^{t(L/T)} \left(\sum_{j\in P_{\ell }}a_{ij} \chi^\ell_j \right)^2}{(1+\eps) (1+2\eps)\frac{\disc}{T}} = \frac{\sum_{\ell=1}^{L} \left(\sum_{j\in P_{\ell }}a_{ij} \chi^\ell_j \right)^2}{(1+\eps) (1+2\eps)\frac{\disc}{T}} \leq (1+\eps)T,
\end{align*}
which implies
\[\sum_{\ell=1}^{L} \left(\sum_{j\in P_{\ell }}a_{ij} \chi^\ell_j \right)^2 \leq (1+\eps)^2(1+2\eps) \disc \leq \left(1+\frac{1}{10\log^2 n} \right)\disc.\]
This proves property (A) and thus concludes the proof of the lemma.
\end{proof}

We can now go back to proving \Cref{lem:inductionUnweightedDisc}.

\begin{proof}[Proof of \Cref{lem:inductionUnweightedDisc}]
We present a proof by induction on $n$ (i.e., creating a recursive algorithm as a function of $n$). Also, we first describe how the algorithm works and provides the desired guarantees, and then in the end come back to bound its computational depth and work.

If $n <c\log^{30} m$ where $c$ is the constant in \Cref{lem:RecursionCreator}, then we are in the base case. Then, we simply solve the problem by invoking the sequential derandomization of \Cref{thm:sequential_derandomization}, which provides a vector $\chi \in \{-1,1\}^n$ such that that $(\sum_{j=1}^{n} a_{ij} \chi_j)^2 \leq (C\log m) \cdot \disc$. Here, $C$ is the constant in \Cref{thm:sequential_derandomization} which satisfies $C(2-1/\log n)\leq C'$ by choosing $C'= 2C$. Otherwise, we are in the case where we solve the problem via recursion, as we discuss next.

By invoking \Cref{lem:RecursionCreator}, we spend $\tilde{O}(n + m + nnz(A))$ work and $\poly(\log (mn))$ depth and we get a vector $\bar{\chi} \in \{-1,1\}^n$ and a partition $[n] = P_1 \sqcup P_2 \sqcup \ldots \sqcup P_L$ for some $L \leq n/2$ satisfying the following two properties:
\begin{itemize}
    \item[(A)] $\sum_{ \ell \in L} \left(\sum_{j \in P_\ell} a_{ij} \bar{\chi}_j \right)^2 \leq\left(1 + \frac{1}{10\log^2 n} \right) \disc$, and
    \item[(B)] for every $i \in [m]$ and $\ell \in [L]$, $\left(\sum_{j \in P_\ell} a_{ij} \bar{\chi}_j \right)^2 \leq \frac{O(\log (nm))}{L}\disc$.
\end{itemize}

Then, the remaining task is to determine a mixture vector $\chi'\in \{-1, 1\}^{L}$ so that we can mix the $\bar{\chi}$ solutions of parts $P_1$ to $P_L$ accordingly, i.e., by setting $\chi_{j} = \bar{\chi}_j \cdot \chi'_{\ell}$ where $j\in P_{\ell}$. This problem is similar to the mixture selection in the proofs of \Cref{thm:rootdepth-suboptimal} and \Cref{thm:rootdepth-optimal}. However, unlike those results, which solve the mixture selection via sequential derandomization, we now have a large number of parts $L$, which can be as large as $n/2$. Thus, it would be too slow to use sequential derandomization here. Instead, we can invoke recursion.

In particular, for each $i\in[m]$ and each $\ell\in [L]$, define $a'_{i\ell} = \left(\sum_{j \in P_\ell} a_{ij} \bar{\chi}_j \right)$, and define $\disc' = \left(1 + \frac{1}{10\log^2 n} \right) \disc$. These satisfy the conditions of our inductive lemma (\Cref{lem:inductionUnweightedDisc}) for $n'=L\leq n/2$, in the sense that for each $i\in [m]$, we have $\sum_{\ell=1}^{n'} (a'_{i\ell})^2 \leq \disc'$ and $\max_{i\in [m], \ell\in [L]} (a'_{i\ell})^2 \leq \frac{\log^{5}(n'm)}{n'}\disc'$. Thus, by invoking \Cref{lem:inductionUnweightedDisc} recursively/inductively on this instance with $n'\leq n/2$ variables, we get a vector $\chi'\in \{-1, 1\}^{L}$ with the guarantee that for each set $i \in [m]$, we have $(\sum_{\ell=1}^{n'} a'_{i\ell} \chi'_j)^2 = (2-1/\log n') \cdot (C'\log m) \cdot \disc'$. Hence, we can conclude that for every set $i\in [m]$, we have
\begin{align*}
 (\sum_{j=1}^{n} a_{ij} \chi_j)^2 & = (\sum_{\ell=1}^{n'} \sum_{j\in P_{\ell}} a_{ij} \chi_j)^2 = (\sum_{\ell=1}^{n'} \sum_{j\in P_{\ell}} a_{ij} \bar{\chi}_j \chi'_{\ell})^2 \\
 &= (\sum_{\ell=1}^{n'} \chi'_{\ell}\sum_{j\in P_{\ell}} a_{ij} \bar{\chi}_j)^2 = (\sum_{\ell=1}^{n'} \chi'_{\ell}a'_{i\ell})^2 \\
&\leq (2-\frac{1}{\log n'}) \cdot (C'\log m) \cdot \disc' \\
&=  (2-\frac{1}{\log n'})(1+\frac{1}{10\log^2 n})\cdot (C'\log m) \cdot \disc \\
&\leq \big((2-\frac{1}{\log n-1})(1+\frac{1}{10\log^2 n})\big)\cdot (C'\log m) \cdot \disc \\
&\leq (2-\frac{1}{\log n})\cdot (C'\log m),
\end{align*}
which satisfies the desired output guarantee. 

Finally, we discuss the computational depth and work of this algorithm. If we are in the base case of $n=O(\log^{30} m)$, the algorithm follows just by invoking the sequential derandomization of \Cref{thm:sequential_derandomization}, which has $\poly(\log (nm))$ depth and $\tilde{O}(n + m + nnz(A))$ work. Let us now examine the depth in the recursive case. For larger $n$, we invoked \Cref{lem:RecursionCreator}, which has depth $\poly(\log (nm))$, and then we solved the remaining (mixture) problem by applying a recursion on an instance with $L\leq n/2$ variables. Hence, the depth of the instance with $n$ variables and $m$ sets satisfies the recursion $D(n, m)\leq \poly(\log (nm)) + D(n/2, m)$ with the base case of $D(n, m) \leq \poly(\log (nm))$ if $n=O(\log^{30} m).$ Hence, $D(n, m) \leq \poly(\log (nm))$. A similar argument shows that the work $W(n, m)$ of the instance with $n$ variables and $m$ satisfies $W(n, m)\leq \tilde{O}(n + m + nnz(A))+ W(n/2, m)$, which thus shows that the work is bounded by $\tilde{O}(n + m + nnz(A))$.
\end{proof}

\subsection{A polylogarithmic-depth algorithm for $O(\sqrt{s\log m})$ discrepancy}
\label{subsec:slogm}
In this subsection, we prove \Cref{thm:main-unweighted}, using \Cref{lem:inductionUnweightedDisc} developed in the previous subsection. For that, we will need an additional helper lemma about the derandomization of pairwise analysis.

\begin{lemma}\label{lem:pairwise_derandomization}
Let $n, m \in \mathbb{N}$ with $m \geq 2$, and let $\{S_1, S_2, \ldots, S_m\}$ be a family of subsets of $[n]$ such that $|S_i|\leq k$ for all $i\in [m]$. Also, suppose that for each $i\in [m]$ we are given an importance value $imp(i)\leq \mathbb{R}^+$. There exists a deterministic parallel algorithm that can compute a vector $\chi \in \{-1, 1\}^n$, using $\tilde{O}(n + m + \sum_{i = 1}^m |S_i|) \cdot \poly(k)$ work and $\poly(k\log(mn))$ depth such that we have $\sum_{i=1}^{m} imp(i) \cdot (\sum_{j \in S_i} \chi_j)^2 \leq \sum_{i=1}^{m} imp(i) \cdot |S_i|$.
\end{lemma}
\begin{proof}[Proof Sketch]
Notice that under a random selection of $\chi$ with merely pairwise independence---i.e., for each $j, j'\in [n]$, where $j\neq j'$,  and each $(a,b)\in \{-1,1\}^{2}$ we have $\Pr[\chi_j=a]=1/2$ and $\Pr[(\chi_j, \chi_{j'})=(a, b)]=1/4$---we have 
\[
\mathbb{E}[\sum_{i=1}^{m} imp(i) \cdot (\sum_{j \in S_i} \chi_j)^2] = \sum_{i=1}^{m} imp(i) \cdot \mathbb{E}[(\sum_{j \in S_i} \chi_j)^2] = \sum_{i=1}^{m} imp(i) \cdot \sum_{j \in S_i} E[(\chi_j)^2] =  \sum_{i=1}^{m} imp(i) \cdot |S_i|,
\]
where the penultimate equality relied on the pairwise independence of $\chi_j$ and $\chi_{j'}$ for $j\neq j'$. Given this, such a vector $\chi$ can be computed deterministically and in parallel, using Luby's method for work-efficient parallel derandomization of pairwise independent analysis~\cite{luby1988removing}. Following this method, the statement follows as a black-box application of Lemma 3.4 in \cite{GGR2023Chernoff}.
\end{proof}

Finally, we go back to \Cref{thm:main-unweighted} and present its proof. For convenience, we restate \Cref{thm:main-unweighted}.

\thmMainUnweighted*
\begin{proof}
    If $s\leq \poly(\log(mn))$, the result follows from \Cref{thm:FOCS23}. Let us assume that $s$ is larger. Set $\eps=0.01$.  First, we partition the $n$ variables into $L=\frac{s}{\Theta(\log m/\eps^2)}$ parts $P_1 \sqcup P_2 \sqcup \ldots \sqcup P_L=[n]$, where $L$ is a power of two, with the following two properties:
\begin{itemize}
    \item[(I)] For each part $\ell \in [L]$ and each set $i\in [m]$, we have $|S_i\cap P_{\ell}| \leq (1+\eps)s/L.$
\end{itemize}
This can be computed using \Cref{lem:weighted_partition} using work $\tilde{O}(nnz(A) + n + m))$ and depth $\poly(\log (mn))$. 
We bundle the parts into $T=\Theta(W\log m/\eps^2)$ groups, by viewing parts $(t-1)(L/T)+1$ to $tL/T$ as group $t$, for each $t\in [T]$. Again, $T$ is a power of two, so we have $T|L$. We will process the groups sequentially, each as one round of MWU. Initially, we set the importance value of each set $i\in[m]$ as $imp(i)=1$. Then, we process the groups one by one and adjust the importance values as we will describe. 

Let us zoom in on one round. Consider round $t\in [T]$ and the corresponding group of $L/T$ parts $P_{(t-1)(L/T)+1}, P_{(t-1)(L/T)+2}, \dots, P_{t(L/T)}$. We invoke the pairwise derandomization of \Cref{lem:pairwise_derandomization} on the family of sets $S_i\cap P_{\ell}$ for all $i\in [m]$  and $\ell\in [(t-1)(L/T)+1, t(L/T)]$. Moreover, all subsets of set $S_i$ are given importance value $imp(i)$ inherited from set $S_i$ where $i\in [m]$. 

Since in each part $\ell \in [L]$, each set $S_i\cap P_{\ell}$ has size at most $k=(1+\eps)s/L=O(\log m)$, applying \Cref{lem:pairwise_derandomization} takes depth $\poly(\log (mn))$. From \Cref{lem:pairwise_derandomization}, we get the following property:

\begin{align} \label{ineq:pairwisebound}
\sum_{i=1}^{m} imp(i) \cdot \Bigg(\sum_{\ell=(t-1)(L/T)+1}^{t(L/T)} \left(\sum_{j \in S_i\cap P_\ell} \bar{\chi}^\ell_j \right)^2 \Bigg) \leq  (1+\eps) \sum_{i=1}^{m} imp(i) \cdot \frac{s}{T}.
\end{align}
Let us define 
\[
gap^{t}(i) = \frac{\Bigg(\sum_{\ell=(t-1)(L/T)+1}^{t(L/T)} \left(\sum_{j \in S_i\cap P_\ell} \bar{\chi}^\ell_j \right)^2 \Bigg)}{(1+\eps)\frac{s}{T}}
\]
From \Cref{ineq:pairwisebound}, we can conclude that $\sum_{i=1}^{m} imp(i) \cdot gap^{t}(i) \leq \sum_{i=1}^{m} imp(i)$.
Furthermore, we have that $gap^{t}(i)\leq O(\log m)$. The reason is that $\left(\sum_{j \in S_i\cap P_\ell} \bar{\chi}^\ell_j \right)^2 \leq (|S_i\cap P_\ell|)^2 \leq (1+\eps)k(s/L),$ and thus \[\Bigg(\sum_{\ell=(t-1)(L/T)+1}^{t(L/T)} \left(\sum_{j \in S_i\cap P_\ell} \bar{\chi}^\ell_j \right)^2 \Bigg) \leq (1+\eps)k \cdot (s/L) \cdot (L/T) = k \cdot (1+\eps) (s/T).\]
So, $gap^{t}(i) \leq k$, which means we have $gap^{t}(i) \in [0, W]$ for $W=O(\log m)$. 

Hence, the guarantees fit exactly the definition of the oracle in the MWU framework, as recapped in \Cref{lem:MWU}. Thus, by running the game for $T=\Theta(W\log m/\eps^2)$ rounds and processing all the groups with importance values updated according to \Cref{lem:MWU}, we get the following guarantee: for each $i\in [m]$, we have $\sum_{t=1}^{T} gap^{t}(i)\leq (1+\eps)T$. That is,

\begin{align*}
    \sum_{t=1}^{T} \frac{\sum_{\ell=(t-1)(L/T)+1}^{t(L/T)} \left(\sum_{j\in S_i \cap P_{\ell }} \bar{\chi}^\ell_j \right)^2}{(1+\eps) \frac{s}{T}} = \frac{\sum_{\ell=1}^{L} \left(\sum_{j\in S_i \cap P_{\ell }}\bar{\chi}^\ell_j \right)^2}{(1+\eps) \frac{s}{T}} \leq (1+\eps)T,
\end{align*}
which implies 
\[\sum_{\ell=1}^{L} \left(\sum_{j\in S_i \cap P_{\ell }} \bar{\chi}^\ell_j \right)^2 \leq (1+\eps)^2 s \leq (1+3\eps)s.\]
There is also the trivial bound that for each $i\in [m]$ and $\ell\in[L]$, we have 
\[\left(\sum_{j\in S_i \cap P_{\ell }} \bar{\chi}^\ell_j \right)^2 \leq (1+\eps)k s/L \leq O(\log m) \cdot s/L.\] 

These two conditions prepare us to invoke \Cref{lem:inductionUnweightedDisc}. In particular, for each $i\in[m]$ and each $\ell\in [L]$, define $a'_{i\ell} = \left(\sum_{j \in S_i\cap P_\ell} \bar{\chi}_\ell \right)$, and define $\disc' = \left(1 + 3\eps \right) s$. These satisfy the condition of \Cref{lem:inductionUnweightedDisc} for $n'=L$, in the sense that for each $i\in [m]$, we have $\sum_{\ell=1}^{n'} (a'_{i\ell})^2 \leq \disc'$ and $\max_{i\in [m], \ell\in [L]} (a'_{i\ell})^2 \leq \frac{\log^{5}(n'm)}{n'}\disc'$. Thus, by invoking \Cref{lem:inductionUnweightedDisc}, we get a vector $\chi'\in \{-1, 1\}^{L}$ with the guarantee that for each set $i \in [m]$, we have $(\sum_{\ell=1}^{n'} a'_{i\ell} \chi'_\ell)^2 =  2(C'\log m) \cdot \disc'$. Hence, we can conclude that for every set $i\in [m]$, we have
\begin{align*}
 (\sum_{j\in S_{i}} \chi_j)^2 & = (\sum_{\ell=1}^{n'} \sum_{j\in S_i\cap P_{\ell}}  \chi_j)^2 = (\sum_{\ell=1}^{n'} \sum_{j\in S_i\cap P_{\ell}} \bar{\chi}_j \chi'_{\ell})^2 \\
 &= (\sum_{\ell=1}^{n'} \chi'_{\ell}\sum_{j\in S_i \cap P_{\ell}} \bar{\chi}_j)^2 = (\sum_{\ell=1}^{n'} \chi'_{\ell}a'_{i\ell})^2 \\
&\leq 2(C'\log m) \cdot \disc' \leq (3C'\log m) \cdot s
\end{align*}
which satisfies the desired output guarantee. 
\end{proof}

\section{Optimal discrepancy in polylogarithmic depth --- weighted}
\label{sec:weighted}
In this section we prove our result for the weighted set balancing problem, \Cref{thm:main-weighted}, restated below.
\thmMainWeighted*

The general scheme will be similar to the recursive approach we had in the unweighted case, but the partitioning here will be far more complex and will have more nuanced properties (and these will spread also to how the partitioning is used). This partitioning is then used to create a recursive instance of the problem, with the help of the multiplicative weights update scheme for controlling the variance losses throughout the recursion levels.

\subsection{Partitioning for arbitrarily weighted set balancing instances}
A key part of proving \Cref{thm:main-weighted} is a new deterministic parallel algorithm for partitioning weighted instances, which we present as \Cref{thm:partition_with_isolation}. On a very high level, the difficulty with weighted instances stems from the fact that in each constraint $i \in [m]$, there can be a very small number of variables (say $K$ variables, which is at most polylogarithmic) that have a very high weight. Due to their small number, basic partitionings do not handle these items well and might potentially put several of them into the same part. Due to the large weight of these items, this would cause a substantial loss in the variance, when we try to solve the problem recursively. Notice also that we cannot easily discard or set aside these few items, to handle them separately, as we have $m$ constraints and each of the $n$ variables might be such a high-weight variable in some of the constraints. Intuitively, the deterministic partitioning we present will try to \textit{isolate} these few heavy-weight variables in each constraint. We do not achieve perfect isolation (which would require a very large number of parts in the partition, and would thus be useless), but we get something with a very strong control on the number of variables that remain non-isolated: Consider one constraint $i\in [m]$. For each part, we call all except one of those high-weight variables that fall into this part \textit{bad} variables. Notice that in a random partitioning with $\Omega(K^2)$ parts, the number of bad variables in each constraint would have a subconstant expectation and also an exponentially decaying tail, i.e., the probability of $z$ bad variables in one constraint will be upper bounded by $exp(-\Theta(z))$. Our deterministic partitioning performs various derandomizations to essentially enforce such an exponential tail on the number of bad variables in each constraint, while simultaneously achieving a nearly-even split of all other variables in the constraint, at the same time for all $m$ constraints.

\paragraph{Remark} The proof of \Cref{thm:partition_with_isolation} is complex and involves many pieces. We present the proof in several subsubsections. To enable a local view of each part, we use a top-down presentation where in each proof we identify the core piece(s) that would be sufficient to obtain the desired result, formulated as a separate lemma statement, and later state and prove that lemma. To help the readability, the subsubsections are labeled with the corresponding lemma names, and we have interleaved the presentations of these lemmas with sentences that indicate where each one was previously used and which next lemmas it will invoke.

\subsubsection{Main weighted partitioning theorem, with isolations (\Cref{thm:partition_with_isolation})}
The following theorem provides the formal statement of our partitioning. Since the statement has too many variables, let us provide some intuitive explanation for the statement, and how to match it to the above intuition. The output has two partitions, a coarse partition $\mathcal{P}$ with only polylogarithmic parts, and a much finer partition $\mathcal{Q}$, which is a refinement of $\mathcal{P}$ and has a number of parts around $\tilde{\Theta(n)}$, though smaller than $n/2$.  Intuitively, each round of MWU will process one part of $\mathcal{P}$, handling all the parts of $\mathcal{Q}$ in that part simultaneously, and eventually each part of $\mathcal{Q}$ will create one variable in our recursion. The reader can think of $S_i$ as the set of those few variables in constraint $i$, which we would like to isolate nearly completely in the fine partition $\mathcal{Q}$. In particular, property 5 in the theorem formalizes the exponential tail guarantee on the non-isolated items, discussed above. At the same time, each constraint has also other variables $B_i$ (a subtlety is that these two are not necessarily disjoint, but the reader can ignore that for now), which we would like to split nearly evenly in the coarse partition $\mathcal{P}$. This is so that the MWU scheme applied later in the recursion can achieve the desired averaging of the variance losses, similar to what we did in \Cref{lem:RecursionCreator}.

We note that the key part in \Cref{thm:partition_with_isolation} is what we present and prove later as \Cref{lem:partition_second}.

\begin{theorem}[Main Partition Theorem]
\label{thm:partition_with_isolation}
The following holds for every sufficiently small absolute constant $\delta> 0$:
Let $n,M \in \mathbb{N}$ with $n \geq (\log M)^{1/\delta}$ and $M \geq 3$.  Also, let $\{S_1,S_2,\ldots,S_{m_S}\}$ and $\{B_1,B_2,\ldots,B_{m_B}\}$ be two families of subsets of $[n]$ with $m_S,m_B \leq M$ satisfying:

\begin{itemize}
    \item $|S_i| \leq \log^{100}(nM)$ for every $i \in [m_S]$
    \item $|B_i| \geq \log^{30}(nM)$ for every $i \in [m_B]$
\end{itemize}
We are also given a vector $imp \in \mathbb{R}^{m_S}_{\geq 0}$.
Then, there exists a deterministic parallel algorithm with work $\left(n + m_S + m_B + \sum_{i=1}^{m_S} |S_i| + \sum_{i=1}^{m_B} |B_i| \right)\poly(\log(nM))$ and depth $\poly(\log(nM))$ that computes two partitions $\mathcal{P} = P_1 \sqcup P_2 \sqcup \ldots \sqcup P_{T_P} = [n]$ and $\mathcal{Q} = Q_1 \sqcup Q_2 \sqcup \ldots \sqcup Q_{T_Q} = [n]$ satisfying:

\begin{enumerate}
    \item $T_P \in \left[\log^{20}(nM), 10\log^{20}(nM) \right]$, $T_Q \leq n/2$ 
    \item $\mathcal{Q}$ is a refinement of $\mathcal{P}$
    \item $|B_i \cap P_t| \leq \left(1 + \frac{1}{\log^3 (n)} \right)\frac{|B_i|}{T_P}$ for every $i \in [m_B]$ and $t \in [T_P]$
    \item $|Q_t| = O(\log n)$ for every $t \in [T_Q]$
    \item $\sum_{i=1}^{m_S} imp_i (1+\delta)^{\sum_{t = 1}^{T_Q} \max(|S_i \cap Q_t| - 1, 0)} \leq \left(1 + \frac{1}{\log^3(n)} \right)\sum_{i=1}^{m_S} imp_i$
\end{enumerate}
\end{theorem}
\begin{proof}
The proof has three steps, we first describe how we compute the coarse partition $\mathcal{P}$ and discuss its properties, then describe how we compute the fine partition $\mathcal{Q}$, and in the end discuss why $\mathcal{Q}$ satisfies the desired properties. This second step will involve invoking our main weighted partitioning lemma, which we present and prove later as \Cref{lem:partition_second}.

\medskip
\paragraph{The coarse parition} We first compute the coarse partition $P_1 \sqcup P_2 \sqcup \ldots \sqcup P_{T_P}$ by using the unweighted partitioning algorithm of $\cref{lemma:partition_unweighted}$. Concretely, let $T_P$ be an arbitrary power of two in the interval $\left[\log^{20}(nM), 10\log^{20}(nM) \right]$. We use the unweighted partitioning algorithm with $\eps_{L\ref{lemma:partition_unweighted}} = \frac{1}{2\log^3(n)}$ to compute a partition $P_1 \sqcup P_2 \sqcup \ldots P_{T_P}$ such that, for every $i \in [m_B]$ and $t \in [T_P]$, we have

\[|B_i \cap P_t| \leq (1+\eps_{L\ref{lemma:partition_unweighted}})|B_i|/T_P + O(\log M )/ \eps^2_{L\ref{lemma:partition_unweighted}}\]
This uses $\left(n + m_B + \sum_{i=1}^{m_B} |B_i|  \right)\poly(\log(nM))$ work and $\poly(\log(nM))$ depth. In particular, for every sufficiently small $\delta > 0, i \in [m_B]$ and $t \in [T_P]$, we have 

\[|B_i \cap P_t| \leq \left(1 + \frac{1}{\log^3 (n)} \right)\frac{|B_i|}{T_P}.\]

\medskip
\paragraph{The fine partition} We now sequentially iterate through the parts of the partition $\mathcal{P}$ and refine each part by using the algorithm of $\cref{lem:partition_second}$ (presented later in this section). More concretely, we refine part $t_P$, for $t_P \in [T_P]$, by invoking the algorithm of $\cref{lem:partition_second}$ with the following input:

\begin{itemize}
	\item $K = \log^{100}(nM)$ 
	\item $R = P_{t_P}$
	\item $\hat{S}_i = S_i \cap P_t$ for every $i \in [m_S]$
	\item $imp = imp^{(t_P)}$ ($imp^{(t_P)}$ has been computed previously and at the beginning we set $imp^{(1)} = imp$)
\end{itemize}
The algorithm has work  $\left(n + m_S + \sum_{i=1}^{m_S} |S_i|  \right)\poly(\log(nM))$ and depth $\poly(\log(nM))$ and computes a partition $Q^{(t_P)}_1 \sqcup Q^{(t_P)}_2 \sqcup \ldots Q^{(t_P)}_{T_{t_P}} = P_{t_P}$ satisfying

\begin{enumerate}
	\item $T_{t_P} \leq \frac{|P_{t_P}|}{100} + \poly(\log(nM))$
	\item $|Q^{(t_P)}_t| = O(\log n)$ for every $t \in [T_{t_P}]$
	\item $\sum_{i=1}^{m_S} imp^{(t_P)}_i (1 + \delta_{L\ref{lem:partition_second}})^{\sum_{t = 1}^{T_{t_P}} \max (|S_i \cap Q^{(t_P)}_t| -1,0)} \leq \left( 1 + \frac{1}{\log^{30}(nM)}\right)\sum_{i=1}^{m_S} imp^{(t_P)}_i$ 
\end{enumerate} 

\noindent Then, for each $i \in [m_S]$, we define $imp^{(t_P+1)} \in \mathbb{R}^{m_S}_{\geq 0}$ as

\[imp^{(t_P+1)}_i = imp^{(t_P)}_i (1 + \delta_{L\ref{lem:partition_second}})^{\sum_{t = 1}^{T_{t_p}} \max (|S_i\cap Q^{(t_P)}_t| -1,0)}.\]
Also, in the end, we define $Q_1 \sqcup Q_2 \sqcup \ldots Q_{T_Q}$ as the partition we obtain by taking the union of the partitions $Q^{(t_P)}_1 \sqcup Q^{(t_P)}_2 \sqcup \ldots Q^{(t_P)}_{T_{t_P}}$ for $t_P \in [T_P]$. 

\medskip
\paragraph{Properties of the fine parition}
We now argue that the fine partition has the desired properties. First, note that
\[T_Q \leq \sum_{t_p=1}^{T_P} \left(\frac{|P_{t_P}|}{100} + \poly(\log(nM))\right) \leq  \frac{n}{100} + \poly(\log nM).\]
Thus, using that $n \geq \left(\log(M) \right)^{1/\delta}$ and $M \geq 3$, we get that $T_Q \leq n/2$ for all sufficiently small $\delta > 0$.
It thus only remains to verify that 
\[\sum_{i=1}^{m_S} imp_i (1+\delta)^{\sum_{t = 1}^{T_Q} \max(|S_i \cap Q_t| - 1, 0)} \leq \left(1 + \frac{1}{\log^3(n)} \right)\sum_{i=1}^{m_S} imp_i.\]
First, recall that $imp^{(1)} = imp$ and for every $t_P \in [T_P]$, we have
 \[\sum_{i=1}^{m_S} imp_i^{(t_P + 1)} \leq \left(1 + \frac{1}{\log^{30}(nM)} \right)\sum_{i=1}^{m_S}imp^{(t_P)}_i.\]
Therefore, we get
\begin{align*}
\sum_{i=1}^{m_S} imp_i^{(T_P + 1)} &
\leq \left(1 + \frac{1}{\log^{30}(nM)} \right)^{T_P} \sum_{i=1}^{m_S}imp_i \\
&\leq \left(1 + \frac{1}{\log^{30}(nM)} \right)^{10\log^{20}(nM)} \sum_{i=1}^{m_S}imp_i \\
&\leq  \left(1 + \frac{1}{\log^3(n)} \right)\sum_{i=1}^{m_S}imp_i.
\end{align*}
Moreover, a simple induction shows that for every $t'_P \in \{0,1,\ldots,T_P\}$ and $i \in [m_S]$, it holds that
\[imp^{(t'_P + 1)}_i = imp_i(1 + \delta_{L\ref{lem:partition_second}})^{\sum_{t_P = 1}^{t'_P} \sum_{t=1}^{T_{t_P}} \max(|S_i \cap Q^{(t_P)}_t)| - 1,0}). \]
Thus, for every sufficienlty small $\delta \in (0,\delta_{L\ref{lem:partition_second}}]$, we get
\begin{align*}
\sum_{i=1}^{m_S} imp_i (1+\delta)^{\sum_{t = 1}^{T_Q} \max(|S_i \cap Q_t| - 1, 0)} 
&\leq \sum_{i=1}^{m_S} imp_i (1+\delta_{L\ref{lem:partition_second}})^{\sum_{t = 1}^{T_Q} \max(|S_i \cap Q_t| - 1, 0)} \\
&= \sum_{i=1}^{m_S}  imp_i(1 + \delta_{L\ref{lem:partition_second}})^{\sum_{t_P = 1}^{T_P} \sum_{t=1}^{T_{t_P}} \max(|S_i \cap Q^{(t_P)}_t)| - 1,0)} \\
&= \sum_{i=1}^{m_S} imp^{(T_P + 1)}_i \\
&\leq \left(1 + \frac{1}{\log^3(n)} \right)\sum_{i=1}^{m_S}imp_i.
\end{align*}
\end{proof}

\subsubsection{Main weighted partitioning lemma, with isolations (\Cref{lem:partition_second})}
A key part in proving our main partitioning theorem was \Cref{lem:partition_second}, which was invoked in each part of the coarse partition to produce the corresponding segment of the fine partition. Indeed, this lemma abstracts most of the technical work. Next, we state this lemma and prove it. This lemma itself will be based on \Cref{lem:partitioning_sample_partition}, which we present and prove afterward.

\begin{lemma} (Main Partition Lemma)
\label{lem:partition_second}
The following holds for every sufficiently small absolute constant $\delta> 0$:
Let $n,M,K \in \mathbb{N}$ with $n \geq 1/\delta$. Let $R \subseteq [n]$ be a subset and $\{\hat{S}_1,\hat{S}_2,\ldots,\hat{S}_{m_S}\}$ be a family of subsets of $R$ with $m_S \leq M$ and $|\hat{S}_i| \leq K$ for every $i \in [m_S]$. Also, let $imp \in \mathbb{R}_{\geq 0}^{m_S}$. 
Then, there exists a deterministic parallel algorithm with work $\left(n + m_S\right)\poly(K\log(nM))$ and depth $\poly(K\log(nM))$ that computes a partition $\hat{Q}_1 \sqcup \hat{Q}_2 \sqcup \ldots \sqcup \hat{Q}_T$ of the subset $R$ such that:

\begin{enumerate}
    \item $T \leq \frac{|R|}{100} + \poly(K\log(nM))$
    \item $|\hat{Q}_t| = O(\log n)$ for every $t \in [T]$
    \item  $\sum_{i=1}^{m_S} imp_i(1 + \delta)^{\sum_{t = 1}^T \max (|\hat{S}_i \cap \hat{Q}_t| -1,0)} \leq \left( 1 + \frac{1}{\log^{30}(nM)}\right) \sum_{i=1}^{m_S} imp_i$ 
\end{enumerate}
\end{lemma}
\begin{proof}
Let $R^{(1)}:= R$, $imp^{(1)} := imp$ and $\ell_{max} := \lceil K \log^2(n)\rceil$. Next, we first describe the algorithm which is basically invoking  \Cref{lem:partitioning_sample_partition} in $\ell_{max}$ iterations, and then discuss the achieved properties.

\medskip
\paragraph{Algorithm} The algorithm runs in $\ell_{max}$ sequential rounds.
In round $\ell \in [\ell_{max}]$, the algorithm invokes the algorithm of \Cref{lem:partitioning_sample_partition} with input 

\begin{itemize}
	\item $R^{L\ref{lem:partitioning_sample_partition}} = R^{(\ell)}$
	\item the family of input sets is $\{\hat{S_1} \cap R^{(\ell)},\hat{S_2} \cap R^{(\ell)},\ldots, \hat{S}_{m_S} \cap R^{(\ell)} \}$
	\item $imp^{L\ref{lem:partitioning_sample_partition}} = imp^{(\ell)}$
\end{itemize}
In particular, this takes work $\left(n + m_S\right)\poly(K\log(nM))$ and depth $\poly(K\log(nM))$ and computes a subset $R^{sub,\ell}\subseteq R^{(\ell)}$ and a partition $Q^{(\ell)}_1 \sqcup Q^{(\ell)}_2 \sqcup \ldots \sqcup Q^{(\ell)}_{T_\ell} = R^{sub,\ell}$ satisfying

\begin{enumerate}
    \item $|R^{sub,\ell}| \geq \Omega(\frac{|R^{(\ell)}|}{K})$
    \item $T_{\ell} \leq \frac{|R^{sub,\ell}|}{100} + \poly(K \log(nM))$
    \item $|Q^{(\ell)}_t| = O(\log n)$ for every $t \in [T_\ell]$
    \item $\sum_{i=1}^{m_S} imp^{(\ell)}_i \cdot (1+\delta_{L\ref{lem:partitioning_sample_partition}})^{\sum_{t=1}^{T_\ell}\max(0,|Q^{(\ell)}_t \cap \hat{S}_i| - 1)} \leq \left(1 + \frac{1}{K \log^{40}(nM)} \right)\sum_{i=1}^{m_S} imp^{(\ell)}_i$
\end{enumerate}
We define $R^{(\ell + 1)} = R^{(\ell)} \setminus R^{sub,\ell}$ and $imp^{(\ell + 1)}_i = imp^{(\ell)}_i \cdot (1+\delta_{L\ref{lem:partitioning_sample_partition}})^{\sum_{t=1}^{T_\ell}\max(0,|Q^{(\ell)}_t \cap \hat{S}_i| - 1)}$ for every $i \in [m_S]$.
Also, let $\hat{Q}_1 \sqcup \hat{Q}_2 \sqcup \ldots \sqcup \hat{Q}_T$ by including for each $\ell \in [\ell_{max}]$ the partition $Q^{(\ell)}_1 \sqcup Q^{(\ell)}_2 \sqcup \ldots \sqcup Q^{(\ell)}_{T_\ell} = R^{sub,\ell}$. 

\medskip
\paragraph{Analyzing the properties} The third condition in the above list directly implies the second output property. Moreover,  the second condition in this list directly gives the following, which satisfies the first output property:
\[T \leq \sum_{\ell = 1}^{\ell_{max}} \left(\frac{|R^{sub,\ell}|}{100} + \poly(K \log(nM))\right) \leq \frac{|R|}{100} + \poly(K \log(nM))\]
Furthermore, the first condition implies that $|R^{(\ell + 1)}| \leq \left(1 - O(\frac{1}{K})\right)|R^{(\ell)}|$. Therefore, a simple induction shows that

\[|R^{(\ell_{max} + 1)}| \leq \left(1 - \Omega \left(\frac{1}{K}\right) \right)^{\ell_{max}}|R| \leq e^{- \Omega(\log^2(n))}|R| < 1\]
and therefore $R^{(\ell_{max} + 1)} = \emptyset$. It is easy to verify that $\hat{Q}_1 \sqcup \hat{Q}_2 \sqcup \ldots \sqcup \hat{Q}_T$ is a partition of $R \setminus R^{(\ell_{max} + 1)}$ and therefore of $R$, as needed.

What remains to be verified is that we have the third output property, i.e., that
\[\sum_{i=1}^{m_S} imp_i (1+\delta)^{\sum_{t = 1}^{T} \max(|S_i \cap \hat{Q}_t| - 1, 0)} \leq \left(1 + \frac{1}{\log^{30}(nM)} \right)\sum_{i=1}^{m_S} imp_i.\]
First, recall that $imp^{(1)} = imp$ and for every $\ell \in [\ell_{max}]$, we have
 \[\sum_{i=1}^{m_S} imp_i^{(\ell + 1)} \leq \left(1 + \frac{1}{K\log^{40}(nM)} \right)\sum_{i=1}^{m_S}imp^{(\ell)}_i.\]
Therefore, we get
\begin{align*}
\sum_{i=1}^{m_S} imp_i^{(\ell_{max} + 1)} &
\leq \left(1 + \frac{1}{K\log^{40}(nM)} \right)^{\ell_{max}} \sum_{i=1}^{m_S}imp_i \\
&\leq \left(1 + \frac{1}{K\log^{40}(nM)} \right)^{\lceil K \log^2(n) \rceil} \sum_{i=1}^{m_S}imp_i \\
&\leq  \left(1 + \frac{1}{\log^{30}(nM)} \right)\sum_{i=1}^{m_S}imp_i
\end{align*}
where the last inequality holds for $n$ being a sufficiently large constant.
Moreover, a simple induction shows that for every $\ell' \in \{0,1,\ldots,\ell_{max}\}$ and $i \in [m_S]$, it holds that
\[imp^{(\ell' + 1)}_i = imp_i(1 + \delta_{L\ref{lem:partitioning_sample_partition}})^{\sum_{\ell = 1}^{\ell'} \sum_{t=1}^{T_{\ell}} \max(|\hat{S}_i \cap Q^{(\ell)}_t)| - 1,0)}. \]
Thus, for every sufficienlty small $\delta \in (0,\delta_{L\ref{lem:partitioning_sample_partition}}]$, we get
\begin{align*}
\sum_{i=1}^{m_S} imp_i (1+\delta)^{\sum_{t = 1}^{T} \max(|\hat{S}_i \cap \hat{Q}_t| - 1, 0)} 
&\leq \sum_{i=1}^{m_S} imp_i (1+ \delta_{L\ref{lem:partitioning_sample_partition}})^{\sum_{t = 1}^{T} \max(|\hat{S}_i \cap \hat{Q}_t| - 1, 0)} \\
&= \sum_{i=1}^{m_S}  imp_i(1 + \delta_{L\ref{lem:partitioning_sample_partition}})^{\sum_{\ell = 1}^{\ell_{max}} \sum_{t=1}^{T_{\ell}} \max(|\hat{S}_i \cap Q^{(\ell)}_t)| - 1,0)}\\
&= \sum_{i=1}^{m_S} imp^{(\ell_{max} + 1)}_i \\
&\leq \left(1 + \frac{1}{\log^{30}(nM)} \right)\sum_{i=1}^{m_S}imp_i.
\end{align*}
\end{proof}

\subsubsection{\Cref{lem:partitioning_sample_partition}}
In proving the main partitioning lemma (\Cref{lem:partition_second}), we sequentially invoked $O(K\log^2 n)$ iterations of \Cref{lem:partitioning_sample_partition}. We now state and prove \Cref{lem:partitioning_sample_partition}. The proof itself relies on two other lemmas (\Cref{lem:partitioning_subset,lem:partitioning_subset_partitioning}), stated and proved later. 

\begin{lemma}
\label{lem:partitioning_sample_partition}
The following holds for every sufficiently small absolute constant $\delta > 0$: Let $n,M,K \in \mathbb{N}$ with $n \geq 1/\delta$. Let $R \subseteq [n]$ be a subset and $\{S_1,S_2,\ldots,S_m\}$ be a family of subsets of $R$ with $m \leq M$ and $|S_i| \leq K$ for every $i \in [m]$. Also, let $imp \in \mathbb{R}^m_{\geq 0}$. Then, there exists a deterministic parallel algorithm with work $\left(n + m\right)\poly(K\log(nM))$ and depth $\poly(K\log(nM))$ that computes a subset $R^{sub} \subseteq R$ and a partition $Q_1 \sqcup Q_2 \sqcup \ldots \sqcup Q_T = R_{sub}$ satisfying
\begin{enumerate}
    \item $|R^{sub}| = \Omega(\frac{|R|}{K})$
    \item $T \leq \frac{|R^{sub}|}{100} + \poly(K \log(nM))$
    \item $|Q_t| = O(\log n)$ for every $t \in [T]$
    \item $\sum_{i=1}^m imp_i \cdot (1+\delta)^{\sum_{t=1}^T\max(0,|Q_t \cap S_i| - 1)} \leq \left(1 + \frac{1}{K \log^{40}(nM)} \right)\sum_{i=1}^m imp_i$
\end{enumerate}
\end{lemma}
\begin{proof}
We first invoke the algorithm of \Cref{lem:partitioning_subset} to compute in $\left(n + m\right)\poly(K\log(nM))$ work and $\poly(K\log(nM))$ depth a subset $R^{sub}$ satisfying:

\begin{enumerate}
	    \item $|R^{sub}| = \Omega(\frac{|R|}{K})$
  	    \item $\sum_{i=1}^m imp_i \cdot (1+\delta_{L\ref{lem:partitioning_subset}})^{\max(0,|R^{sub}\cap S_i| - 1)} \leq 2\sum_{i=1}^m imp_i$
\end{enumerate}
Next, we invoke \Cref{lem:partitioning_subset_partitioning} (with $S_1 \cap R^{sub},S_2 \cap R^{sub},\ldots,S_m \cap R^{sub}$ being the family of input sets and $K_{L\ref{lem:partitioning_subset_partitioning}} = \lceil K^2\log^{100}(nM) \rceil$) to compute in $\left(n + m\right)\poly(K\log(nM))$ work and $\poly(K\log(nM))$ depth a partitioning $Q_1 \sqcup Q_2 \sqcup \ldots \sqcup Q_T =  R^{sub}$ satisfying

\begin{enumerate}
    \item $T \leq \frac{|R^{sub}|}{100} + \poly(K \log(nM))$
    \item $|Q_t| = O(\log n)$ for every $t \in [T]$
    \item $\sum_{i=1}^m imp_i \cdot \left(\sum_{t=1}^T\max(|Q_t \cap \hat{S}_i| -  1,0) \right) \leq \frac{1}{K^2\log^{100}(nM)}\sum_{i=1}^m imp_i$.
\end{enumerate}

Note that $R^{sub}$ and $Q_1 \sqcup Q_2 \sqcup \ldots \sqcup Q_T$ satisfy the first three output guarantees. We next show that for all sufficiently small $\delta > 0$, the fourth property holds as well, namely that

\[\sum_{i=1}^m imp_i \cdot (1+\delta)^{\sum_{t=1}^T\max(0,|Q_t \cap S_i| - 1)} \leq \left(1 + \frac{1}{K \log^{40}(nM)} \right)\sum_{i=1}^m imp_i.\]

For each $i \in [m]$, let $\gamma_i := (1+\delta_{L\ref{lem:partitioning_subset}})^{\sum_{t=1}^T\max(0,|Q_t\cap S_i| - 1)}$. Note that for sufficiently small $\delta > 0$, it holds that $(1+\delta)^{\sum_{t=1}^T\max(0,|Q_t \cap S_i|-1)|} \leq (\gamma_i)^{0.5}$ for every $i \in [m]$. In particular, we get

\begin{align*}
\sum_{i\in [m] \colon \gamma_i \geq K^2\log^{100}(nM)} imp_i (1+\delta)^{\sum_{t=1}^T \max(0,|Q_t \cap S_i| - 1)}
&\leq \sum_{i\in [m] \colon \gamma_i \geq K^2\log^{100}(nM)} imp_i (\gamma_i)^{-0.5}(1+\delta_{L\ref{lem:partitioning_subset}})^{\sum_{t=1}^T \max(0,|Q_t \cap S_i| - 1)}  \\
&\leq \frac{1}{K \cdot \log^{50}(nM)}\sum_{i = 1}^m imp_i (1+\delta_{L\ref{lem:partitioning_subset}})^{\sum_{t=1}^T \max(0,|Q_t \cap S_i| - 1)}  \\
&\leq \frac{1}{K \cdot \log^{50}(nM)} 2\sum_{i=1}^m imp_i  \\
&\leq \frac{1}{2K \log^{40}(nM)} \sum_{i=1}^m imp_i.
\end{align*}

Next, let $COL = \{i \in [m] \colon \text{there exists $t \in [T]$ with $|Q_t \cap \hat{S}_i| \geq 2$}\}$. For sufficiently small $\delta > 0$, we have
\begin{align*}
\sum_{i \in COL \colon \gamma_i < K^2\log^{100}(nM) } imp_i (1+\delta)^{{\sum_{t=1}^T \max(0,|Q_t \cap S_i| - 1)}} 
&\leq \sum_{i \in COL \colon \gamma_i < K^2\log^{100}(nM) } imp_i \sqrt{\gamma_i} \\
&\leq K \log^{50}(nM)\sum_{i=1}^m imp_i \left(\sum_{t=1}^T\max(|Q_t \cap \hat{S}_i| -  1,0) \right) \\
&\leq K \log^{50}(nM) \cdot \frac{1}{K^2\log^{100}(nM)}\sum_{i=1}^m imp_i \\
&\leq \frac{1}{2K \log^{40}(nM)} \sum_{i=1}^m imp_i.
\end{align*}
Moreover, it trivially holds that
\[\sum_{i \in [m] \setminus COL} imp_i (1+\delta)^{{\sum_{t=1}^T \max(0,|Q_t \cap S_i| - 1)}} \leq \sum_{i=1}^m imp_i.\]
Thus, for sufficiently small $\delta  > 0$, we get
\[\sum_{i = 1}^m imp_i (1+\delta)^{{\sum_{t=1}^T \max(0,|Q_t \cap S_i| - 1)}} \leq \left(1 + 2 \cdot \frac{1}{2K \log^{40}(nM)}\right) \sum_{i=1}^m imp_i = \left(1 + \frac{1}{K \log^{40}(nM)}  \right)\sum_{i=1}^m imp_i.\]
\end{proof}
The rest of this subsection is dedicated to stating and proving the two lemmas that were used in the above proof, namely \Cref{lem:partitioning_subset} and \Cref{lem:partitioning_subset_partitioning}. We discuss these in two different subsubsections. 

\subsubsection{\Cref{lem:partitioning_subset}}
The proof of \Cref{lem:partitioning_subset} uses a variant of a result of Ghaffari, Grunau, Rozhon~\cite{GGR2023Chernoff}, which we restate as \Cref{thm:focs_with_importance} and derive \Cref{cor:focs_cor2} out of it. After stating these results from \cite{GGR2023Chernoff}, we state and prove \Cref{lem:partitioning_subset}.

\begin{theorem}[Chernoff-like concentrations, Ghaffari, Grunau, Rozhon~\cite{GGR2023Chernoff}]
\label{thm:focs_with_importance}
There exists a constant $c > 0$ such that the following holds. Let $n,m,k \in \mathbb{N}$ be arbitrary, $\{S_1,S_2,\ldots,S_m\}$ be a family of subsets of $[n]$ and $imp \in \mathbb{R}^m_{\geq 0}$. Also, let $p \in [0,1]$ and $\Deltav \in \mathbb{R}_{>0}^m$. For each $i \in [m]$, let
\[\fail := c \exp \left(-(1/c) \min \left(\frac{\Deltai^2}{p|S_i|}, \Deltai, \frac{\Deltai k}{p |S_i|} \right)\right).\]
There exists a deterministic parallel algorithm with $\tilde{O}(\max(\sum_{i=1}^m |S_i|,n,m)\poly(k))$ work and $\poly(\log(nm),k)$ depth that outputs a subset $R^{sub} \subseteq [n]$ and a set $\Ibad \subseteq [m]$ with $\sum_{i \in \Ibad} imp_i \leq \sum_{i \in [m]} imp_i \cdot \fail$ such that for every $i \in [m] \setminus \Ibad$ it holds that
\[|p|S_i| - |R^{sub} \cap S_i|| \leq \Deltai.\]
\end{theorem}
From \Cref{thm:focs_with_importance}, by plugging values and rephrasing, we get the following two corollaries:

\begin{corollary}
\label{cor:focs_cor1}
There exists a constant $c > 0$ such that the following holds. Let $n,m,K \in \mathbb{N}$ be arbitrary, $\{S_1,S_2,\ldots,S_m\}$ be a family of subsets of $[n]$ with $|S_i| \leq K$ for every $i \in [m]$ and $imp \in \mathbb{R}^m_{\geq 0}$. Also, let $p \in [0,1]$ and $\Deltav \in \mathbb{R}_{>0}^m$ with $\Delta_i \geq 0.5p|S_i|$ for every $i \in [m]$. For each $i \in [m]$, let
\[\fail := c \exp \left(-(1/c) \Delta_i\right).\]
There exists a deterministic parallel algorithm with $\tilde{O}(\max(n,m)\poly(K))$ work and $\poly(\log(nm),K)$ depth that outputs a subset $R^{sub} \subseteq [n]$ and a set $\Ibad \subseteq [m]$ with $\sum_{i \in \Ibad} imp_i \leq \sum_{i \in [m]} imp_i \cdot \fail$ such that for every $i \in [m] \setminus \Ibad$ it holds that
\[|p|S_i| - |R^{sub} \cap S_i|| \leq \Deltai.\]
\end{corollary}

\begin{corollary}
\label{cor:focs_cor2}
There exists a constant $c > 0$ such that the following holds.
Let $n,m,K \in \mathbb{N}$ be arbitrary, $p \in [0,1]$ and $R \subseteq [n]$. Let $\mathcal{C}$ be a collection of at most $m$ triples $C= (S,\Delta,imp(C))$ with $S \subseteq R$, $|S| \leq K$, $\Delta \geq 0.5p|S|$ and $imp(C) \in \mathbb{R}_{\geq 0}$. For each triple $C =(S,\Delta,imp(C)) \in \mathcal{C}$, we define

\[prob^{bad}_C := ce^{-\Delta/c}\]

There exists a deterministic parallel algorithm with  $\tilde{O}(\max(n,m)\poly(K))$ work and $\poly(\log(nm),K)$ depth that outputs a subset $R^{sub} \subseteq R$ satisfying that $\mathcal{C}^{bad} \subseteq \mathcal{C}$ with $\sum_{C \in \mathcal{C}^{bad}} imp(C) \leq \sum_{C \in \mathcal{C}} imp(C)\cdot prob^{bad}_C $ and for every $(S,\Delta,imp) \in \mathcal{C} \setminus \mathcal{C}^{bad}$, it holds that
\[|p|S| - |R^{sub} \cap S||\leq \Delta.\]
\end{corollary}

Having these helper tools from prior work, we now go back to proving \Cref{lem:partitioning_subset}, by using \Cref{cor:focs_cor2}.

\begin{lemma}
\label{lem:partitioning_subset}
The following holds for every sufficiently small absolute constant $\delta > 0$: Let $n,m,K \in \mathbb{N}$. Let $R \subseteq [n]$ be a subset and $\{S_1,S_2,\ldots,S_m\}$ be a family of subsets of $R$  and $|S_i| \leq K$ for every $i \in [m]$. Also, let $imp \in \mathbb{R}^m_{\geq 0}$. Then, there exists a deterministic parallel algorithm with work $\left(n + m\right)\poly(K\log(2nm))$ and depth $\poly(K\log(2nm))$ that computes a subset $R^{sub} \subseteq R$ satisfying
\begin{enumerate}
    \item $|R^{sub}| = \Omega(\frac{|R|}{K})$
    \item $\sum_{i=1}^m imp_i \cdot (1+\delta)^{\max(0,|R^{sub}\cap S_i| - 1)} \leq 2\sum_{i=1}^m imp_i$
\end{enumerate}
\end{lemma}
\begin{proof}
In the following, we denote by $\delta' > 0$ a sufficiently small constant (not dependent on $\delta$). If $|R| \leq (1/\delta') K$, then  $\Omega(\frac{|R|}{K}) = \Omega(1)$, and therefore we can return $R^{sub} = R$ if $|R| \leq 1$ and if $|R| \geq 2$, we can define $R^{sub}  = \{j\}$ where $j$ is the smallest number contained in $R$.
Thus, we will now assume that $|R| \geq (1/\delta') K$.

Let $\mathcal{C}$ be the collection containing the following $mK + 1$ triples:

\begin{enumerate}
	\item $(S_i,k,imp_i(1+\delta')^{k})$ for every $i \in [m]$, $k \in [K]$.
	\item $(R,0.5|R|/K, (10/\delta')\sum_{i=1}^m imp_i)$
\end{enumerate}

By invoking \cref{cor:focs_cor2} with $\mathcal{C}$ and $p = \frac{1}{1.1K}$, we can compute  with work $\left(n + m\right)\poly(K\log(2nm))$ and depth $\poly(K\log(2nm))$ a subset $R^{sub} \subseteq R$ and $\mathcal{C}^{bad} \subseteq \mathcal{C}$ such that for $c > 0$ being the constant of \cref{cor:focs_cor2} and $\delta' \ll 1/c$ being sufficiently small, we get

\begin{align*}
\sum_{C \in \mathcal{C}^{bad}} imp(C) &\leq 
\sum_{i=1}^m\sum_{k=1}^K imp_i (1+\delta')^{k} \cdot ce^{-k/c} + (10/\delta') (\sum_{i=1}^m imp_i)ce^{-(0.5|R|/K)/c} \\
&\leq\sum_{i=1}^{m} imp_i \left((10/\delta')ce^{-0.5(1/\delta')/c} + \sum_{k=1}^{K}(1+\delta')^{k} \cdot ce^{-k/c}\right) \\
&\leq (1/\delta')\sum_{i=1}^m imp_i.
\end{align*}
Moreover, if $(S_i,k,imp_i(1+\delta')^k) \notin  \mathcal{C}^{bad}$ for some $i \in [m]$ and $k \in [K]$, then

\[  |R^{sub} \cap S_i| \leq p|S_i| + k \leq k+\frac{1}{1.1} < k + 1\]

and if $(R,0.5|R|/K, (10/\delta')\sum_{i=1}^m imp_i) \notin \mathcal{C}^{bad}$, then

\[|R^{sub}| \geq p|R| - 0.5|R|/K \geq 0.1|R|/K.\]

As $\sum_{C \in \mathcal{C}^{bad}} imp(C) \leq (1/\delta')\sum_{i=1}^m imp_i$, we have $(R,0.5|R|/K, (10/\delta')\sum_{i=1}^m imp_i) \notin \mathcal{C}^{bad}$ and therefore $|R^{sub}| \geq 0.1|R|/K = \Omega(|R|/K)$.

For every $i \in [m]$ with $|R^{sub}\cap S_i| \geq 2$, it holds that $(S_i,|R^{sub}\cap S_i| - 1,imp_i(1+\delta')^{|R^{sub}\cap S_i| -1}) \in  \mathcal{C}^{bad}$ and therefore

\begin{align*}
  \sum_{i =1}^m imp_i \cdot (1+\delta')^{\max(0,|R^{sub}\cap S_i| - 1)}  
&\leq 2\sum_{i=1}^m imp_i + \sum_{i \in [m] \colon |R^{sub} \cap S_i| \geq 2}  imp_i \cdot (1+\delta')^{|R^{sub}\cap S_i| - 1} \\
& \leq 2\sum_{i=1}^m imp_i  + \sum_{C \in \mathcal{C}^{bad}} imp(C) \\
&\leq (2/\delta')\sum_{i=1}^m imp_i.
\end{align*}

However, this implies that for $\delta \ll \delta'$ being sufficiently small, we indeed have

 \[\sum_{i=1}^m imp_i \cdot (1+\delta)^{\max(0,|R^{sub}\cap S_i| - 1)} \leq 2\sum_{i=1}^m imp_i.\]
\end{proof}

\subsubsection{\Cref{lem:partitioning_subset_partitioning}}
The other lemma that we used in proving \Cref{lem:partitioning_sample_partition} is \Cref{lem:partitioning_subset_partitioning}, which we state and prove here. We note that this proof itself uses \Cref{lem:sequential_collisions} stated and proved later in this subsubsection, and that is the end of the top-down hierarchy of the proofs of our weighted partitioning result, hence finishing the proof of \Cref{thm:partition_with_isolation}.

\begin{lemma}
\label{lem:partitioning_subset_partitioning}
Let $n,K \in \mathbb{N}$ with $n,K$ being larger than a sufficiently large constant. Let $R^{sub} \subseteq [n]$ be a subset and $\{\hat{S}_1,\hat{S}_2,\ldots,\hat{S}_m\}$ be a family of subsets of $R^{sub}$ with $|\hat{S}_i| \leq K$ for every $i \in [m]$. Also, let $imp \in \mathbb{R}^m_{\geq 0}$. Then, there exists a deterministic parallel algorithm with work $\left(n + m  \right)\poly(K\log(nm))$ and depth $\poly(K\log(nm))$ that computes a partitioning $Q_1 \sqcup Q_2 \sqcup \ldots \sqcup Q_T =  R^{sub}$ satisfying
\begin{enumerate}
    \item $T \leq \frac{|R^{sub}|}{100} + \poly(K)$
    \item $|Q_t| = O(\log n)$ for every $t \in [T]$
    \item $\sum_{i=1}^m imp_i \cdot \left(\sum_{t=1}^T\max(|Q_t \cap \hat{S}_i| -  1,0) \right) \leq \frac{1}{K}\sum_{i=1}^m imp_i$
\end{enumerate}
\end{lemma}

\begin{proof}
First, note that if $|R^{sub}| \leq K^4$, then we can simply return the partition into singletons. Thus, it remains to consider the case that $|R^{sub}| \geq K^4$.
Let $P_1 \sqcup P_2 \sqcup \ldots P_{T_P} = R^{sub}$ be an arbitrary partition satisfying $|P_t| \geq K^4$ and $|P_t| \leq \poly(K)$ for every $t \in [T_P]$. Note that such a partition can be efficiently computed in parallel using a prefix-sum algorithm.

In parallel, we invoke for each $t \in [T_P]$ the algorithm of \cref{lem:sequential_collisions} with input

\begin{enumerate}
    \item $T^{(t)}_{L \ref{lem:sequential_collisions}} = \lfloor\frac{|P_t|}{100} \rfloor$
    \item $R^{(t)}_{L \ref{lem:sequential_collisions}} = P_t$
    \item $\mathcal{C}^{(t)}_{L \ref{lem:sequential_collisions}} = \{ (\hat{S}_i \cap P_t,imp_i) \colon i \in [m] \text{ with }\hat{S}_i \cap P_t \neq \emptyset\}$
\end{enumerate}
to compute a partition $Q^{(t)}_1 \sqcup Q^{(t)}_2 \sqcup \ldots \sqcup Q^{(t)}_{T^{(t)}_{L \ref{lem:sequential_collisions}}}$ satisfying

\begin{enumerate}
    \item $\sum_{i = 1}^{m} imp_i \sum_{t'=1}^{T^{(t)}_{L \ref{lem:sequential_collisions}}} {|\hat{S}_i \cap Q^{(t)}_{t'}| \choose 2} \leq \frac{2}{T^{(t)}_{L \ref{lem:sequential_collisions}}} \sum_{i = 1}^{m} imp_i {|\hat{S}_i \cap P_t| \choose 2} \leq \frac{1}{K^3} \sum_{i = 1}^{m} imp_i {|\hat{S}_i \cap P_t| \choose 2}$
    \item $|Q^{(t)}_{t'}| \leq \max(1,10\log(n)|P_t|/T^{(t)}_{L \ref{lem:sequential_collisions}}) = O(\log n)$ for every $t' \in [T^{(t)}_{L \ref{lem:sequential_collisions}}]$
\end{enumerate}
We then define $Q_1 \sqcup Q_2 \sqcup \ldots \sqcup Q_T$ as the union of the partitions $Q^{(t)}_1 \sqcup Q^{(t)}_2 \sqcup \ldots \sqcup Q^{(t)}_{T^{(t)}_{L \ref{lem:sequential_collisions}}}$ for every $t \in [T_p]$. Note that $T = \sum_{t=1}^{T_P}T^{(t)}_{L \ref{lem:sequential_collisions}} \leq \sum_{t=1}^{T_P} \frac{|P_t|}{100} \leq \frac{|R^{sub}|}{100}$. Also,

\begin{align*}
    \sum_{i=1}^m imp_i \cdot \left(\sum_{t=1}^T\max(|Q_t \cap \hat{S}_i| -  1,0) \right) 
    &\leq \sum_{i=1}^{m} imp_i \sum_{t=1}^T {|\hat{S}_i \cap Q_t| \choose 2} \\
    &= \sum_{t=1}^{T_P}\sum_{i = 1}^{m} imp_i \sum_{t'=1}^{T^{(t)}_{L \ref{lem:sequential_collisions}}} {|\hat{S}_i \cap Q^{(t)}_{t'}| \choose 2} \\
    &\leq  \sum_{t=1}^{T_p}\frac{1}{K^3} \sum_{i = 1}^{m} imp_i {|\hat{S}_i \cap P_t| \choose 2} \\
    &\leq \frac{1}{K^3} \sum_{i = 1}^{m} imp_i {|\hat{S}_i| \choose 2} \\
    &\leq \frac{1}{K^3} \cdot K^2 \sum_{i = 1}^{m} imp_i \\
    &\leq  \frac{1}{K}\sum_{i = 1}^{m} imp_i.
\end{align*}

It remains to discuss the work and depth of the algorithm. First, note that we can compute $\mathcal{C}^{(t)}_{L \ref{lem:sequential_collisions}}$ for every $t \in [T_P]$ with work $\left(n + m  \right)\poly(K\log(nm))$ and depth $\poly(K\log(nm))$, as $|\hat{S}_i| \leq K$ for every $i \in [m]$.   Using that $|R^{(t)}_{L \ref{lem:sequential_collisions}}| \leq \poly(K)$  for every $t \in [T_P]$, the algorithm of \cref{lem:sequential_collisions} computes the partition $Q^{(t)}_1 \sqcup Q^{(t)}_2 \sqcup \ldots \sqcup Q^{(t)}_{T^{(t)}_{L \ref{lem:sequential_collisions}}}$ in $(|\mathcal{C}^{(t)}_{L \ref{lem:sequential_collisions}}| + 1)\poly(K\log(mn))$ work and $\poly(K\log(mn))$ depth.
As $\sum_{t=1}^{T_P} |\mathcal{C}^{(t)}_{L \ref{lem:sequential_collisions}}| \leq m \cdot K$, we can therefore conclude that the algorithm has work $\left(n + m  \right)\poly(K\log(nm))$ and depth $\poly(K\log(nm))$.
\end{proof}

Finally, in the above proof of \Cref{lem:partitioning_subset_partitioning}), we used \Cref{lem:sequential_collisions}, which we state and prove here.

\begin{lemma}
\label{lem:sequential_collisions}
Let $n,M,T \in \mathbb{N}$ with $M \geq 3$. Let $R \subseteq [n]$ and $\mathcal{C}$ be a collection of at most $M$ tuples $(S,imp_S)$ with $S \subseteq R$ and $imp_S \in \mathbb{R}_{\geq 0}$.
There exists a deterministic parallel algorithm with work $(T + |\mathcal{C}|)\poly((|R|+1) \log(Mn))$ and depth $\poly((|R|+1) \log(Mn))$ that computes a partition $Q_1 \sqcup Q_2 \sqcup \ldots \sqcup Q_T = R$ satisfying

\begin{enumerate}
    \item $\sum_{(S,imp_S) \in \mathcal{C}}  imp_S \sum_{t=1}^T {|S \cap Q_t|\choose 2} \leq \frac{2}{T}\sum_{(S,imp_S) \in \mathcal{C}} imp_S{|S| \choose 2}$
    \item For every $t \in [T] \colon |Q_t| \leq \max(1,10\log(n)|R|/T)$
\end{enumerate}

\end{lemma}

\begin{proof}

First, note that if $T \geq |R|$, then we can simply output a partition into singletons. Thus, we now consider the case that $T < |R|$. Also, note that we can assume without loss of generality that for every $(S,imp_S) \in \mathcal{C}$, it holds that $|S| \geq 2$.
Let $R = \{j_1,j_2,\ldots,j_{|R|}\}$ and $R_{\ell} = \{j_1,j_2,\ldots,j_\ell\}$ for every $\ell \in \{0,1,\ldots,|R|\}$.
Consider each $\ell \in \{0,1,\ldots,|R|\}$ and an arbitrary $\rho^{(\ell)} \colon R_{\ell} \mapsto [T]$. Also, we use the notation $(\rho^{(\ell)})^{-1}$ to denote the ``inverse function'' where for each $t\in T$ we define $(\rho^{(\ell)})^{-1}(t)=\{j\in R_{\ell} \,|\, \rho^{(\ell)}(j)=t\}$. We define 

\begin{itemize}
    \item $\Phi_{\ell,S}(\rho^{(\ell)}) = \sum_{t=1}^T {|S \cap (\rho^{(\ell)})^{-1}(t)|\choose 2} + \frac{1}{T} \left({|S| \choose 2} -  {|S \cap R_\ell| \choose 2}\right)$ for any $S \subseteq R$ ,
    \item $\Phi_{\ell,t}(\rho^{(\ell)}) =  2^{|(\rho^{(\ell)})^{-1}(t)|} \cdot \left( 1 + \frac{1}{T}\right)^{|R| -\ell}$ for any $t \in [T]$, and
    \item $\Phi_\ell(\rho^{(\ell)}) = \frac{\sum_{(S,imp_S) \in \mathcal{C}} imp_S \Phi_{\ell,S}(\rho^{(\ell)})}{\sum_{(S,imp_S) \in \mathcal{C}} imp_S \Phi_{0,S}(\rho^{(0)})} + \frac{1}{T}\sum_{t=1}^T \frac{\Phi_{\ell,t}(\rho^{(\ell)})}{\Phi_{0,t}(\rho^{(0)})}$, where $\rho^{(0)}$ is the empty function from $\emptyset$ to $[T]$.
\end{itemize}

For every $\ell \in \{0,1,\ldots,|R| - 1\}$, $\rho^{(\ell)} \in  R_{\ell} \mapsto [T]$ and $t \in [T]$, we define $\rho^{(\ell)} + t$ as the function from $R_{\ell + 1}$ to $[T]$ that agrees with $\rho^{(\ell)}$ on $R_{\ell}$ and maps $j_{\ell + 1}$ to $t$.

The algorithm runs in $|R|$ iterations. At the beginning of the $\ell$-th iteration, the elements $j_1,j_2,\ldots,j_{\ell-1}$ have each been assigned to one of the $T$ parts and we denote by $\rho^{(\ell - 1)} \colon R_{\ell -1} \mapsto [T]$ the corresponding mapping. In the $\ell$-th iteration, the algorithm first computes $t_\ell = \arg \min_{t \in [T]} \Phi_{\ell}(\rho^{(\ell -1)} + t)$ and then assigns $j_\ell$ to the part $Q_{t_{\ell}}$.
Let $\rho^{|R|} \colon R  \mapsto [T]$ be the resulting mapping. Note that this mapping can indeed be computed with work $(1 + |\mathcal{C}|)\poly((|R|+1) \log(Mn))$ and depth $\poly((|R|+1) \log(Mn))$ and it naturally defines a partition of $R$ into $T$ parts, namely by defining $Q_t = (\rho^{|R|})^{-1}(t)$ for every $t \in [T]$. It remains to show that the desired partition satisfies the desired properties.
Note that $\Phi_0(\rho^{(0)}) = 2$ and below we will argue that $\Phi_{|R|}(\rho^{(|R|)}) \leq 2$. In particular,

\[\sum_{(S,imp_S) \in \mathcal{C}} imp_S \Phi_{|R|,S}(\rho^{(|R|)}) \leq 2 \sum_{(S,imp_S) \in \mathcal{C}} imp_S \Phi_{0,S}(\rho^{(0)}),\]

which implies that the first property holds because

\[\sum_{(S,imp_S) \in \mathcal{C}} imp_S \Phi_{|R|,S}(\rho^{(|R|)}) = \sum_{(S,imp_S) \in \mathcal{C}} imp_S\sum_{t=1}^T {|S \cap (\rho^{(|R|)})^{-1}(t)|\choose 2} = \sum_{(S,imp_S) \in \mathcal{C}} imp_S\sum_{t=1}^T {|S \cap Q_t|\choose 2}\]

and 

\[2 \sum_{(S,imp_S) \in \mathcal{C}} imp_S \Phi_{0,S}(\rho^{(0)}) = \frac{2}{T}\sum_{(S,imp_S) \in \mathcal{C}} imp_S{|S| \choose 2}.\]

Next, fix some $t \in [T]$. The fact that $\Phi_{|R|}(\rho^{(|R|)}) \leq 2$ directly implies

\[\Phi_{|R|,t}(\rho^{(|R|)}) \leq 2T \Phi_{0,t}(\rho^{(0)}) = 2T\left(1 + \frac{1}{T} \right)^{|R|} \leq 2Te^{|R|/T},\]

But $\Phi_{|R|,t}(\rho^{(|R|)}) = 2^{|Q_t|}$ and therefore

\[|Q_t| \leq \log_2(2|T|e^{|R|/T}) \leq 1 + \log_2(|T|) + 2|R|/T \leq \max(1,10\log(n)|R|/T),\]

as needed. All that is left to show is that $\Phi_{|R|}(\rho^{(|R|)}) \leq 2$. This is a simple consequence of the facts that $\Phi_0(\rho^{(0)}) = 2$ and for every $\ell \in \{0,1,\ldots,|R| -1\}$ and $\rho^{(\ell)} \colon R_{\ell} \mapsto  [T]$, it holds that
\[\frac{1}{T}\sum_{t=1}^T \Phi_{\ell+1}(\rho^{(\ell)} + t) = \Phi_\ell(\rho^{(\ell)}),\]

which we will show next.

Fix any $S \subseteq R$. If $j_{\ell + 1} \notin S$, then $\Phi_{\ell+1,S}(\rho^{(\ell)} + t) = \Phi_{\ell,S}(\rho^{(\ell)})$ for every $t \in [T]$. If $j_{\ell + 1} \in S$, then for any $t \in [T]$, it holds that

\begin{align*}
\Phi_{\ell+1,S}(\rho^{(\ell)} + t) - \Phi_{\ell,S}(\rho^{(\ell)}) &= {{|S \cap (\rho^{(\ell)})^{-1}(t)| + 1\choose 2} } - {{|S \cap (\rho^{(\ell)})^{-1}(t)|\choose 2} } \\
&+ \frac{1}{T}\left({|S \cap R_\ell| \choose 2} -  {|S \cap R_{\ell+ 1}| \choose 2} \right) \\
&= |S \cap (\rho^{(\ell)})^{-1}(t)|   - \frac{1}{T}|S \cap R_\ell|.
\end{align*}

Therefore 

\begin{align*}
    \frac{1}{T}\sum_{t=1}^T\Phi_{\ell+1,S}(\rho^{(\ell)} + t) &= \Phi_{\ell,S}(\rho^{(\ell)}) + \frac{1}{T}\sum_{t=1}^T \left(\Phi_{\ell+1,S}(\rho^{(\ell)} + t) - \Phi_{\ell,S}(\rho^{(\ell)}) \right) \\
    &= 
    \Phi_{\ell,S}(\rho^{(\ell)}) + \frac{1}{T}\sum_{t=1}^T \left(|S \cap (\rho^{(\ell)})^{-1}(t)|  - \frac{1}{T}|S \cap R_\ell|\right) \\
    &= \Phi_{\ell,S}(\rho^{(\ell)}).
\end{align*}

Next, consider some arbitrary $t' \in [T]$. Note that

\begin{align*}
    \frac{\frac{1}{T}\sum_{t=1}^T\Phi_{\ell+1,t'}(\rho^{(\ell)} + t)}{\Phi_{\ell,t'}(\rho^{(\ell)})} = \frac{\frac{1}{T} (2 + (T-1)\cdot 1)}{1 + \frac{1}{T}} = 1
\end{align*}

Thus, we can indeed conclude that 

\[\frac{1}{T}\sum_{t=1}^T \Phi_{\ell+1}(\rho^{(\ell)} + t) = \Phi_\ell(\rho^{(\ell)}).\]
\end{proof}

\newpage

\subsection{Creating a recursion, via the partitioning and MWU}
Next, we describe the lemma that uses our weighted partitioning theorem (\Cref{thm:partition_with_isolation}) along with the multiplicative weights update method to create a smaller weighted set balancing problem, which we solve recursively. We later use this lemma to present our recursive solution for the weighted set balancing problem.

\begin{lemma}
    \label{lem:mwu_weighted}
     There exists an absolute constant $\delta \in (0,0.5]$ such that the following holds.
     Let $n,m,\Delta \in \mathbb{N}$ with $m \geq 2$, $\Delta \geq n$ and $n \geq (\log 3m\Delta)^{1/\delta}$, $A \in \{-\Delta,-\Delta + 1, \ldots, -1,0,1,\ldots,\Delta\}^{m \times n}$ and $imp \in \mathbb{R}_{\geq 0}^m$. There exists a deterministic parallel algorithm with work $(n + m + nnz(A))\cdot\poly(\log(\Delta n m))$ and depth $\poly(\log(\Delta n m))$ that computes a vector $\chi \in \{-1,1\}^n$, a partition $Q_1 \sqcup Q_2 \sqcup \ldots \sqcup Q_T = [n]$ with $T \leq n/2$ and for every $i \in [m]$ a subset  $COL_i \subseteq [n]$ satisfying
    \begin{enumerate}
        \item $\sum_{i=1}^m imp_i (1+\delta)^{|COL_i|} \leq (1 + \frac{1}{10 \log^2(n)}) \sum_{i=1}^m imp_i $, and  \\
        \item  for every $i \in [m]$, we have $\sum_{t=1}^T \left(\sum_{j \in Q_t \setminus COL_i} a_{ij} \chi_j \right)^2 \leq (1 + \frac{1}{10\log^2 (n)}) \sum_{j \in [n] \setminus COL_i} a^2_{ij}$.
    \end{enumerate}
\end{lemma}
\begin{proof}

Notice that $ |a_{ij}| \in \{0,1,2,\ldots,\Delta\}$. First, we compute a bucketing of the coefficients $a_{ij}$ in each constraint $i\in [m]$ such that coefficients in the same bucket are nearly the same, up to a $1 + \frac{1}{\log^3(n)}$ factor. That is, for every $i \in [m]$, we compute a bucketing

\[\{j \in [n] \colon a_{ij} \neq  0\} = B_{i,1} \sqcup B_{i,2} \sqcup \ldots \sqcup B_{i,\ell_{max}}\]
for some $\ell_{max} = O(\log (\Delta) \log^3(n))$, such that for every $\ell \in [\ell_{max}]$, it holds that
\[\frac{\max_{j \in B_{i,\ell}} |a_{ij}|}{\min_{j \in B_{i,\ell}} |a_{ij}|} \leq 1 + \frac{1}{\log^3(n)}.\]
We refer to bucket $B_{i,\ell}$ as \textit{small} if $|B_{i,\ell}| \leq \log(nm \Delta)^{50}$, as \textit{large} if $|B_{i,\ell}| \geq \log(nm \Delta)^{90}$, and otherwise we refer to it as a \textit{medium bucket}. Accordingly, we define $S_i$, $M_i$, and $L_i$ to be the union of all items in small, medium, and large buckets, correspondingly. That is, we set
\begin{itemize}
	\item $S_i = \left(\bigcup_{\ell \in [\ell_{max}] \colon \text{$B_{i,\ell}$ is a small bucket}} B_{i,\ell} \right) $ 
	\item $M_i = \left(\bigcup_{\ell \in [\ell_{max}] \colon \text{$B_{i,\ell}$ is a medium bucket}} B_{i,\ell} \right) $
	\item $L_i = \left(\bigcup_{\ell \in [\ell_{max}] \colon \text{$B_{i,\ell}$ is a large bucket}} B_{i,\ell} \right) $.
\end{itemize}

\bigskip
A major ingredient in the proof is summarized by the following claim:

\begin{claim}\label{clm:in-weighted-recursion-creator}
We can compute a vector $\chi \in \{-1,1\}^n$ and a partition $Q_1 \sqcup Q_2 \sqcup \ldots \sqcup Q_T = [n]$ with $T \leq n/2$ satisfying the following guarantees:

\begin{enumerate}
		\item $\sum_{i=1}^m imp_i (1+\delta)^{\sum_{t =1}^T \max(0,|Q_t \cap(S_i \cup M_i)| - 1)} \leq \left(1 + \frac{1}{10 \log^2(n)} \right)\sum_{i=1}^m imp_i $
	\item $\sum_{t=1}^T \left(\sum_{j \in Q_t \cap(M_i \cup L_i)} a_{ij} \chi_j \right)^2 \leq \left(1 + \frac{100}{\log^3 n} \right) \sum_{j \in M_i \cup L_i} a^2_{ij}$
	\item $|Q_t| = O(\log n)$ for every $t \in [T]$
\end{enumerate}
\end{claim}

The remaining part of the proof is structured in two components. In the first component, we argue that the three guarantees in \Cref{clm:in-weighted-recursion-creator}, together with some simple definitions for $COL_i$, imply the output guarantees of \cref{lem:mwu_weighted}. Later, in the second component, we prove \Cref{clm:in-weighted-recursion-creator}.

\medskip
\paragraph{Component I of the proof--concluding \Cref{lem:mwu_weighted} from \Cref{clm:in-weighted-recursion-creator}} For every $i \in [m]$ and $t \in [T]$, we define $S_{i,t} = S_i \cap Q_t$, $M_{i,t} = M_i \cap Q_t$ and $L_{i,t} = L_i \cap Q_t$. We also define $COL_i \subseteq S_i$ by putting $\min(|S_{i,t}|,|S_{i,t}| + |M_{i,t}| - 1)$ arbitrary elements of $S_{i,t}$ in $COL_i$, for each $t \in [T]$. We argue that with this definition, by the three properties in \Cref{clm:in-weighted-recursion-creator}, the sets $COL_i$ satisfy the output guarantees of \cref{lem:mwu_weighted}.

Note that the first property of \cref{lem:mwu_weighted}, namely that $\sum_{i=1}^m imp_i (1+\delta)^{|COL_i|} \leq (1 + \frac{1}{10 \log^2(n)}) \sum_{i=1}^m imp_i $ directly follows from the way we defined $COL_i$ for every $i \in [m]$ and our assumption $$\sum_{i=1}^m imp_i (1+\delta)^{\sum_{t =1}^T \max(0,|S_{i,t} \cup M_{i,t}| - 1)} \leq (1 + \frac{1}{10 \log^2(n)}) \sum_{i=1}^m imp_i. $$ Thus, it remains to show that the second property of \cref{lem:mwu_weighted} is satisfied.

Fix some $i \in [m]$. We have to show that

\[\sum_{t=1}^T \left(\sum_{j \in Q_t \setminus COL_i} a_{ij} \chi_j \right)^2 \leq \left(1 + \frac{1}{10\log^2 (n)} \right) \sum_{j \in [n] \setminus COL_i} a^2_{ij}.\]
We define the set of elements in \textit{small} buckets \textit{without collisions} $SWC_i = S_i \setminus COL_i$ and $SWC_{i,t} = SWC_i \cap Q_t$ for every $t \in [T]$.
From the way $COL_i$ is defined, it holds that $|SWC_{i,t}| \leq 1$ and $ |SWC_{i,t}|\cdot |M_{i,t}| = 0$ for every $t \in [T]$. Using both of these properties, together with our assumption
\[\sum_{t=1}^T \left(\sum_{j \in M_{i,t} \cup L_{i,t}} a_{ij} \chi_j \right)^2 \leq \left(1 + \frac{100}{\log^3 n} \right) \sum_{j \in M_i \cup L_i} a^2_{ij},\]
we can deduce that
\begin{align*}
	& &&\sum_{t=1}^T \left(\sum_{j \in Q_t \setminus COL_i} a_{ij} \chi_j \right)^2 \\
        &= &&\sum_{t=1}^T \left( \left(\sum_{j \in SWC_{i,t}} a_{ij}\chi_j \right)  + \sum_{j \in M_{i,t} \cup L_{i,t}} a_{ij} \chi_j\right)^2 \\
         &= &&\sum_{t=1}^T \left( \sum_{j \in SWC_{i,t}} a_{ij}\chi_j\right)^2 
         + 2\sum_{t=1}^T \sum_{j \in SWC_{i,t}}\sum_{j '\in M_{i,t} \cup L_{i,t}} a_{ij}a_{ij'}\chi_j \chi_{j'} 
         + \sum_{t=1}^T \left(\sum_{j \in M_{i,t} \cup L_{i,t}} a_{ij} \chi_j \right)^2  \\
        &\leq &&\sum_{j \in SWC_i} a_{ij}^2 
        + 2\sum_{t=1}^T \sum_{j \in SWC_{i,t}}\sum_{j '\in L_{i,t}} |a_{ij}a_{ij'}| 
         + \left(1 + \frac{100}{\log^3 n} \right) \sum_{j \in M_i \cup L_i} a^2_{ij} \\
        &\leq &&\left(1 + \frac{100}{\log^3 n} \right) \sum_{j \in [n] \setminus COL_i} a^2_{ij} + 2\sum_{t=1}^T \sum_{j \in SWC_{i,t}}\sum_{j '\in  L_{i,t}} |a_{ij}a_{ij'}|. 
\end{align*}

As we want to show that $\sum_{t=1}^T \left(\sum_{j \in Q_t \setminus COL_i} a_{ij} \chi_j \right)^2 \leq \left(1 + \frac{1}{10\log^2 (n)} \right) \sum_{j \in [n] \setminus COL_i} a^2_{ij}$, and we can assume that $n$ is sufficiently large, it suffices to show that $\sum_{t=1}^T \sum_{j \in SWC_{i,t}}\sum_{j '\in  L_{i,t}} |a_{ij}a_{ij'}| \leq O\left(\frac{1}{\log^3(n)}\right)\sum_{j \in SWC_i \cup L_i} a_{ij}^2$, which we will do next.
Consider some $\ell,\ell' \in [\ell_{max}]$ such that $B_{i,\ell}$ is  a small bucket and $B_{i,\ell'}$ is a big bucket. 
We next show that
\[\sum_{t=1}^T   \sum_{j \in SWC_{i,t}  \cap B_{i,\ell}}\sum_{j' \in L_{i,t} \cap B_{i,\ell'}} |a_{ij}a_{ij'}| \leq O\left(\frac{1}{\log^6(n)\log(\Delta)}\right) \sum_{j \in \left(SWC_i \cap B_{i,\ell} \right) \cup B_{i,\ell'}} a_{ij}^2.\]
Note that this suffices because

\begin{align*}
& &&\sum_{t=1}^T \sum_{j \in SWC_{i,t}}\sum_{j' \in L_{i,t}}|a_{ij}a_{ij'}| \\
&=&&\sum_{\ell  = 1}^{\ell_{max}}\sum_{\ell' =1}^{\ell_{max}} \sum_{t=1}^T  \sum_{j \in SWC_{i,t} \cap B_{i,\ell}}\sum_{j' \in L_{i,t} \cap B_{i,\ell'}}|a_{ij}a_{ij'}| \\
&\leq &&\sum_{\ell  = 1}^{\ell_{max}}\sum_{\ell' =1}^{\ell_{max}} O\left(\frac{1}{\log^6(n)\log(\Delta)}\right) \sum_{j \in \left(SWC_i \cap B_{i,\ell} \right) \cup (L_i \cap B_{i,\ell'})} a_{ij}^2 \\
&\leq &&\left(\sum_{\ell  = 1}^{\ell_{max}}\sum_{\ell' =1}^{\ell_{max}} O\left(\frac{1}{\log^6(n)\log(\Delta)}\right) \sum_{j \in SWC_i \cap B_{i,\ell}} a_{ij}^2 \right)
+ \sum_{\ell  = 1}^{\ell_{max}}\sum_{\ell' =1}^{\ell_{max}} O\left(\frac{1}{\log^6(n)\log(\Delta)}\right) \sum_{j \in L_i \cap B_{i,\ell'}} a_{ij}^2 \\
&\leq &&\ell_{max} \cdot O\left(\frac{1}{\log^6(n)\log(\Delta)}\right)\left(\left(\sum_{\ell  = 1}^{\ell_{max}} \sum_{j \in SWC_i \cap B_{i,\ell}} a_{ij}^2\right) +\sum_{\ell'  = 1}^{\ell_{max}}  \sum_{j \in L_i \cap B_{i,\ell'}} a_{ij}^2\right) \\
&\leq &&O(\log^3 (n) \log(\Delta)) \cdot O\left(\frac{1}{\log^6(n)\log(\Delta)}\right)  \sum_{j \in SWC_i \cup L_i} a_{ij}^2 \\
&= &&O\left(\frac{1}{\log^3(n)}\right)\sum_{j \in SWC_i \cup L_i} a_{ij}^2.
\end{align*}

Thus, it indeed suffices to show that for any small bucket $B_{i,\ell}$ and any big bucket $B_{i,\ell'}$, it holds that
\[\sum_{t=1}^T   \sum_{j \in SWC_{i,t}  \cap B_{i,\ell}}\sum_{j' \in L_{i,t} \cap B_{i,\ell'}} |a_{ij}a_{ij'}| \leq O\left(\frac{1}{\log^6(n)\log(\Delta)}\right) \sum_{j \in \left(SWC_i \cap B_{i,\ell} \right) \cup B_{i,\ell'}} a_{ij}^2.\]

We use a case distinction. First, consider the case that $\frac{\max_{j' \in B_{i,\ell'}} |a_{ij'}|}{\max_{j\in B_{i,\ell}} |a_{ij}|} \leq \frac{1}{\log^6(n)\log(\Delta)}\frac{1}{\max_{t \in [T]} |Q_t|}$.
Then,
\begin{align*}
\sum_{t=1}^T   \sum_{j \in SWC_{i,t}  \cap B_{i,\ell}}\sum_{j' \in L_{i,t} \cap B_{i,\ell'}} |a_{ij}a_{ij'}|
&\leq \sum_{t=1}^T |SWC_{i,t} \cap B_{i,\ell}||L_{i,t} \cap B_{i,\ell'}| (\max_{j \in B_{i,\ell}} |a_{ij}|)(\max_{j' \in B_{i,\ell'}} |a_{ij'}|) \\
 &\leq \sum_{t=1}^T |SWC_{i,t} \cap B_{i,\ell}||Q_t| (\max_{j \in B_{i,\ell}} |a_{ij}|)(\max_{j' \in B_{i,\ell'}} |a_{ij'}|)  \\
&\leq \frac{1}{\log^6(n)\log(\Delta)}\sum_{t=1}^T |SWC_{i,t} \cap B_{i,\ell}| (\max_{j \in B_{i,\ell}} |a_{ij}|)^2 \\
&= O\left(\frac{1}{\log^6(n)\log(\Delta)}\right)\sum_{t=1}^T \sum_{j \in SWC_{i,t} \cap B_{i,\ell}} a_{ij}^2 \\
&= O\left(\frac{1}{\log^6(n)\log(\Delta)}\right) \sum_{j \in \left(SWC_i \cap B_{i,\ell} \right) \cap B_{i,\ell'}} a_{ij}^2.
\end{align*}

Thus, it remains to consider the case that $\frac{\max_{j' \in B_{i,\ell'}} |a_{ij'}|}{\max_{j\in B_{i,\ell}} |a_{ij}|} \geq \frac{1}{\log^6(n)\log(\Delta)}\frac{1}{\max_{t \in [T]} |Q_t|}$. In particular,
\begin{align*}
 \max_{j \in B_{i,\ell}} |a_{ij}| \leq \left(\max_{j' \in B_{i,\ell'}} |a_{ij'}| \right) \cdot \log^6(n)\log(\Delta)\max_{t \in [T]} |Q_t| \leq \left(\max_{j' \in B_{i,\ell'}} |a_{ij'}| \right) \frac{1}{\log^6(n)\log(\Delta)}\frac{|B_{i,\ell'}|}{|B_{i,\ell}| \max_{t \in [T]} |Q_t|}
\end{align*}
where the last inequality uses $|Q_t| = O(\log n)$ for every $t \in [T]$, $|B_{i,\ell'}| \geq \log(nm\Delta)^{90}$ and $|B_{i,\ell}| \leq \log(nm\Delta)^{50}$.
Therefore,
\begin{align*}
\sum_{t=1}^T   \sum_{j \in SWC_{i,t}  \cap B_{i,\ell}}\sum_{j' \in L_{i,t} \cap B_{i,\ell'}} |a_{ij}a_{ij'}|
&\leq \sum_{t=1}^T |SWC_{i,t} \cap B_{i,\ell}||L_{i,t} \cap B_{i,\ell'}| (\max_{j \in B_{i,\ell}} |a_{ij}|)(\max_{j' \in B_{i,\ell'}} |a_{ij'}|) \\
 &\leq \sum_{t=1}^T |SWC_{i,t} \cap B_{i,\ell}||Q_t| (\max_{j \in B_{i,\ell}} |a_{ij}|)(\max_{j' \in B_{i,\ell'}} |a_{ij'}|)  \\
&\leq   |B_{i,\ell}|\left(\max_{t \in [T]} |Q_t| \right) (\max_{j \in B_{i,\ell}} |a_{ij}|)(\max_{j' \in B_{i,\ell'}} |a_{ij'}|) \\
&\leq \frac{1}{\log^6(n)\log(\Delta)}|B_{i,\ell'}| (\max_{j' \in B_{i,\ell'}} |a_{ij'}|) ^2 \\
&= O\left(\frac{1}{\log^6(n)\log(\Delta)}\right) \sum_{j \in \left(SWC_i \cap B_{i,\ell} \right) \cup B_{i,\ell'}} a_{ij}^2.
\end{align*}

\medskip
\paragraph{Component II of the proof--proving \Cref{clm:in-weighted-recursion-creator}}
We next prove \Cref{clm:in-weighted-recursion-creator}, i.e., we show how to compute a vector $\chi \in \{-1,1\}^n$ and a partition $Q_1 \sqcup Q_2 \sqcup \ldots \sqcup Q_T = [n]$ with $T \leq n/2$ satisfying the 
following guarantees

\begin{enumerate}
		\item $\sum_{i=1}^m imp_i (1+\delta)^{\sum_{t =1}^T \max(0,|Q_t \cap(S_i \cup M_i)| - 1)} \leq \left(1 + \frac{1}{10 \log^2(n)} \right)\sum_{i=1}^m imp_i $
	\item $\sum_{t=1}^T \left(\sum_{j \in Q_t \cap(M_i \cup L_i)} a_{ij} \chi_j \right)^2 \leq \left(1 + \frac{100}{\log^3 n} \right) \sum_{j \in M_i \cup L_i} a^2_{ij}$
	\item $|Q_t| = O(\log n)$ for every $t \in [T]$
\end{enumerate}

We first use the partitioning algorithm of \cref{thm:partition_with_isolation} to compute  with work $(n + m + nnz(A))\cdot\poly(\log(\Delta n m))$ and depth $\poly(\log(\Delta n m))$ two partitions $\mathcal{P} = P_1 \sqcup P_2 \sqcup \ldots \sqcup P_{T_P} = [n]$ and $\mathcal{Q} = Q_1 \sqcup Q_2 \sqcup \ldots \sqcup Q_{T} = [n]$, the latter being the one we will output, that satisfy

\begin{enumerate}
    \item $T_P \in [\log^{20}(nm \Delta), 10\log^{20}(nm \Delta)]$, $T \leq n/2$ 
    \item $\mathcal{Q}$ is a refinement of $\mathcal{P}$
    \item $|B_{i,\ell}\cap P_t| \leq \left(1 + \frac{1}{\log^3 (n)} \right)\frac{|B_{i,\ell}|}{T_P}$ for every $i \in [m]$, $\ell \in [\ell_{max}]$ such that $B_{i,\ell}$ is medium or large and $t \in [T_P]$
    \item $|Q_t| = O(\log n)$ for every $t \in [T]$
    \item $\sum_{i=1}^{m} imp_i (1+\delta)^{\sum_{t = 1}^{T} \max(|\left(S_i \cup M_i \right) \cap Q_t| - 1, 0)} \leq \left(1 +  \frac{1}{\log^3 (n)} \right) \sum_{i=1}^m imp_i$
\end{enumerate}
Concretely, we get the guarantees by setting $M = n\Delta$, using $\{S_i \cup M_i \colon i \in [m]\}$ as the family of "small" subsets and $\{B_{i,\ell} \colon i \in [m],\ell \in [\ell_{max}] \text{ and $B_{i,\ell}$ is a medium or big bucket}\}$ as the family of "big" subsets. In order to satisfy the preconditions of \cref{thm:partition_with_isolation}, one has to verify that

\begin{enumerate}
    \item $|S_i \cup M_i| \leq \log^{100}(nM) = \log^{100}(nm\Delta)$ for every $i \in [m]$
    \item $|B_{i,\ell}| \geq \log^{30}(nm\Delta)$ for every medium or large bucket $B_{i,\ell}$
    \item $n \geq (\log M)^{1/\delta_{T\ref{thm:partition_with_isolation}}}$
    \item $m \cdot \ell_{max} \leq M$
\end{enumerate}
The first property follows from $|S_i \cup M_i| \leq \ell_{max}\log^{90}(nm \Delta) = O(\log^3 (n) \log(\Delta))\log^{90}(nm \Delta) \leq \log^{100}(nm\Delta)$.
The second property directly follows from how medium and large buckets are defined. The third property holds from the input guarantee that $n \geq (\log(3m\Delta))^{1/\delta}$ as long as $\delta$ is sufficiently small. The fourth property follows from $O(\log^3(n)\log(\Delta)) \leq \Delta$ together with the input guarantee that $\Delta \geq n$.

By properties 4 and 5 above, the partition $Q_1 \sqcup  Q_2 \ldots \sqcup Q_T$ already satisfies the first and third guarantees of \Cref{clm:in-weighted-recursion-creator}. It thus remains to show how to compute $\chi \in \{-1,1\}^n$ that satisfies the second guarantee of \Cref{clm:in-weighted-recursion-creator}, namely that 

\[\sum_{t=1}^T \left(\sum_{j \in Q_t \cap(M_i \cup L_i)} a_{ij} \chi_j \right)^2 \leq \left(1 + \frac{100}{\log^3 n} \right) \sum_{j \in M_i \cup L_i} a^2_{ij}.\]

For every $i \in [m]$, $t \in [T_P]$ and $\ell \in [\ell_{max}]$ such that $B_{i,\ell}$ is a medimum or large bucket, we have
\begin{align*}
\sum_{j \in P_t \cap B_{i,\ell}} a_{ij}^2  
&\leq  |B_{i,\ell} \cap P_t| \cdot \left(\max_{j \in B_{i,\ell}} |a_{ij}| \right)^2 \\
&\leq \left(1 + \frac{1}{\log^3(n)} \right) \frac{|B_{i,\ell}|}{T_P} \cdot \left(1 + \frac{1}{\log^3(n)} \right)^2 \left( \min_{j \in B_{i,\ell}} |a_{ij}| \right)^2 \\
&\leq \left(1 + \frac{1}{\log^3(n)} \right)^3 \frac{1}{T_P} \sum_{j \in B_{i,\ell}} a_{ij}^2 \\
&\leq  \left(1 + \frac{4}{\log^3(n)} \right) \frac{1}{T_P} \sum_{j \in B_{i,\ell}} a_{ij}^2.
\end{align*}
Hence, summing over all medium and large buckets, we get that for every $t \in [T_P]$, we have

\[\sum_{j \in P_t \cap (M_i \cup L_i)} a_{ij}^2 \leq \left(1 + \frac{4}{\log^3(nm)} \right) \frac{1}{T_P}\sum_{j \in M_i \cup L_i} a_{ij}^2.\]

We now process the parts in  $T_P = \Omega(\log^{10}(mn))$ rounds, in each round $t$ determining the part of the vector $\chi$ for variables in $[P_t]$. We will use MWU throughout these rounds to average out the losses in each constraint $i\in[m]$, in a manner similar to what we did in the proof of \Cref{thm:rootdepth-optimal} and \Cref{lem:RecursionCreator}. At the beginning, we set $imp(i)=1$ for all $i\in [m]$. Let us consider round $t\in [T_P]$, and let $imp(i)$ be the importance of constraint $i$ in the MWU framework in this round, as stated in \Cref{lem:MWU}. In round $t$, we invoke the sequential derandomization of \Cref{thm:sequential_derandomization} in each of the parts $Q_r$ with $Q_r \subseteq P_t$, $r \in [T]$, all in parallel, but with a scaling in the importance of \Cref{thm:sequential_derandomization}: for each $i$, set the importance of the $disc^2(i)$ to be $imp(i) \cdot T_P/(\sum_{j \in M_i \cup L_i} a_{ij}^2)$. Since $|Q_r|= O(\log n)$, the sequential derandomization takes $\poly(\log(nm))$ depth. From \Cref{thm:sequential_derandomization}, for each part $Q_r$, we get a vector $\chi\in \{-1, 1\}^{[Q_r]}$ such that 

\begin{align*}
\sum_{i=1}^{m} \frac{imp(i) \cdot T_P}{(\sum_{j \in M_i \cup L_i} a_{ij}^2)} \cdot \left(\sum_{j \in Q_r \cap(M_i \cup L_i)} a_{ij} \chi_j \right)^2 \leq \left(1 + \frac{1}{\log^3 n} \right) \sum_{i=1}^{m} \frac{imp(i)\cdot T_P}{(\sum_{j \in M_i \cup L_i} a_{ij}^2)} \cdot \left(\sum_{j \in Q_r \cap(M_i \cup L_i)} a^2_{ij} \right),    
\end{align*}
and 
\begin{align*}
\left(\sum_{j \in Q_r \cap(M_i \cup L_i)} a_{ij} \chi_j \right)^2 \leq C \log (mn) \cdot \left(\sum_{j \in Q_r \cap(M_i \cup L_i)} a^2_{ij} \right).    
\end{align*}
We next argue that these two properties are exactly sufficient for us to run the MWU framework of \Cref{lem:MWU}, and as a result, conclude the second property of \Cref{clm:in-weighted-recursion-creator}. For the gap parameters in the MWU framework, let us define 
$$
gap^{t}(i) =\frac{T_P}{\left(1+\frac{7}{\log^3 n}\right) \cdot (\sum_{j \in M_i \cup L_i} a_{ij}^2)} \sum_{r \in [T] \colon Q_r \subseteq P_t}\left(\sum_{j \in Q_r \cap(M_i \cup L_i)} a_{ij} \chi_j \right)^2.$$
We argue that we have $\sum_{i=1}^{m} imp(i) \cdot gap^{t}(i) \leq \sum_{i=1}^{m} imp(i)$, thus satisfying a key condition of the MWU framework. In particular, by summing over different $r\in [T]$ with $Q_r \subseteq P_t$, we have

\begin{align*}
& &&\left(1+\frac{7}{\log^3 n}\right) \sum_{i=1}^{m} imp(i) \cdot gap^{t}(i) \\
&= &&\sum_{r \in [T] \colon Q_r \subseteq P_t}\sum_{i=1}^{m} \frac{imp(i)\cdot T_P}{(\sum_{j \in M_i \cup L_i} a_{ij}^2)} \cdot \left(\sum_{j \in Q_r \cap(M_i \cup L_i)} a_{ij} \chi_j \right)^2 \\
& \leq  &&\left(1 + \frac{1}{\log^3 n} \right) \sum_{r \in [T] \colon Q_r \subseteq P_t}\sum_{i=1}^{m} \frac{imp(i)\cdot T_P}{(\sum_{j \in M_i \cup L_i} a_{ij}^2)} \cdot \left(\sum_{j \in Q_r \cap(M_i \cup L_i)} a^2_{ij} \right)  \\
& \leq  &&\left(1 + \frac{1}{\log^3 n} \right) \sum_{i=1}^{m} \frac{imp(i)\cdot T_P}{(\sum_{j \in M_i \cup L_i} a_{ij}^2)} \cdot \sum_{r \in [T] \colon Q_r \subseteq P_t} \left(\sum_{j \in Q_r \cap(M_i \cup L_i)} a^2_{ij} \right) \\
& =  &&\left(1 + \frac{1}{\log^3 n} \right) \sum_{i=1}^{m} \frac{imp(i)\cdot T_P}{(\sum_{j \in M_i \cup L_i} a_{ij}^2)} \cdot \sum_{j \in P_t \cap (M_i \cup L_i)} a_{ij}^2 \\
& \leq && \left(1 + \frac{1}{\log^3 n} \right) \sum_{i=1}^{m} \frac{imp(i)\cdot T_P}{(\sum_{j \in M_i \cup L_i} a_{ij}^2)} \cdot \left(1 + \frac{4}{\log^3(nm)} \right) \frac{1}{T_P}\sum_{j \in M_i \cup L_i} a_{ij}^2 \\
& \leq && \left(1+\frac{7}{\log^3 n}\right) \sum_{i=1}^{m} imp(i).
\end{align*}
Hence, indeed $\sum_{i=1}^{m} imp(i) \cdot gap^{t}(i) \leq \sum_{i=1}^{m} imp(i)$. We next argue that we also have $gap^{t}(i)\in[0, W]$ for $W=O(\log (mn))$. Notice that

\begin{align*}
gap^{t}(i) &\leq &&\sum_{r \in [T] \colon Q_r \subseteq P_t} \frac{T_P}{(\sum_{j \in M_i \cup L_i} a_{ij}^2)} \cdot \left(\sum_{j \in Q_r \cap(M_i \cup L_i)} a_{ij} \chi_j \right)^2 \\
& \leq  && \sum_{r \in [T] \colon Q_r \subseteq P_t}\frac{C\log (mn) \cdot T_P}{(\sum_{j \in M_i \cup L_i} a_{ij}^2)} \cdot \left(\sum_{j \in Q_r \cap(M_i \cup L_i)} a^2_{ij} \right)  \\
& \leq  && \frac{C\log (mn)\cdot T_P}{(\sum_{j \in M_i \cup L_i} a_{ij}^2)} \cdot \sum_{r \in [T] \colon Q_r \subseteq P_t} \left(\sum_{j \in Q_r \cap(M_i \cup L_i)} a^2_{ij} \right) \\
& =  && \frac{C\log (mn) \cdot T_P}{(\sum_{j \in M_i \cup L_i} a_{ij}^2)} \cdot \sum_{j \in P_t \cap (M_i \cup L_i)} a_{ij}^2 \\
& \leq && \frac{C\log (mn) \cdot T_P}{(\sum_{j \in M_i \cup L_i} a_{ij}^2)} \cdot \left(1 + \frac{4}{\log^3(nm)} \right) \frac{1}{T_P}\sum_{j \in M_i \cup L_i} a_{ij}^2 \\
& \leq &&  2C\log (mn).
\end{align*}

Hence, $imp(i)$ and $gap^{t}(i)$ satisfy both properties of the MWU framework phrased in \Cref{lem:MWU}. Hence, by MWU, after processing all rounds $t\in[T_P]$, we get that for each $i\in [m]$, we have $\sum_{t=1}^{T_P} gap^{t}(i)/T_P \leq 1+1/\log^3 n.$ This implies that for the overall vector $\chi \in \{-1,1\}^n$, we have

\begin{align*}
    \sum_{t=1}^{T_P} \frac{1}{T_P} \cdot gap^{t}(i) = \sum_{t=1}^{T_P}\frac{1}{\left(1+\frac{7}{\log^3 n}\right) \cdot (\sum_{j \in M_i \cup L_i} a_{ij}^2)} \sum_{r \in [T] \colon Q_r \subseteq P_t}\left(\sum_{j \in Q_r \cap(M_i \cup L_i)} a_{ij} \chi_j \right)^2 \leq 1+\frac{1}{\log^3 n}.
\end{align*}
Hence, we have
\begin{align*}
& &&\sum_{t=1}^T \left(\sum_{j \in Q_t \cap(M_i \cup L_i)} a_{ij} \chi_j \right)^2 \\
&=  &&\sum_{t=1}^{T_P} \sum_{r \in [T] \colon Q_r \subseteq P_t}\left(\sum_{j \in Q_r \cap(M_i \cup L_i)}  a_{ij} \chi_j \right)^2 \\
&\leq &&\left(1+\frac{7}{\log^3 n}\right)\left(1+\frac{1}{\log^3 n}\right) \cdot (\sum_{j \in M_i \cup L_i} a_{ij}^2) \\
&\leq &&\left(1+\frac{100}{\log^3 n}\right) \cdot (\sum_{j \in M_i \cup L_i} a_{ij}^2),
\end{align*}
which provides the second property of \Cref{clm:in-weighted-recursion-creator} and therefore finishes the proof.
\end{proof}


\subsection{The recursive algorithm for weighted set balancing}
\label{subsec:weightedRecursion}

We now present our recursive algorithm for the weighted set balancing problem, as \Cref{lem:recursion_weighted}. When we are not in the base case, this algorithm will itself invoke \Cref{lem:mwu_weighted}, which prepares an instance to be solved recursively. Once we are done with this lemma, we restate and prove\Cref{thm:main-weighted}.

\begin{lemma}
    \label{lem:recursion_weighted}
    There exists an absolute constant $\delta \in (0,0.5]$ such that the following holds:
    Let $n,m,\Delta \in \mathbb{N}$ with $m \geq 2$, $\Delta \geq n$ and $n \geq 10$. Let $A \in \mathbb{Z}^{m \times n}$ with $\sum_{j=1}^n |a_{ij}| \leq \Delta$ for every $i \in [m]$ and $budget \in (\mathbb{N}_0)^m$ satisfying $\sum_{i=1}^m (1-\delta)^{budget_i} \leq 0.5\left(1 + \frac{1}{\log(n)}\right)$. There exists a deterministic parallel algorithm with work $(n + m + nnz(A)) \cdot \poly(\log(nm\Delta))$ and depth $\poly(\log(nm\Delta))$ that can compute a vector $\chi \in \{-1,1\}^n$ and for every $i \in [m]$ a subset $Z_i \subseteq [n]$ along with a partition $Z_{i,1} \sqcup Z_{i,2} \sqcup \ldots \sqcup Z_{i,budget_i}$ satisfying 
    
    \begin{itemize}
        \item $\sum_{b = 1}^{budget_i} \left(\sum_{j \in Z_{i,b}} a_{ij}\chi_j \right)^2 \leq \left(2 - \frac{1}{\log(n)}\right) \sum_{j=1}^n a^2_{ij}$  (Assuming $\delta \leq \delta_{L\ref{lem:mwu_weighted}}/2$)\\
        \item $(\sum_{j \in [n] \setminus Z_i} a_{ij} \chi_j)^2 \leq \left(1 - \frac{1}{\log(n)}\right) O(\log m \sum_{j=1}^n a^2_{ij})$.
    \end{itemize}
\end{lemma}
\begin{proof}
We give a recursive algorithm.
First, consider the base case $n \in [10,(\log 3m \Delta)^{1/\delta_{\ref{lem:mwu_weighted}}}]$. Then, we can define $Z_i = \emptyset$ for every $i\in [m]$ and use the sequential derandomization algorithm of \Cref{thm:sequential_derandomization} to compute with work $(n + m + nnz(A)) \cdot \poly(\log(nm\Delta))$ and depth $n\poly(\log(nm\Delta)) = \poly(\log(nm\Delta))$ a vector $\chi \in \{-1,1\}^n$ satisfying 

\[(\sum_{j \in [n] \setminus Z_i} a_{ij} \chi_j)^2 = (\sum_{j \in [n]} a_{ij} \chi_j)^2 = 0.5O(\log m \sum_{j=1}^n a^2_{ij}) \leq \left(1 - \frac{1}{\log(n)} \right)O(\log m \sum_{j=1}^n a^2_{ij}).\]

Now, consider the case that $n \geq (\log 3m \Delta)^{1/\delta_{L\ref{lem:mwu_weighted}}}$.

We use the algorithm of \cref{lem:mwu_weighted} to compute with work $(n + m + nnz(A)) \cdot \poly(\log(nm\Delta))$ and depth $\poly(\log(nm\Delta))$ a vector $\bar{\chi} \in \{-1,1\}^n$, a partition $Q_1 \sqcup Q_2 \sqcup \ldots \sqcup Q_T = [n]$ with $T \leq n/2$ (In fact, without loss of generality we can assume $T = \lfloor n/2\rfloor \geq 10$) and for every $i \in [m]$ a subset $COL_i \subseteq [n]$ such that for $imp_i := (1-\delta)^{budget_i}$, it holds that:

\begin{itemize}
    \item $\sum_{i=1}^m imp_i(1+2\delta)^{|COL_i|} \leq \left( 1 +\frac{1}{10\log^2(n)}\right) \sum_{i=1}^m imp_i $
    \item $\sum_{t=1}^T \left(\sum_{j \in Q_t \setminus COL_i} a_{ij} \bar{\chi}_j \right)^2 \leq \left(1 + \frac{1}{10\log^2 (n)}\right) \sum_{j \in [n] \setminus COL_i} a^2_{ij}$
\end{itemize}

Now, let $A^{rec} \in \mathbb{Z}^{m \times T}$ with $a^{rec}_{it} = \sum_{j \in Q_t \setminus COL_i} a_{ij} \bar{\chi}_j$ for every $i \in [m]$ and $t \in [T]$. Note that
$\sum_{t=1}^T|a^{rec}_{it}| \leq \sum_{j=1}^n |a_{ij}| \leq \Delta$. Also, let $budget^{rec} \in (\mathbb{N}_0)^m$ with $budget^{rec}_i = budget_i - |COL_i|$ for every $i \in [m]$.

Note that for $\delta \leq 1/3$ being sufficiently small, it holds that

\begin{align*}
    \sum_{i=1}^m (1-\delta)^{budget^{rec}_i} &= \sum_{i=1}^m (1-\delta)^{budget_i} \cdot (1- \delta)^{-|COL_i|} \\
    &= \sum_{i=1}^m imp_i \left(\frac{1}{1-\delta}\right)^{|COL_i|} \\
    &\leq \sum_{i=1}^m imp_i (1+2\delta)^{|COL_i|} \\
    &\leq \left( 1 +\frac{1}{10\log^2(n)}\right) \sum_{i=1}^m imp_i \\
    &= \left( 1 +\frac{1}{10\log^2(n)}\right) \sum_{i=1}^m (1-\delta)^{budget_i} \\
    &\leq \left( 1 +\frac{1}{10\log^2(n)}\right) 0.5\left(1 + \frac{1}{\log(n)}\right)\\
    &\leq 0.5\left(1 + \frac{1}{\log(n/2)}\right)  \\
    &\leq 0.5\left(1 + \frac{1}{\log(T)}\right).
\end{align*}

In particular, $(1-\delta)^{budget^{rec}_i} \leq 1$ and therefore $budget^{rec}_i \geq 0$.

By invoking induction, we can conclude that we can efficiently compute in parallel a vector $\chi^{rec} \in \{-1,1\}^T$ and for every $i \in [m]$ a subset $Z^{rec}_i \subseteq [T]$ along with a partition $Z^{rec}_{i,1} \sqcup Z^{rec}_{i,2} \sqcup \ldots \sqcup Z^{rec}_{i,budget^{rec}_i} = Z^{rec}_i$ satisfying

\begin{itemize}
    \item $\sum_{b = 1}^{budget^{rec}_i} \left(\sum_{t \in Z^{rec}_{i,b}} a^{rec}_{it}\chi^{rec}_t \right)^2 \leq \left(2 - \frac{1}{ \log(T)}\right) \sum_{t=1}^T (a^{rec}_{it})^2$ \\
    \item $(\sum_{t \in [T] \setminus Z^{rec}_i} a^{rec}_{it} \chi^{rec}_t)^2 \leq \left(1 - \frac{1}{\log(T)}\right) O(\log m \sum_{t=1}^T (a^{rec}_{it})^2)$.
\end{itemize}

We now define $\chi_j = \bar{\chi}_j \cdot \chi^{rec}_t$ for every $t \in [T]$ and $j \in Q_t$. We also define $Z_{i,b} = \left(\bigcup_{t \in Z^{rec}_{i,b}} Q_t \right) \setminus COL_i$ for $b \in [budget^{rec}_i]$ and $Z_{i,(budget^{rec}_i + b)} = \{j_{i,b}\}$ for every $b \in [|COL_i|]$ where $COL_i = \{j_{i,1},j_{i,2},\ldots,j_{i,|COL_i|}\}$ and then define $Z_i = \bigcup_{b \in [budget_i]} Z_{i,b}$.

On one hand, for variables in $Z_i$, we have
\begin{align*}
    \sum_{b = 1}^{budget_i} \left(\sum_{j \in Z_{i,b}} a_{ij}\chi_j \right)^2 &= \sum_{b=1}^{budget^{rec}_i} \left(\sum_{j \in Z_{i,b}} a_{ij}\chi_j \right)^2 + \sum_{j \in COL_i} (a_{ij} \chi_j)^2 \\
    &= \sum_{b=1}^{budget^{rec}_i} \left( \sum_{t \in Z^{rec}_{i,b}} \chi^{rec}_t\sum_{j \in Q_t \setminus COL_i} a_{ij}\bar{\chi}_j\right)^2 + \sum_{j \in COL_i} a_{ij}^2 \\
     &= \sum_{b=1}^{budget^{rec}_i} \left( \sum_{t \in Z^{rec}_{i,b}} \chi^{rec}_t a^{rec}_{it} \right)^2 + \sum_{j \in COL_i} a_{ij}^2 \\
     &\leq \left(2 - \frac{1}{ \log(T)}\right) \sum_{t=1}^T (a^{rec}_{it})^2 + \sum_{j \in COL_i} a_{ij}^2 \\
     &\leq \left(2 - \frac{1}{ \log(T)}\right) \left(1 + \frac{1}{10\log^2 (n)} \right) \sum_{j \in [n] \setminus COL_i} a^2_{ij}  + \sum_{j \in COL_i} a_{ij}^2 \\
     &\leq \left(2 - \frac{1}{ \log(n)} \right) \sum_{j=1}^n a^2_{ij}.
\end{align*}
This satisfies the first property of the output in the lemma. On the other hand, for variables in $ [n] \setminus Z_i$, we have
\begin{align*}
 \left(\sum_{j \in [n] \setminus Z_i} a_{ij} \chi_j\right)^2 &= \left(\sum_{t \in [T]} \chi^{rec}_t\sum_{j \in Q_T \setminus Z_i} a_{ij} \bar{\chi}_j \right)^2 \\
 &= \left(\sum_{t \in [T] \setminus Z^{rec}_i} \chi^{rec}_t\sum_{j \in Q_T \setminus COL_i} a_{ij} \bar{\chi}_j \right)^2 \\
  &= \left(\sum_{t \in [T] \setminus Z^{rec}_i} \chi^{rec}_t a^{rec}_{it} \right)^2 \\
  &\leq \left(1 - \frac{1}{ \log(T)}\right) O(\log m \sum_{t=1}^T (a^{rec}_{it})^2) \\
  &\leq \left(1 - \frac{1}{ \log(T)}\right) \left(1 + \frac{1}{10\log^2 (n)}\right)O\left((\log m  \sum_{j \in [n] \setminus COL_i} a^2_{ij}\right) \\
  &\leq \left(1 - \frac{1}{ \log(n)}\right) O\left(\log m \sum_{j \in [n] \setminus COL_i} a^2_{ij} \right).
\end{align*}
This satisfies the second property of the output in the lemma. 

Finally, note that the algorithm has a recursion depth of $O(\log n)$. The base case runs in work $(n + m + nnz(A)) \cdot \poly(\log(nm\Delta))$ and depth $\poly(\log(nm\Delta))$ and the same bounds hold for preparing the recursion. Therefore, the overall work of the algorithm is $(n + m + nnz(A)) \cdot \poly(\log(nm\Delta))$ and the overall depth is $(n + m + nnz(A)) \cdot \poly(\log(nm\Delta))$.
\end{proof}

\bigskip
\paragraph{Wrapping up the proof of \Cref{thm:main-weighted}}. Finally, we now restate \Cref{thm:main-weighted} and present its proof.

\thmMainWeighted*
\begin{proof}[Proof of \cref{thm:main-weighted}]
First, note that we can assume without loss of generality that the following is satisfied:
\begin{enumerate}
    \item $n \geq 10$
    \item $\max_{j\in [n]} |a_{ij}| > 0$ for every $i \in [m]$
    \item $\max_{j \in [n]} |a_{ij}| = n^{10}$ for every $i \in [m]$
    \item $a_{ij} \in \mathbb{Z}$
\end{enumerate}
If $n < 10$, we can simply add $10$ new dummy variables. The second property holds because we can simply drop rows in $A$ that consist of just zeros. Next, given that the second property holds, we can multiply each $a_{ij}$ by a factor of $\frac{n^{10}}{\max_{j \in [n]} |a_{ij}|}$ for every $i \in [m]$, which thus satisfies the third property. Now, given that the third property holds, we can round each $a_{ij}$ to the closest integer. This introduces an additive error of at most $\frac{\max_{j \in [n]} |a_{ij}|}{n^9} = O(\sum_{j=1}^n a^2_{ij})$.
Now, given that all properties above are satisfied, we can use the algorithm of \cref{lem:recursion_weighted} (with $\Delta = n^{11}$ and $budget_i = c\log(m)$ for a suitably large constant $c$) to compute a vector $\chi \in \{-1,1\}^n$ and for every $i \in [m]$ a subset $Z_i \subseteq [n]$ along with a partition $Z_{i,1} \sqcup Z_{i,2} \sqcup \ldots \sqcup Z_{i,budget_i}$ satisfying

\begin{enumerate}
    \item $\sum_{b=1}^{budget_i}\left(\sum_{j \in Z_{i,b}}a_{ij} \chi_j \right)^2 \leq 2\sum_{j=1}^n a^2_{ij}$
    \item $\left(\sum_{j \in [n] \setminus Z_i} a_{ij} \chi_j \right)^2 \leq O(\log (m) \sum_{j=1}^n a^2_{ij})$
\end{enumerate}

For every $b \in [budget_{i}]$, let $u_b := 1$ and $v_b := \sum_{j \in Z_{i,b}}a_{ij} \chi_j$. By Cauchy-Schwarz, it holds that

\begin{align*}
|\sum_{j \in Z_i}^n a_{ij} \chi_j| &= |\sum_{b=1}^{budget_i}u_bv_b| \\
&\leq \sqrt{\sum_{b=1}^{budget_i} u_b^2}\sqrt{\sum_{b=1}^{budget_i} v_b^2} \\
&= O(\sqrt{\log m}) \sqrt{\sum_{b=1}^{budget_i} \left( \sum_{j \in Z_{i,b}}a_{ij}\right)^2} \\
&= O\left(\sqrt{\sum_{j=1}^n a_{ij}^2 \log m} \right).
\end{align*}

Therefore,
\[|\sum_{j=1}^n a_{ij} \chi_j| \leq |\sum_{j \in Z_i} a_{ij} \chi_j| + |\sum_{j\in [n] \setminus Z_i} a_{ij} \chi_j| = O\left(\sqrt{\sum_{j=1}^n a_{ij}^2 \log m} \right).\]
\end{proof}

\newpage

\bibliographystyle{alpha}
\bibliography{ref}

\appendix

\section{Generalizations of GGR's crude partitioning}
\label{app:prelimAppendix-partition}
Here, we restate and prove the two simple generalizations of the result of \cite{GGR2023Chernoff}, for multi-way partitioning, which we derive in this paper and use in our algorithms. The proofs use the result of \cite{GGR2023Chernoff}, as stated in \Cref{thm:FOCS23}, in a black-box manner.

\UnweightedPartition*
\begin{proof}[Proof of \Cref{lemma:partition_unweighted}]
Let $c$ be a sufficiently large constant.
We prove the following statement by induction on $L$. The partition theorem then follows as a simple corollary.
Let $n,m,k \in \mathbb{N}$ and let $\{S_1,S_2,\ldots,S_m\}$ be a family of subsets of $[n]$. Let $L \in \mathbb{N}$ be a power of $2$ with $L \leq 2^{k \cdot (\eps/2)}$. We define
\[\eps_L = (\eps/2)(1- \sqrt{1/L}) + \frac{\log(L)}{k}.\]
 We prove by induction on $L$ that there is a deterministic parallel algorithm with depth $(\poly(\log(nm)k\log(2L))$ and work $\tilde{O}((n + m + \sum_{i=1}^m|S_i|)\poly(k\log(2L)))$ that computes a partition $P_1 \sqcup P_2 \sqcup \ldots \sqcup P_L = [n]$ satisfying that $|S_i \cap P_\ell| \leq \max((1+\eps_L)|S_i|/L, c \cdot \log(m)/\eps^2)$.

\medskip
For every $L$, we define $\delta_L = 0.01 \left(\eps \sqrt{1/L} + \frac{1}{k} \right)$. We have the following properties:

\paragraph{Properties about $\eps_L$}

\begin{enumerate}
    \item $\eps_L \in [0,\eps]$ for every $L \in [1,2^{k (\eps/2)}]$
    \item $(1+\eps_L) \geq (1+\delta_L)(1+\eps_{L/2})$ for every $L \in \{2,4,\ldots,2^{k(\eps/2)}\}$.
\end{enumerate}
The first property holds by definition of $\eps_L$, and the second property is satisfied because, for every $L \in \{2,4,\ldots,2^{k(\eps/2)}\}$, we have
\begin{align*}
(1 + \delta_L)(1+\eps_{L/2}) &\leq (1+\eps_{L/2}) + 2\delta_L \\ &\leq \left(1+ (\eps/2)(1- \sqrt{2/L}) + \frac{\log(L/2)}{k}\right) +  0.02 \left(\eps \sqrt{1/L} + \frac{1}{k} \right) \\
&\leq 1 + (\eps/2)(1- \sqrt{1/L}) + \frac{\log(L)}{k} \\
&= (1 + \eps_L).
\end{align*}

In the rest of this proof, we argue the existence of the claimed deterministic parallel algorithm that computes a partition into $L$ parts, with the aforementioned properties, by induction on $L$.

\paragraph{Base Case}
The base case $L = 1$ follows trivially as $\eps_1 \geq 0$.

\paragraph{Inductive Case}
 Next, we assume that $L$ is a power of two with $2 \leq L \leq 2^{k(\eps/2)}$. As a first step, we compute a partition $P_1 \sqcup P_2 = [n]$ that satisfies $\max(|S_i \cap P_1|,|S_i \cap P_2|) = |S_i|/2 + O(\sqrt{|S_i| \log m}) + 0.001\frac{|S_i|}{k}$ for every $i \in [m]$ using the algorithm of \Cref{thm:FOCS23}, which takes $\tilde{O}(n + m + \sum_{i=1}^m |S_i|)\poly(k)$ work and depth $\poly(\log(nm)k)$. Next, we inductively refine both $P_1 = P_{1,1} \sqcup P_{1,2} \sqcup \ldots \sqcup P_{1,L/2}$ and $P_2 = P_{2,1} \sqcup P_{2,2} \sqcup \ldots \sqcup P_{2,L/2}$, each 
 using the algorithm guaranteed by the induction hypothesis. The final output of the algorithm is then $P_{1,1} \sqcup P_{1,2} \sqcup \ldots \sqcup P_{1,L/2} \sqcup  P_{2,1} \sqcup P_{2,2} \sqcup \ldots \sqcup P_{2,L/2} = [n]$.

We next show that $|S_i \cap P_{1,1}| \leq \max((1+\eps_L)|S_i|/L,c \cdot \log(m)/\eps^2)$. A completely analogous reasoning holds for all the other parts of the partition.

The induction hypothesis implies that $|S_i \cap P_{1,1}| \leq \max \left((1+\eps_{L/2})\frac{|S_i \cap P_1|}{L/2}, c \cdot \log(m)/\eps^2 \right)$.
If $(1 + \eps_{L/2})\frac{|S_i \cap P_1|}{L/2} \leq c \cdot \log(m)/\eps^2$, then there is nothing left to show.
Thus, it remains to consider the case that $(1+\eps_{L/2})\frac{|S_i \cap P_1|}{L/2} > c \cdot \log(m)/\eps^2$. In particular, this gives that
\[|S_i| \geq \frac{L \cdot c \cdot \log(m)/\eps^2}{2 (1+\eps_{L/2})} \geq (c/4) \cdot (\log(m)/\eps^2) \cdot L.\]
Here we use that $\eps_{L/2} \leq 1$.
Thus, for $c$ being a large enough constant, we get
\begin{align*}
|S_i \cap P_{1,1}| &\leq (1 + \eps_{L/2})\frac{|S_i \cap P_1|}{L/2}\\
&\leq (1+\eps_{L/2})\frac{|S_i|/2 + O(\sqrt{|S_i|\log m}) + 0.001\frac{|S_i|}{k}}{L/2} \\ 
&\leq (1+\eps_{L/2}) \left( 1 + O\left(\sqrt{\frac{\log(m)}{|S_i|}}\right) + \frac{0.002}{k} \right) \frac{|S_i|}{L} \\
&\leq (1+\eps_{L/2}) \left( 1 + 0.01 \eps\sqrt{\frac{1}{L}} + \frac{0.002}{k} \right) \frac{|S_i|}{L} \\
&\leq  (1+\eps_{L/2})(1+\delta_L)\frac{|S_i|}{L} \\
&\leq  (1+\eps_L) \frac{|S_i|}{L}.
\end{align*}
Finally, the work of the algorithm is 
\[\tilde{O}(n + m + \sum_{i=1}^m |S_i|)\poly(k) + \tilde{O}(n + m + \sum_{i=1}^m|S_i|)\poly(k\log(2(L/2))) \leq \tilde{O}(n + m + \sum_{i=1}^m|S_i|)\poly(k\log(2L))\]
and the depth is  
\[\poly(\log(nm)k) + \poly(\log(nm)k\log(2(L/2))) \leq \poly(\log(nm)k\log(2L)).\]
\end{proof}

\WeightedPartition*
\begin{proof}[Proof of \Cref{lem:weighted_partition}]
First, mote that we can assume without loss of generality that $\eps \geq \frac{1}{n}$, which we will do from now on.
For each $i \in [m]$ and $k \in \mathbb{N}_0$, we define
\[S_{i,k} = \{j \in [n] \colon (1-\eps/3)^{k+1}(a_{max})^2 < a^2_{ij} \leq (1-\eps/3)^k (a_{max})^2\}.\]
Consider the following family of subsets of $[n]$:
\begin{itemize}
    \item $\{1,2,\ldots,n\}$
    \item $\{j \in [n] \colon a_{ij} \neq 0\}$ for every $i \in [m]$
    \item $S_{i,k}$ for every $i \in [m]$ and $k \in [\lceil 10 \log (n+1) /\eps \rceil]$
\end{itemize}
Note that the family contains $\hat{m} = 1 + m + \lceil 10 \log(n+1)/\eps \rceil m \leq (nm)^{100}$ sets. Using \cref{lemma:partition_unweighted}, there exists a parallel algorithm with $\tilde{O}(n + \hat{m} + \sum_{i=1}^m \sum_{k = 1}^{\lceil 10 \log(n)/\eps\rceil} |S_{i,k}|) \cdot \poly(1/\eps) = \tilde{O}(n + m + nnz(A)) \cdot \poly(1/\eps)$ work and $\poly(\log(n \hat m)/\eps) = \poly(\log(nm)/\eps)$ depth that computes a partition $P_1 \sqcup P_2 \sqcup \ldots \sqcup P_L = [n]$ satisfying the following for every $\ell \in [L]$:

\begin{itemize}
    \item $|P_\ell| \leq (1+\eps/3)n/L + O(\log (\hat{m})/\eps^2) = (1+\eps/3)n/L + O(\log (nm)/\eps^2)$
    \item $|P_\ell \cap \{j \in [n] \colon a_{ij} \neq 0\}| \leq \frac{(1+\eps)}{L}|\{j \in [n] \colon a_{ij} \neq 0\}| + O(\log (nm)/\eps^2) $ for every $i \in [m]$.
    \item $|P_\ell \cap S_{i,k}|\leq \frac{1+\eps}{L}|S_{i,k}| + O(\log (nm))/\eps^2) $ for every $k \in [\lceil 10 \log(n+1)/\eps \rceil]$ and $i \in [m]$.
\end{itemize}
Note that the first two guarantees directly imply that the partition satisfies the first two output guarantees of \cref{lem:weighted_partition}. It thus remains to verify that the third output guarantee is satisfied as well. For every $\ell \in [L], k \in [\lceil 10 \log(n+1)/\eps\rceil]$ and $i \in [m]$, we have
\begin{align*}
\sum_{j \in P_\ell \cap S_{i,k}} a^2_{ij} &\leq |P_{\ell} \cap S_{i,k}|(1-\eps/3)^k(a_{max})^2 \\
&\leq \left((1+\eps/3)|S_{i,k}|/L + O(\log (nm))/\eps^2) \right)(1-\eps/3)^k(a_{max})^2 \\
&\leq  \frac{1+\eps/3}{1-\eps/3}(1/L) |S_{i,k}|(1-\eps)^{k+1}(a_{max})^2 + O(\log(nm) / \eps^2) (a_{max})^2 \\
&\leq (1+\eps)(1/L)\sum_{j \in S_{i,k}} a^2_{ij} + O(\log (nm))/\eps^2) (a_{max})^2.
\end{align*}
Thus, for every $\ell \in [L]$ and $i \in [m]$, we get
\begin{align*}
\sum_{j \in P_\ell} a^2_{ij} &\leq \left(\sum_{k = 1}^{\lceil 10 \log(n+1)/\eps\rceil} \sum_{j \in P_\ell \cap S_{i,k}} a^2_{ij} \right) + n \cdot \left(1 - \eps/3 \right)^{\lceil 10 \log(n+1)/\eps\rceil}(a_{max})^2 \\
&\leq \left(\sum_{k = 1}^{\lceil 10 \log(n+1)/\eps\rceil} \left((1+\eps)(1/L)\sum_{j \in S_{i,k}} a^2_{ij} + O(\log (nm))/\eps^2) (a_{max})^2 \right)  \right) + (a_{max})^2 \\
&\leq \frac{1+\eps}{L}\sum_{j=1}^n a^2_{ij}+ \lceil10\log(n+1)/\eps \rceil O(\log (nm))/\eps^2)(a_{max})^2 \\
&\leq \frac{1+\eps}{L}\sum_{j=1}^n a^2_{ij}+ O(\log^2 (nm))/\eps^3)(a_{max})^2.
\end{align*}
\end{proof}

\section{Recapping the Multiplicative Weights Update method}

\label{app:MWU}
We present here a particular instantiation of the well-known multiplicative weights update method (see e.g. \cite{arora2012multiplicative}), which we use throughout our algorithms. We also present a simple and self-contained proof.
\MWU*
\begin{proof}[Proof of \Cref{lem:MWU}] 
Set $\eps'=\eps/3$. Notice that $\eta \cdot gap^{t}(i) \leq \eps'$
Define $\Phi^t=\sum_{1}^{m} imp^t(i)$. We have $\Phi^1=m$ and 
\begin{align*} 
\Phi^{t+1} &= \sum_{1}^{m} imp^t(i) =  \sum_{1}^{m} imp^{t}(i) \cdot (1+\eta \cdot gap^t(i)) = \sum_{1}^{m} imp^{t}(i) + \eta \sum_{1}^{m} imp^{t}(i) \cdot gap^t(i) 
\\ &\leq \sum_{1}^{m} imp^{t}(i) + \eta \sum_{1}^{m} imp^{t}(i) = (1+\eta) \Phi^{t} \leq exp(\eta) \cdot \Phi^{t}.
\end{align*}
Hence, we have $\Phi^{T} \leq exp(\eta T) \cdot m.$ Thus, zooming on each constraint $i\in [m]$, we have 
\begin{align} imp^{T}(i) = \prod_{t=1}^{T} (1+\eta \cdot gap^{t}(i))\leq exp(\eta T+\ln m). \label{ineq:boundOnIndivdualImp}
\end{align} We would like to obtain an upper bound on $(\sum_{t=1}^{T} gap^t(i))/T$ from this. Let $z=\eta \cdot gap^{t}(i)$. We have $z\in [0, \eps'].$ Thus, 
\begin{align*}
exp(z/(1+\eps')) &\leq 1+ z/(1+\eps') + (z/(1+\eps'))^2 \\
&\leq 1+ z/(1+\eps') + (\eps' z/(1+\eps') \\
& = 1+(z/(1+\eps')) (1+\eps') = 1+z.
\end{align*}
Therefore, we have 
\begin{align*} 
exp\bigg(\sum_{t=1}^{T} \eta \cdot gap^{t}(i)/(1+\eps')\bigg) \leq \Bigg(\prod_{t=1}^{T} (1+\eta\cdot gap^{t}(i))\Bigg) 
\end{align*}
Now, using \Cref{ineq:boundOnIndivdualImp}, we can conclude that 
\begin{align*}
    exp\bigg(\sum_{t=1}^{T} \eta \cdot gap^{t}(i)/(1+\eps')\bigg) \leq exp(\eta T+\ln m),
\end{align*}
which implies that
\begin{align*}
\sum_{t=1}^{T} gap^{t}(i) \leq (T + \frac{\ln m}{\eta})(1+\eps')
\end{align*}
In other words, 
\begin{align*}
\frac{\sum_{t=1}^{T} gap^{t}(i)}{T} &\leq (1 + \frac{\ln m}{T\eta})\cdot (1+\eps') \\ 
&\leq (1+\eps')(1+\eps') \leq 1+\eps.
\end{align*}
Here, the penultimate inequality uses that $\eta={\eps}/{(3W)}$, $\eps'=\eps/3$, and $T\geq {9W \ln m}/{\eps^2}.$
\end{proof}


\section{Sequential derandomization, augmented with averaging}
\label{app:seqDerandwithAverage}

We present here the proof of the sequential derandomization method, augmented with importance-weighted averaging, as stated in \Cref{thm:sequential_derandomization}. We first restate the result, and then present the proof. 


\SeqDerandWithAverage*

We first start with some useful definitions.

\begin{definition}[$\Delta_i,\lambda_i,\Phi^{(avg)}_{i,j}, \Phi^{(upper)}_{i,j}, \Phi^{(lower)}_{i,j},Pot_j$]
For every $i \in [m]$, we define $\Delta_i = 1000\sqrt{\sum_{j=1}^n a^2_{ij} \log M}$ and $\lambda_i = \frac{\Delta_i}{100 \sum_{j=1}^n a_{ij}^2}$.
Moreover, for every $j \in \{0,1,\ldots,n\}$, we define 

\begin{itemize}
\item $\Phi^{(avg)}_{i,j}(\chi_1,\ldots,\chi_j) = \left( \sum_{j' = 1}^j  a_{ij'} \chi_{j'} \right)^2 + \sum_{j' = j+1}^n a^2_{ij'}$
\item $\Phi^{(upper)}_{i,j}(\chi_1,\chi_2,\ldots,\chi_j) = \prod_{j'=1}^j \left(1 + \lambda_i a_{ij'} \chi_{j'} + (\lambda_i a_{ij'})^2 \right) \cdot \prod_{j' = j+1}^n (1 + (\lambda_i a_{ij'})^2)$
\item $\Phi^{(lower)}_{i,j}(\chi_1,\chi_2,\ldots,\chi_j) = \prod_{j'=1}^j \left(1 - \lambda_i a_{ij'} \chi_{j'} + (\lambda_i a_{ij'})^2 \right)  \cdot \prod_{j' = j+1}^n (1 + (\lambda_i a_{ij'})^2)$ and
\item $Pot_j(\chi_1,\chi_2,\ldots,\chi_j) = \frac{\sum_{i=1}^m imp(i)\cdot \Phi^{(avg)}_{i,j}(\chi_1,\chi_2,\ldots,\chi_j)}{\sum_{i=1}^m imp(i)\cdot\Phi^{(avg)}_{i,0}} +  \frac{1}{2m \cdot M} \sum_{i \in [m]}  \left( \frac{\Phi^{(upper)}_{i, j}(\chi_1,\chi_2,\ldots,\chi_j)}{\Phi^{(upper)}_{i, 0}} + \frac{\Phi^{(lower)}_{i, j}(\chi_1,\chi_2,\ldots,\chi_j)}{\Phi^{(lower)}_{i, 0}} \right)$.
\end{itemize}
\end{definition}

We make use of the following two claims, which will be proven below.

\begin{claim}
    \label{claim:seq_pot_expectation}
    For every $j \in [n]$, it holds that 

    \[0.5\left(Pot_j(\chi_1,\ldots,\chi_{j-1}, 1) + Pot_j(\chi_1,\ldots,\chi_{j-1},-1)\right) = Pot_{j-1}(\chi_1,\ldots,\chi_{j-1}).\]
\end{claim}

\begin{claim}
\label{claim:seq_pot_conclusion}
    It holds that $Pot_0 = 1 + \frac{1}{M}$. Moreover, if $Pot_n(\chi_1,\chi_2,\ldots,\chi_n) \leq 1 + \frac{1}{M}$ for $\chi \in \{-1,1\}^n$, then it holds that
    \begin{enumerate}
    \item $\sum_{i=1}^m imp(i)\cdot disc^2_i \leq \left(1 + \frac{1}{M}\right) \sum_{i=1}^m imp(i)\cdot  \sum_{j = 1}^n a^2_{ij}$ and
    \item $disc_i^2 \leq \Delta_i^2$ for every $i \in [m]$.
    \end{enumerate}
\end{claim}
Before proving the two claims, we use them to prove \cref{thm:sequential_derandomization}.
\begin{proof}[Proof of \cref{thm:sequential_derandomization}]
Our algorithm will compute $\chi \in \{-1,1\}^n$ in an iterative manner. In particular, assume the algorithm has already computed $\chi_1,\chi_2,\ldots,\chi_{j-1}$ for some $j \in [n]$. Then, if $Pot_j(\chi_1,\ldots,\chi_{j-1},1) \leq Pot_{j-1}(\chi_1,\ldots,\chi_{j-1})$, the algorithm sets $\chi_j = 1$ and otherwise the algorithm sets $\chi_j = -1$. \cref{claim:seq_pot_expectation} implies that for every $j \in [n]$, it holds that $Pot_j(\chi_1,\ldots,\chi_j) \leq Pot_{j-1}(\chi_1,\ldots,\chi_{j-1})$. Using \cref{claim:seq_pot_conclusion}, we can then deduce that $Pot_n(\chi_1,\ldots,\chi_n) \leq 1 + \frac{1}{M}$ and therefore $\chi$ satisfies the output guarantees if $C$ is a sufficiently large constant. It remains to discuss some implementation details and then argue about its work and depth. First, note that the algorithm can compute the following in $\tilde{O}(nnz(A) + n + m)$ work and $\poly(\log(nM))$ depth at the beginning:
\begin{enumerate}
    \item $\Phi^{(avg)}_{i,0}$ for every $i \in [m]$
    \item $\Phi^{(upper)}_{i,0}$ for every $i \in [m]$
    \item $\Phi^{(lower)}_{i,0}$ for every $i \in [m]$
    \item $\sum_{i=1}^m imp(i) \Phi^{(avg)}_{i,0}$
    \item $Pot_0$
    \item $\lambda_i$ for every $i \in [m]$
\end{enumerate}
Assume we are at step $j \in [n]$ and that we are given a set $I_j = \{i\in [m] \colon a_{ij} \neq 0\}$.
We inductively assume that the algorithm has computed the following quantities:

\begin{enumerate}
    \item $\sum_{j'=1}^{j-1} a_{ij'} \chi_{j'}$ for every $i \in [m]$
    \item $\Phi^{(avg)}_{i,j-1}(\chi_1,\ldots,\chi_{j-1})$ for every $i \in [m]$
    \item $\Phi^{(upper)}_{i,j-1}(\chi_1,\ldots,\chi_{j-1})$ for every $i \in [m]$
    \item $\Phi^{(lower)}_{i,j-1}(\chi_1,\ldots,\chi_{j-1})$ for every $i \in [m]$
    \item $Pot_{j-1}(\chi_1,\chi_2,\ldots,\chi_{j-1})$
\end{enumerate}
Now, let $\chi_j \in \{-1,1\}$. Note that given this information, we can compute in $(1 + |I_j|)\poly(\log(nM))$ work and $\poly(\log(nM))$ depth the following:
\begin{enumerate}
    \item $\sum_{j'=1}^{j} a_{ij'} \chi_{j'}$ for every $i \in I_j$
    \item $\Phi^{(avg)}_{i,j}(\chi_1,\ldots,\chi_{j-1},\chi_j)$ for every $i \in I_j$
    \item $\Phi^{(upper)}_{i,j}(\chi_1,\ldots,\chi_{j-1},\chi_j)$ for every $i \in I_j$
    \item $\Phi^{(lower)}_{i,j}(\chi_1,\ldots,\chi_{j-1},\chi_j)$ for every $i \in I_j$
 \end{enumerate}
 The reason is that for every $i \in I_j$, we have

\[\sum_{j'=1}^{j} a_{ij'} \chi_{j'} = \left(\sum_{j'=1}^{j-1} a_{ij'} \chi_{j'} \right) + a_{ij}\chi_j,\]
\[\Phi^{(avg)}_{i,j}(\chi_1,\ldots,\chi_{j-1},\chi_j) = \Phi^{(avg)}_{i,j-1}(\chi_1,\ldots,\chi_{j-1}) + 2\chi_j a_{ij}\left( \sum_{j'=1}^{j-1} a_{ij'} \chi_{j'}\right) ,\]
\[\Phi^{(upper)}_{i,j}(\chi_1,\ldots,\chi_{j-1},\chi_j) = \frac{1 + \lambda_i a_{ij}\chi_j + (\lambda_ia_{ij})^2 }{1 + (\lambda_ia_{ij})^2  }\Phi^{(upper)}_{i,j-1}(\chi_1,\ldots,\chi_{j-1})\]

and

\[\Phi^{(lower)}_{i,j}(\chi_1,\ldots,\chi_{j-1},\chi_j) = \frac{1 - \lambda_i a_{ij}\chi_j + (\lambda_ia_{ij})^2 }{1 + (\lambda_ia_{ij})^2  }\Phi^{(lower)}_{i,j-1}(\chi_1,\ldots,\chi_{j-1}).\]

Since $\Phi^{(upper)}_{i,0}$ and $\Phi^{(lower)}_{i,0}$ have been computed at the beginning for every $i \in I_j$, and we have also computed $\sum_{i=1}^m imp(i) \Phi^{(avg)}_{i,0}$, having this information allows us to compute $Pot_{j}(\chi_1,\chi_2,\ldots,\chi_{j-1},\chi_j)$ in $(1 + |I_j|)\poly(\log(nM))$ work and $\poly(\log(nM))$ depth. Hence, we can indeed implement the algorithm described above in $\tilde{O}(nnz(A) + n + m)$ work and $\poly(\log(nM))$ depth.
\end{proof}

\begin{proof}[Proof of \cref{claim:seq_pot_conclusion}]
The fact that $Pot_0 = 1 + \frac{1}{M}$ follows straightforwardly from the definition. Now, assume that $Pot_n(\chi_1,\chi_2,\ldots,\chi_n) \leq 1 + \frac{1}{M}$.
In particular, this implies that $1 + \frac{1}{M} \geq Pot_n(\chi_1,\chi_2,\ldots,\chi_n) \geq \frac{\sum_{i=1}^m imp(i)\cdot\Phi^{(avg)}_{i,n}(\chi_1,\chi_2,\ldots,\chi_n)}{\sum_{i=1}^m imp(i)\cdot \Phi^{(avg)}_{i,0}}$, where the last inequality follows from the fact that all the defined potentials are non-negative at all times. We therefore get

\begin{align*}
\sum_{i=1}^m imp(i)\cdot disc_i^2 &= \sum_{i=1}^m imp(i)\cdot\left( \sum_{j=1}^n a_{ij} \chi_j \right)^2 = \sum_{i=1}^m imp(i)\cdot \Phi^{(avg)}_{i,n}(\chi_1,\chi_2,\ldots,\chi_n)  \\
&\leq \left(1 + \frac{1}{M}\right) \sum_{i=1}^m imp(i)\cdot \Phi^{(avg)}_{i,0} \\
&= \left(1 + \frac{1}{M}\right) \sum_{i=1}^m imp(i)\cdot \sum_{j = 1}^n a^2_{ij},
\end{align*}
which proves the first property in the statement of \Cref{claim:seq_pot_conclusion}.

Next, we argue that the second property in the statement of \Cref{claim:seq_pot_conclusion} also holds. Fix some $i \in [m]$. We will show that the following three properties hold:

\begin{itemize}
    \item $\Phi^{(upper)}_{i,0} = \Phi^{(lower)}_{i,0} \leq M^{100}$.
    \item If $\sum_{j=1}^n a_{ij} \chi_j \geq \Delta_i$, then $\Phi^{(upper)}_{i,n}(\chi_1,\chi_2,\ldots,\chi_n)  \geq M^{1000}$.
    \item If $\sum_{j=1}^n a_{ij} \chi_j \leq - \Delta_i$, then $\Phi^{(lower)}_{i,n}(\chi_1,\chi_2,\ldots,\chi_n)  \geq M^{1000}$.
\end{itemize}
In particular, if $disc_i^2 > \Delta_i^2$, we would have $\max\left(\frac{\Phi^{(upper)}_{i,n} }{\Phi^{(upper)}_{i,0}},\frac{\Phi^{(lower)}_{i,n} }{\Phi^{(lower)}_{i,0}} \right) \geq M^5$, which, for $M > 1$, would imply that $Pot_n(\chi_1,\chi_2,\ldots,\chi_n)  > 1 + \frac{1}{M}$. Therefore, $Pot_n(\chi_1,\chi_2,\ldots,\chi_n)  \leq 1 + \frac{1}{M}$ indeed implies that $disc_i^2 \leq \Delta_i^2 = O\left(\sum_{j=1}^n a^2_{ij} \log M \right)$, thus proving the second property in the statement of \Cref{claim:seq_pot_conclusion}. It remains to show that the three properties above hold. The first property holds, because

\begin{align*}
    \Phi^{(upper)}_{i,0} =  \Phi^{(lower)}_{i,0} &= \prod_{j=1}^n (1 + (\lambda_i a_{ij})^2) \leq e^{\sum_{j=1}^{n} (\lambda_i a_{ij})^2} \\
    &= e^{\left(\frac{\Delta_i}{100 \sum_{j=1}^n a^2_{ij}} \right)^2\sum_{j=1}^n a^2_{ij}} = e^{\frac{\Delta_i^2}{10000\sum_{j=1}^n a^2_{ij}}} \\
    &=  e^{100\log(M)} = M^{100}.
\end{align*}

Next, we show that the second property holds (the argument for the third property is completely analogous, and thus omitted). Assume that $\sum_{j=1}^n a_{ij} \chi_j \geq \Delta_i$. We have to show that $\Phi^{(upper)}_{i,n}(\chi_1,\chi_2,\ldots,\chi_n)  \geq M^{1000}$.
First, note that
\begin{align*}
\sum_{j \in [n] \colon |\lambda_i a_{ij}| > 0.5} |a_{ij}| \leq 2 \lambda_i \sum_{j \in [n] \colon |\lambda_i a_{ij}| > 0.5} a^2_{ij} =  \frac{2\Delta_i}{100 \sum_{j=1}^n a^2_{ij}}\sum_{j \in [n] \colon |\lambda_i a_{ij}| > 0.5} a^2_{ij} \leq \frac{\Delta_i}{2}.
\end{align*}
Next, we use the fact that $1 + x + x^2 \geq e^x$ for all $x \leq 0.5$ to conclude that
\begin{align*}
\Phi^{(upper)}_{i,n}(\chi_1,\chi_2,\ldots,\chi_n)  &=  \prod_{j=1}^n \left(1 + \lambda_i a_{ij} \chi_{j} + (\lambda_i a_{ij})^2 \right)\\
&\geq \prod_{j \in [n] \colon \lambda_i a_{ij} \chi_j \leq 0.5} e^{\lambda_i a_{ij} \chi_{j}} \\ 
&=  e^{\lambda_i\sum_{j \in [n] \colon \lambda_i a_{ij} \chi_j \leq 0.5} a_{ij} \chi_{j}} \\
&\geq e^{\lambda_i(\sum_{j \in [n]} a_{ij} \chi_{j} - |\sum_{j \in [n] \colon \lambda_i a_{ij} \chi_j > 0.5}a_{ij} \chi_{j}|)} \\
&\geq e^{\lambda_i(\Delta_i - \Delta_i/2)} \\
&\geq e^{\frac{\Delta_i^2}{1000 \sum_{j=1}^n a^2_{ij}}} \geq M^{1000}.
\end{align*}
\end{proof}

\begin{proof}[Proof of \cref{claim:seq_pot_expectation}]
Fix some $j \in [n]$ and $i \in [m]$. First, note that it suffices to show the following:

\begin{itemize}
    \item $0.5\left(\Phi^{(avg)}_{i,j}(\chi_1,\ldots,\chi_{j-1}, 1) + \Phi^{(avg)}_{i,j}(\chi_1,\ldots,\chi_{j-1},-1)\right) = \Phi^{(avg)}_{i,j-1}(\chi_1,\ldots,\chi_{j-1})$
    \item $0.5\left(\Phi^{(upper)}_{i,j}(\chi_1,\ldots,\chi_{j-1}, 1) + \Phi^{(upper)}_{i,j}(\chi_1,\ldots,\chi_{j-1},-1)\right) = \Phi^{(upper)}_{i,j-1}(\chi_1,\ldots,\chi_{j-1})$
    \item $0.5\left(\Phi^{(lower)}_{i,j}(\chi_1,\ldots,\chi_{j-1}, 1) + \Phi^{(lower)}_{i,j}(\chi_1,\ldots,\chi_{j-1},-1)\right) = \Phi^{(lower)}_{i,j-1}(\chi_1,\ldots,\chi_{j-1})$
\end{itemize}

\paragraph{The first equality} We have
\[\Phi^{(avg)}_{i,j}(\chi_1,\ldots,\chi_{j-1}, 1) = \left( \left( \sum_{j' = 1}^{j-1}  a_{ij'} \chi_{j'}\right) + a_{ij} \right)^2 + \sum_{j' = j+1}^n a^2_{ij'} = \left(  \sum_{j' = 1}^{j-1} a_{ij'}  \chi_{j'}  \right)^2 + \sum_{j' = j}^n a^2_{ij'} + 2\left( \sum_{j' = 1}^{j-1}  a_{ij'} \chi_{j'} \right)a_{ij} \]
and 
\[\Phi^{(avg)}_{i,j}(\chi_1,\ldots,\chi_{j-1}, -1) = \left( \left( \sum_{j' = 1}^{j-1}  a_{ij'} \chi_{j'} \right) - a_{ij} \right)^2 + \sum_{j' = j+1}^n a^2_{ij'} = \left(  \sum_{j' = 1}^{j-1}  a_{ij'} \chi_{j'}  \right)^2 + \sum_{j' = j}^n a^2_{ij'} - 2\left( \sum_{j' = 1}^{j-1}  a_{ij'} \chi_{j'} \right)a_{ij}. \]
Therefore, we can conclude the first equality, as follows:
\[0.5 \left(\Phi^{(avg)}_{i,j}(\chi_1,\ldots,\chi_{j-1}, 1) + \Phi^{(avg)}_{i,j}(\chi_1,\ldots,\chi_{j-1}, -1) \right) = \left(\sum_{j' = 1}^{j-1}  a_{ij'} \chi_{j'}  \right)^2 + \sum_{j' = j}^n a^2_{ij'} = \Phi^{(avg)}_{i,j-1}(\chi_1,\chi_2,\ldots,\chi_{j-1}).\]

\paragraph{The second equality} We have 
\[\Phi^{(upper)}_{i,j}(\chi_1,\ldots,\chi_{j-1},1) = \left(1 + \lambda_ia_{ij} + \lambda_i^2a^2_{ij}\right) \cdot \prod_{j' = 1}^{j-1} \left(1 + \lambda_i a_{ij'} \chi_{j'} + (\lambda_i a_{ij'})^2 \right) \cdot \prod_{j' = j+1}^n (1 + (\lambda_i a_{ij'})^2)\]
and
\[\Phi^{(upper)}_{i,j}(\chi_1,\ldots,\chi_{j-1},-1) = \left(1 - \lambda_ia_{ij} + \lambda_i^2a^2_{ij}\right) \cdot \prod_{j' = 1}^{j-1} \left(1 + \lambda_i a_{ij'} \chi_{j'} + (\lambda_i a_{ij'})^2 \right) \cdot \prod_{j' = j+1}^n (1 + (\lambda_i a_{ij'})^2).\]
Therefore, we can conclude the second equality as follows
\begin{align*}
&0.5 \left(\Phi^{(upper)}_{i,j}(\chi_1,\ldots,\chi_{j-1},1) + \Phi^{(upper)}_{i,j}(\chi_1,\ldots,\chi_{j-1},-1) \right) \\
&=
0.5 \left(1 + \lambda_ia_{ij} + \lambda_i^2a^2_{ij} +  1 - \lambda_ia_{ij} + \lambda_i^2a^2_{ij}\right)  \cdot \prod_{j' = 1}^{j-1} \left(1 + \lambda_i a_{ij'} \chi_{j'} + (\lambda_i a_{ij'})^2 \right) \cdot \prod_{j' = j+1}^n (1 + (\lambda_i a_{ij'})^2) \\
&= \prod_{j' = 1}^{j-1} \left(1 + \lambda_i a_{ij'} \chi_{j'} + (\lambda_i a_{ij'})^2 \right) \cdot \prod_{j' = j}^n (1 + (\lambda_i a_{ij'})^2) \\
&= \Phi^{(upper)}_{i,j-1}(\chi_1,\ldots,\chi_{j-1}).
\end{align*}

\paragraph{The third equality} A calculation analogous to above shows the third equality 
\[0.5\left(\Phi^{(lower)}_{i,j}(\chi_1,\ldots,\chi_{j-1}, 1) + \Phi^{(lower)}_{i,j}(\chi_1,\ldots,\chi_{j-1},-1)\right) = \Phi^{(lower)}_{i,j-1}(\chi_1,\ldots,\chi_{j-1})\]
\end{proof}


\section{Example application: deterministic parallel edge coloring}
\label{app:edge-coloring}
\CrlEdgeColoring*

\begin{proof}[Proof Sketch of \Cref{crl:edge-coloring}]
The proof follows the general approach of Karloff and Shmoys~\cite{karloff1987edge-coloring}, as also used by Motwani, Naor, and Naor~\cite{motwani1989probabilistic}. 
The algorithm of Karloff and Shmoys~\cite{karloff1987edge-coloring} is randomized, but it uses randomness in only one part of partitioning the vertices into two parts, with certain degree-evenness properties, which can be seen as a set balancing problem. We can compute this part deterministically using our deterministic parallel set balancing algorithm, and thus the entire algorithm becomes deterministic. Below, we provide a brief overview of their approach.

The edge coloring is computed recursively. In the base case where $\Delta \leq \poly(\log n)$, then the result follows from a deterministic parallel algorithm of Karloff and Shmoys that computes a $\Delta+1$ edge coloring in $\poly(\Delta \log n)$ depth and using $m \cdot \poly(\Delta\log n)$ work.~\footnote{Their explicit statement for the the work bound is $n^{O(1)}$, but a closer inspection of their algorithm shows that it satisfies the work bound $m(\Delta\log n)^{O(1)}$, with only one change: for the maximal independent set subroutine, one should invoke the work-efficient deterministic algorithm of Luby~\cite{luby1988removing}, instead of his polynomial work variant~\cite{luby86}.} The more interesting regime is when $\Delta\geq \Omega(\log^2 n)$, where the problem is treated recursively. For that, we compute a partitioning of the vertices into two disjoint parts $V_0$ and $V_1$, such that for each node $v$, for each $i\in[0,1]$, we have $|N(v)\cap V_i| \leq \Delta/2+O(\sqrt{\Delta \log n})$. Here, $N(v)$ denotes the set of neighbors of $v$ in the graph. Notice that this is an instance of set balancing where the ground set is $V$ and for each node $v$, we have a set $N(v)$ in the set system of the set balancing problem. Thus, the desired partition follows from \Cref{thm:main-unweighted}. Let us see how this partition is used.

We perform the edge coloring in two different parts: 
\begin{itemize}
\item[(A)] Let $G[V_0, V_1]$ denote the bipartite induced subgraph, which has exactly the edges of $G$ with one endpoint in $V_0$ and the other endpoint in $V_1$. This graph has maximum degree $\Delta'=\Delta/2+O(\sqrt{\Delta \log n})$. We can color the edges of $G[V_0, V_1]$ optimally with $\Delta'$ colors, since it is a bipartite graph, using the parallel algorithm of Lev, Pippenger, and Valiant~\cite{lev1981fast} for edge-coloring bipartite graphs. This runs in $\poly(\log n)$ depth and using $\tilde{O}(|E(G[V_0, V_1])|)$ work. One could also apply here the algorithm of Alon~\cite{alon2003edge-coloring} (even though it is not explicitly claimed as a parallel algorithm).
\item[(B)] Separately, we use different colors from those used in (A) to color the edges of $G[V_0]$ and the edges in $G[V_1]$, by leveraging the degree in each of these two subgraphs is reduced to $\Delta/2+O(\sqrt{\Delta\log n}$. Importantly, we use the same set of colors for $G[V_0]$ and $G[V_1]$, and we solve these two problems in parallel and independently. Indeed, steps (A) and (B) can also be performed in parallel.
\end{itemize}
Let $D(\Delta)$ denote the depth of the algorithm for instances with degree $\Delta$. We get the recursion $D(\Delta) = D(\Delta/2+O(\sqrt{\Delta\log n})) + \poly(\log n)$, with the base case $D(\Delta)= \poly(\log n)$ when $\Delta \leq \poly(\log n)$. The solution is $D(\Delta)=\poly(\log n)$ for all $\Delta\in [n]$. Let us next examine the number of colors, using the notation $q(\Delta)$ for the number of colors needed to edge color a graph with maximum degree $\Delta$. We have $q(\Delta) = q(\Delta/2+O(\sqrt{\Delta\log n}) + \Delta/2+O(\sqrt{\Delta \log n})$, with the base case of $q(\Delta) = \Delta+1$ for $\Delta=\poly(\log n)$. The solution to this recursion is $q(\Delta)=\Delta+O(\sqrt{\Delta\log n})$. Finally, notice also that the algorithm has near-linear work $\tilde{O}(m)$, since each edge is processed once in the set balance computation, and then either in coloring the bipartite graph or in one of the two graphs $G[V_0]$ and $G[V_1]$ which are treated recursively.
\end{proof}

\section{Lattice Approximation}
\label{app:lattice}
We now restate and present a proof sketch of our lattice approximation result, following the reduction of Motwani, Naor, and Naor~\cite[Section 8]{motwani1994probabilisticJournal}. In their reduction, we use our own solution for the weighted set balancing problem (as stated in \Cref{thm:main-weighted}), instead of the solution that they build for this problem\footnote{They call (a limited variant of) the latter the \textit{vector balancing} problem and build a solution for it via reduction to the unweighted set balancing problem. However, for that, they make an additional assumption of $m\geq \poly(n)$. Removing this assumption from their approach appears to result in a suboptimal bound, so we do not follow that direction. }
\thmLattice*
\begin{proof}[Proof Sketch of \Cref{thm:main-lattice}]
Let us assume that $p_j\in (0,1)$ for all $j\in[n]$ as for any $j$ that has $p_j\in\{0,1\}$, we can directly set $q_j=p_j$. Moreover, without loss of generality, we can assume that each $a_{ij}$ and $p_j$ has at most $B=10\log n$ bits; the error introduced by less significant bits is negligible. 

We determine $\mathbf{q}$ by gradually rounding $\mathbf{p}$ in $B$ stages, where in each stage we remove one bit from $\mathbf{p}$. Set $\mathbf{p}^1=\mathbf{p}$. Let us describe stage $k\in[B]$: By induction, assume that each $p^{k}_j$ has $B-k+1$ bits. In stage $k$, we focus on the least significant bit of $p^{k}_j$. For those variables $j\in [n]$ for which the least significant bit of $p^{k}_{j}$ is zero, we do nothing and set $p^{k+1}_j=p^{k}_{j}$. For those $j$ that have their least significant bit of $p^{k}_j$ equal to one, we consider two options of rounding up---i.e., setting $p^{k+1}_j=p^{k}_{j}+2^{-(B-k+1)}$--- or rounding down---i.e., setting $p^{k+1}_j=p^{k}_{j}-2^{-(B-k+1)}$. If we continue this for $B$ stages, we would have in $\mathbf{q}=p^{B+1}_j\in \{0,1\}$. But how do we decide whether to round each variable up or down in each stage? 

In a randomized scheme, we would decide randomly. Let $J_k$ be the set of $j\in [n]$ such that $p^{k}_j$ has its least significant bit equal to $1$. For each such $j\in J_k$, we would do a rounding up as described above with probability $1/2$ and a rounding down with probability $1/2$. Overall, this randomized scheme would ensure that $\E[p^{B+1}_j]=p_j$. We do deterministic counterpart of this by using weighted set balancing, instead of randomly deciding to round up or down.

Concretely, we invoke the weighted set balancing result of \Cref{thm:main-weighted} on the set $J_k$, with the coefficient $a_{ij}2^{-(B-k+2)}$ for variable $j\in J_k$ in constraint $i\in [m]$. Define the error in the $i^{th}$ constraint in stage $k$ as $|\sum_{j=1}^{n} a_{ij} (p^{k+1}_j - p^{k}_j) |$. From \Cref{thm:main-weighted}, we get the following upper bound for this error: \[O(\sqrt{\sum_{j\in J_k} a_{ij}^2 2^{-2(B-k)} \log m})=O(\sqrt{\sum_{j\in J_k} a_{ij} p^{k}_j 2^{-(B-k)} \log m}) = O(\frac{\sqrt{\sum_{j=1}^{n} a_{ij}p^{k}_j \log m}}{2^{(B-k)/2}}).\] Here, the first inequality holds because for each $j\in J_k$, we have $p^k_j\geq 2^{-(B-k+1)}$. One can see that these errors essentially form a geometric series through the stages; intuitively this is because $\sum_{j=1}^{n} a_{ij}p^k_j$ remains roughly equal as $\sum_{j=1}^{n} a_{ij}p_j$, modulo a smaller additive term (except when the error is below $O(\log m)$, but those stages add up to $O(\log m)$ error overall). Thus, with some calculations, one can bound the total sum of these errors through the stages to be $O(\sqrt{\sum_{j=1}^{n}a_{ij}p_j\log m} + \log m)$, thus concluding the theorem. We do not repeat the calculations here; they can be found in ~\cite[Lemma 8.3]{motwani1994probabilisticJournal}.    
\end{proof}

\end{document}

%% file: macros.tex
\usepackage[IL2]{fontenc}
\usepackage[letterpaper,margin=1.00in]{geometry}
\usepackage{amsmath, amssymb, amsthm, amsfonts}
\usepackage{bbm}
\usepackage{amscd}
\usepackage{mathrsfs}

\usepackage{comment} 
\usepackage{ifthen}
\usepackage{tikz}
\usetikzlibrary{positioning,decorations.pathreplacing}
\usepackage{ bbold }
\usepackage{appendix}
\usepackage{graphicx}
\usepackage{color}
\usepackage{epstopdf}
\usepackage{wrapfig}
\usepackage{paralist}
\usepackage{wasysym}
\usepackage[textsize=tiny]{todonotes}

\usepackage{listings} 

\usepackage{pdflscape} 

\usepackage{framed}
\usepackage[framemethod=tikz]{mdframed}
\usepackage[bottom]{footmisc}
\usepackage{enumitem}
\setitemize{noitemsep,topsep=3pt,parsep=3pt,partopsep=3pt}
\usepackage[font=small]{caption}
\usepackage{xspace}

\usepackage{thmtools} 
\usepackage{thm-restate} 

\usepackage{hyperref}
\hypersetup{
    unicode=false,          
    colorlinks=true,        
    linkcolor=red,          
    citecolor=darkgreen,        
    filecolor=magenta,      
    urlcolor=cyan           
}

\newtheorem{theorem}{Theorem}[section]
\newtheorem{lemma}[theorem]{Lemma}
\newtheorem{meta-theorem}[theorem]{Meta-Theorem}
\newtheorem{claim}[theorem]{Claim}

\newtheorem{corollary}[theorem]{Corollary}

\newtheorem{definition}[theorem]{Definition}

\usepackage[capitalize, nameinlink,noabbrev]{cleveref}

\crefname{theorem}{Theorem}{Theorems}
\crefname{proposition}{Proposition}{Propositions}
\crefname{observation}{Observation}{Observations}
\crefname{lemma}{Lemma}{Lemmas}
\crefname{claim}{Claim}{Claims}
\crefname{problem}{Problem}{Problems}
\crefname{conjecture}{Conjecture}{Conjectures}
\crefname{question}{Question}{Questions}
\crefname{example}{Example}{Examples}
\crefname{fact}{Fact}{Facts}

\definecolor{darkgreen}{rgb}{0,0.5,0}

\usepackage{algcompatible}
\algnewcommand\algorithmicswitch{\textbf{switch}}
\algnewcommand\algorithmiccase{\textbf{case}}

\algdef{SE}[SWITCH]{Switch}{EndSwitch}[1]{\algorithmicswitch\ #1\ \algorithmicdo}{\algorithmicend\ \algorithmicswitch}%
\algdef{SE}[CASE]{Case}{EndCase}[1]{\algorithmiccase\ #1}{\algorithmicend\ \algorithmiccase}%
\algtext*{EndSwitch}%
\algtext*{EndCase}%

\newcommand{\disc}{\Delta}

\newcommand{\eps}{\varepsilon}

\renewcommand{\P}{\textrm{P}}

\newcommand{\poly}{\operatorname{poly}}

\renewcommand{\phi}{\varphi}

\newcommand{\E}{\mathbb{E}}
\renewcommand{\Pr}{\P}

\renewcommand{\paragraph}[1]{\vspace{0.15cm}\noindent {\bf #1}:}

\newcommand{\FullOrShort}{full}
\ifthenelse{\equal{\FullOrShort}{full}}{
  
  \newcommand{\fullOnly}[1]{#1}
  \newcommand{\shortOnly}[1]{}

  }{

    \newcommand{\fullOnly}[1]{}
    \newcommand{\IncludePictures}[1]{}
   
  }

\newcommand{\fail}{prob^{bad}_i}

\newcommand{\Deltav}{\mathbf{\Delta}}
\newcommand{\Deltai}{\Delta_i}

\newcommand{\Ibad}{I^{bad}}

\newcommand{\Pot}

\renewcommand{\Pr}{\P}

\usepackage{algorithm}
\usepackage[noend]{algpseudocode}

%% file: main.bbl
\newcommand{\etalchar}[1]{$^{#1}$}
\begin{thebibliography}{RGH{\etalchar{+}}22}

\bibitem[AB21]{anderson2021parallel}
Daniel Anderson and Guy~E Blelloch.
\newblock Parallel minimum cuts in ${O} (m log^{2} n)$ work and low depth.
\newblock In {\em ACM Symposium on Parallelism in Algorithms and Architectures (SPAA)}, 2021.

\bibitem[AHK12]{arora2012multiplicative}
Sanjeev Arora, Elad Hazan, and Satyen Kale.
\newblock The multiplicative weights update method: a meta-algorithm and applications.
\newblock {\em Theory of computing}, 8(1):121--164, 2012.

\bibitem[Alo03]{alon2003edge-coloring}
Noga Alon.
\newblock A simple algorithm for edge-coloring bipartite multigraphs.
\newblock {\em Information Processing Letters}, 85(6):301--302, 2003.

\bibitem[ASZ20]{andoni2020parallel}
Alexandr Andoni, Clifford Stein, and Peilin Zhong.
\newblock Parallel approximate undirected shortest paths via low hop emulators.
\newblock In {\em ACM SIGACT Symposium on Theory of Computing (STOC)}, pages 322--335, 2020.

\bibitem[Ban10]{bansal2010constructive}
Nikhil Bansal.
\newblock Constructive algorithms for discrepancy minimization.
\newblock In {\em IEEE Symposium on Foundations of Computer Science (FOCS)}, pages 3--10, 2010.

\bibitem[BGSS20]{blelloch2020parallelism}
Guy~E Blelloch, Yan Gu, Julian Shun, and Yihan Sun.
\newblock Parallelism in randomized incremental algorithms.
\newblock {\em Journal of the ACM (JACM)}, 67(5):1--27, 2020.

\bibitem[Ble96]{blelloch1996programming}
Guy~E Blelloch.
\newblock Programming parallel algorithms.
\newblock {\em Communications of the ACM}, 39(3):85--97, 1996.

\bibitem[BR89]{berger1989simulating}
Bonnie Berger and John Rompel.
\newblock Simulating $(\log^{c} n)$-wise independence in nc.
\newblock In {\em IEEE Symposium on Foundations of Computer Science (FOCS)}, pages 2--7, 1989.

\bibitem[Bre74]{brent1974parallel}
Richard~P Brent.
\newblock The parallel evaluation of general arithmetic expressions.
\newblock {\em Journal of the ACM (JACM)}, 21(2):201--206, 1974.

\bibitem[CFR20]{cao2020efficient}
Nairen Cao, Jeremy~T Fineman, and Katina Russell.
\newblock Efficient construction of directed hopsets and parallel approximate shortest paths.
\newblock In {\em ACM SIGACT Symposium on Theory of Computing (STOC)}, pages 336--349, 2020.

\bibitem[DBS21]{dhulipala2021theoretically}
Laxman Dhulipala, Guy~E Blelloch, and Julian Shun.
\newblock Theoretically efficient parallel graph algorithms can be fast and scalable.
\newblock {\em ACM Transactions on Parallel Computing (TOPC)}, 8(1):1--70, 2021.

\bibitem[Fin18]{fineman2018nearly}
Jeremy~T Fineman.
\newblock Nearly work-efficient parallel algorithm for digraph reachability.
\newblock In {\em ACM Symposium on Theory of Computing (STOC)}, pages 457--470, 2018.

\bibitem[GGR23]{GGR2023Chernoff}
Mohsen Ghaffari, Christoph Grunau, and V{\'a}clav Rozho{\v{n}}.
\newblock Work-efficient parallel derandomization {I}: Chernoff-like concentrations via pairwise independence.
\newblock In {\em IEEE Symposium on Foundations of Computer Science (FOCS)}, page to appear, 2023.

\bibitem[Har19]{harris2019deterministic}
David~G Harris.
\newblock Deterministic parallel algorithms for bilinear objective functions.
\newblock {\em Algorithmica}, 81:1288--1318, 2019.

\bibitem[J{\'a}J92]{jaja1992introduction}
Joseph J{\'a}J{\'a}.
\newblock An introduction to parallel algorithms.
\newblock {\em Reading, MA: Addison-Wesley}, 10:133889, 1992.

\bibitem[JLS19]{jambulapati2019parallel}
Arun Jambulapati, Yang~P Liu, and Aaron Sidford.
\newblock Parallel reachability in almost linear work and square root depth.
\newblock In {\em IEEE Symposium on Foundations of Computer Science (FOCS)}, pages 1664--1686. IEEE, 2019.

\bibitem[KK94]{karger1994NC}
DR~Karger and D~Koller.
\newblock (de) randomized construction of small sample spaces in {N}{C}.
\newblock In {\em IEEE Symposium on Foundations of Computer Science (FOCS)}, pages 252--263, 1994.

\bibitem[KS87]{karloff1987edge-coloring}
Howard~J Karloff and David~B Shmoys.
\newblock Efficient parallel algorithms for edge coloring problems.
\newblock {\em Journal of Algorithms}, 8(1):39--52, 1987.

\bibitem[Li20]{li2020faster}
Jason Li.
\newblock Faster parallel algorithm for approximate shortest path.
\newblock In {\em ACM SIGACT Symposium on Theory of Computing (STOC)}, pages 308--321, 2020.

\bibitem[LM12]{lovett2012constructive}
Shachar Lovett and Raghu Meka.
\newblock Constructive discrepancy minimization by walking on the edges.
\newblock In {\em IEEE Symposium on Foundations of Computer Science (FOCS)}, pages 61--67, 2012.

\bibitem[LPV81]{lev1981fast}
Gavriela~Freund Lev, Nicholas Pippenger, and Leslie~G Valiant.
\newblock A fast parallel algorithm for routing in permutation networks.
\newblock {\em IEEE transactions on Computers}, 100(2):93--100, 1981.

\bibitem[Lub86]{luby86}
Michael Luby.
\newblock A simple parallel algorithm for the maximal independent set problem.
\newblock {\em SIAM Journal on Computing (SICOMP)}, 15:1036--1053, 1986.

\bibitem[Lub88]{luby1988removing}
Michael Luby.
\newblock Removing randomness in parallel computation without a processor penalty.
\newblock In {\em IEEE Symposium on Foundations of Computer Science (FOCS)}, pages 162--173, 1988.

\bibitem[MNN89]{motwani1989probabilistic}
Rajeev Motwani, Joseph Naor, and Moni Naor.
\newblock The probabilistic method yields deterministic parallel algorithms.
\newblock In {\em IEEE Symposium on Foundations of Computer Science (FOCS)}, pages 8--13, 1989.

\bibitem[MNN94]{motwani1994probabilisticJournal}
Rajeev Motwani, Joseph~Seffi Naor, and Moni Naor.
\newblock The probabilistic method yields deterministic parallel algorithms.
\newblock {\em Journal of Computer and System Sciences}, 49(3):478--516, 1994.

\bibitem[MRS01]{mahajan2001solving}
Sanjeev Mahajan, Edgar~A Ramos, and KV~Subrahmanyam.
\newblock Solving some discrepancy problems in {N}{C}.
\newblock {\em Algorithmica}, 29(3):371--395, 2001.

\bibitem[Rag86]{raghavan1986probabilistic}
Prabhakar Raghavan.
\newblock Probabilistic construction of deterministic algorithms: Approximating packing integer programs.
\newblock In {\em IEEE Symposium on Foundations of Computer Science (FOCS)}, 1986.

\bibitem[REGH22]{rozhovn2022deterministic}
V{\'a}clav Rozho{\v{n}}, Michael Elkin, Christoph Grunau, and Bernhard Haeupler.
\newblock Deterministic low-diameter decompositions for weighted graphs and distributed and parallel applications.
\newblock In {\em IEEE Symposium on Foundations of Computer Science (FOCS)}, pages 1114--1121, 2022.

\bibitem[RGH{\etalchar{+}}22]{rozhovn2022undirected}
V{\'a}clav Rozho{\v{n}}, Christoph Grunau, Bernhard Haeupler, Goran Zuzic, and Jason Li.
\newblock Undirected (1+ $\eps$)-shortest paths via minor-aggregates: near-optimal deterministic parallel and distributed algorithms.
\newblock In {\em ACM Symposium on Theory of Computing (STOC)}, pages 478--487, 2022.

\bibitem[Spe77]{spencer1977balancing}
Joel Spencer.
\newblock Balancing games.
\newblock {\em Journal of Combinatorial Theory, Series B}, 23(1):68--74, 1977.

\bibitem[Spe85]{spencer1985six}
J.~Spencer.
\newblock Six standard deviations suffice.
\newblock {\em Trans.\ of the American Mathematical Society}, 289(2):679--706, 1985.

\end{thebibliography}
